\documentclass[10.5pt]{article}
\usepackage{setspace}
\usepackage[margin=0.75in]{geometry}

\usepackage[utf8]{inputenc}
\usepackage{hanging}	
\usepackage{soul}
\usepackage{xcolor}
\usepackage{amsmath}
\usepackage{amsthm}
\usepackage{amssymb}
\usepackage{amsfonts}
\usepackage{setspace}
\usepackage{multirow}
\usepackage{caption}
\DeclareCaptionLabelFormat{AppendixTables}{A.#2}
\usepackage{subcaption}
\usepackage{graphics}
\usepackage{lscape}
\usepackage{changepage}
\usepackage{graphicx}% http://ctan.org/pkg/graphicx
\usepackage{xcolor,soul}
\usepackage{booktabs}
\usepackage{tikz-cd}
\usepackage[makeroom]{cancel}
\usepackage{mathrsfs}
\usepackage{algorithm}
\usepackage{algpseudocode}
\usepackage{bbm}

\usepackage{natbib}
\usepackage{hyperref}
\usepackage{xr-hyper}
\usepackage{subfiles}

\tikzset{
  symbol/.style={
    draw=none,
    every to/.append style={
      edge node={node [sloped, allow upside down, auto=false]{$#1$}}}
  }
}

\theoremstyle{definition}
\newtheorem{theorem}{Theorem}[section]  \newtheorem{assumption}{Assumption}[section]  \newtheorem{lemma}{Lemma}[section] \newtheorem{proposition}{Proposition}[section] \newtheorem{definition}{Definition}[section]\newtheorem{corollary}[theorem]{Corollary}

\title{Double/Debiased Machine Learning for Functional-Form-Robust Spatial Autoregression}
\author{Jieun Lee\footnote{Department of Economics, Emory University. \textit{jieun.lee@emory.edu, jieunlee.sophia@gmail.com}}}

\begin{document}
\maketitle

\begin{abstract}

Spatial autoregressive inference is typically conditional on the spatial
weights matrix, $W$, even though the underlying interaction structure is often
unknown and empirical conclusions can be sensitive to its specification. This
paper develops double/debiased machine learning (DML) inference for
low-dimensional SAR parameters when the spatial interaction operator is
learned flexibly from potentially endogenous characteristics. Within a
maintained admissible support, interaction strength is generated by an unknown
function of geographic and socioeconomic characteristics, making inference
robust to functional-form specification of the weights within that support.
Endogeneity in the characteristics generating $W$ is addressed through a
nonlinear control function based on locally relevant first-stage residual
information. Because the learned operator enters both the spatial lag and
spatially transformed instruments, treating the estimated $W$ as known
generally leaves a first-order generated-$W$ effect. I construct an
operator-orthogonal SAR-IV/GMM score that removes this leading sensitivity and
combine it with buffered spatial cross-fitting that separates evaluation-score
footprints from nuisance-training observations. Under near-epoch dependence
on a spatially mixing innovation field and target-relevant nuisance-rate and
regularity conditions, the estimator is asymptotically linear and root-$n$
normal. Monte Carlo simulations show improved finite-sample inference relative
to nonorthogonal alternatives when the interaction function is misspecified,
weight-generating characteristics are endogenous, and observations are
spatially dependent. In a U.S.\ application, diabetes estimates vary with the choice of $W$, showing the sensitivity of SAR inference to the interaction structure. Even for the same learned $W$, results differ across inferential methods, highlighting the importance of inference when $W$ is learned.

\medskip

\noindent\textbf{Keywords:}
Spatial autoregression; spatial weights matrix; machine learning;
double machine learning; endogenous spatial weights; control function;
Neyman orthogonality; spatial cross-fitting; near-epoch dependence.

\medskip

\noindent\textbf{JEL Codes:}
C14, C21, C26, C45.

\end{abstract}

\section{Introduction}\label{sec:introduction}

Spatial autoregressive (SAR) models are widely used to examine aggregate
patterns of spatial dependence when outcomes in one location may depend on
outcomes elsewhere. Their empirical content depends critically on the spatial
weights matrix, which determines which locations interact and how strongly. In
most applications, however, the underlying interaction structure is not
directly observed. Researchers therefore construct the weights matrix using
geographic distance, contiguity, economic similarity, social characteristics,
or other measures of proximity.

The choice of spatial weights can materially affect empirical conclusions.
Inverse-distance rules, parametric decay functions, nearest-neighbor
structures, and contiguity matrices each impose a particular view of spatial
interaction, and alternative choices can produce different spatial lags and
parameter estimates from the same data
\citep{StakhovychBijmolt2009,HarrisMoffatKravtsova2011,Juhl2020,Lee2022}.
Moreover, economically relevant proximity need not coincide with physical
distance \citep{ConleyLigon2002}. Spatial interaction may depend jointly and
nonlinearly on geography, income, population, migration, infrastructure,
trade, and other characteristics.

This paper develops a robust estimation and inference framework for spatial
autoregressive models in which the spatial weights matrix is constructed from
observed pair characteristics, while the functional form mapping those
characteristics into relative spatial weights is treated as a nuisance.
The characteristics entering the weights matrix, together with a maintained
admissible support, determine the economic content of the interaction
structure. Conditional on these inputs, however, the mapping from pair
characteristics into interaction intensity---including the rate at which
interaction decays with distance or dissimilarity---is learned flexibly using
double/debiased machine learning (DML; e.g.,
\citet{ChernozhukovEtAl2018}). The robustness claim therefore concerns the
functional form of the weighting rule, rather than unrestricted recovery of
the connectivity structure or the choice of characteristics entering the
weights matrix.

Rather than treating the unknown weighting function as a structural object
that must be recovered jointly with the SAR parameters, I take the
low-dimensional SAR coefficients as the inferential target and treat the
functional mapping from the observed pair characteristics into relative
weights as a learned nuisance component. Because this nuisance function
determines the interaction operator entering the structural SAR-IV moments, I
construct estimation and inference that are locally robust to errors in its
estimation.

This distinction matters because estimation of the weighting function changes
the structural moment itself. The resulting operator determines the spatially
lagged outcome and may also enter spatially transformed instruments.
Consequently, even when the weighting function is consistently estimated,
plugging the resulting spatial weights matrix into a conventional SAR-IV or
GMM procedure and subsequently treating it as known can leave a first-order
generated-operator effect. Learning the weighting rule and conducting
structural inference therefore cannot generally be separated into two
independent steps.

I address this generated-operator problem using Neyman orthogonality. Starting
from the original SAR-IV moment conditions, I construct an
operator-orthogonal score using a Riesz-representer correction
\citep{ChernozhukovEtAl2018,ChernozhukovNeweySingh2022}. The spatially lagged
outcome remains a structural regressor and its coefficient remains a parameter
of interest; what is treated as a nuisance is the functional component that
generates the operator used to construct that regressor. The correction
removes the leading sensitivity of the structural moments to estimation of
the weighting function and the remaining nuisance components. Flexible
learning therefore reduces dependence on a predetermined parametric weighting
rule, while orthogonalization protects inference on the SAR parameters against
the first-order effect of estimating that rule. As a result, inference depends
on products of nuisance estimation errors rather than requiring every nuisance
component to be estimated at the root-$n$ rate.

This also changes what must be learned accurately. The objective is not global
recovery of the function mapping pair characteristics into spatial weights.
What matters for inference on the SAR parameters is estimation error in those
directions of the weighting function that affect the spatial lag and the
spatially transformed instruments entering the target moment. Directions that
leave these target-relevant objects unchanged, including directions eliminated
by row normalization, need not be recovered precisely. The framework
therefore preserves the economic information contained in the chosen
weight-generating characteristics while reducing the inferential importance
of accurately specifying every feature of the functional form that maps those
characteristics into relative interaction intensity.

A further complication arises when the characteristics generating spatial
interaction are themselves endogenous. Variables such as income, employment,
migration, population, or trade may shape spatial connections while also being
related to unobserved determinants of the outcome. Building on the
control-function approach to endogenous spatial weights
\citep{QuLee2015,QuLeeYang2021}, I allow both the weighting function and the
associated control adjustment to be learned flexibly. Because bilateral
characteristics can transmit endogeneity through more than the own-unit
residual, the control index may also include predetermined local residual
summaries. Conceptually, the two issues are distinct: the control-function
component addresses endogeneity in the inputs used to construct the weights,
whereas the orthogonalization component addresses estimation of the
functional form mapping those inputs into relative interaction intensity.
Importantly, addressing endogeneity in the weight-generating characteristics
does not by itself guarantee validity of the resulting SAR instruments, so
instrument validity is maintained as a separate identifying requirement.

Spatial dependence creates an additional challenge for orthogonal inference.
Standard cross-fitting separates nuisance-training and score-evaluation
observations, but different folds need not be independent in a spatial cross
section. Moreover, an evaluated SAR score may directly use neighboring
outcomes and covariates, while spatial feedback can transmit local shocks
through paths of arbitrary length. Random sample splitting therefore does not
provide the usual independence argument underlying cross-fitted DML.

I address this problem through buffered spatial cross-fitting. Evaluation
observations are organized into geographically coherent blocks. Nuisance
training first excludes the raw-data footprint required to construct the
corresponding evaluation scores and then places an additional guard region
between that footprint and the training sample. The asymptotic argument
controls the remaining training-to-evaluation dependence along the nuisance
direction actually generated by the fold-specific learner rather than through
a worst-case requirement over all possible training-sample directions.

The spatial dependence analysis builds on the random-field and
near-epoch-dependence framework of
\citet{JenishPrucha2009,JenishPrucha2012}. Rather than assuming that the
globally simultaneous SAR outcome is itself spatially mixing, I impose
weak-dependence conditions on underlying innovations and derive local
approximations for the outcome, the random interaction operator, and the
spatial objects entering the orthogonal score. Under suitable stability,
learning-rate, and spatial-decay conditions, the feasible estimator has the
same first-order behavior as an oracle estimator that knows the relevant
nuisance objects.

This paper makes three main contributions. First, it formulates the
functional-form problem for spatial weights as one of robust structural
estimation and inference with a learned nuisance weighting function and
constructs an operator-orthogonal SAR-IV/GMM score that removes the leading
effect of its estimation error. Second, it combines flexible learning of the
weighting rule with a control-function treatment of endogenous
weight-generating characteristics while keeping the economic inputs to the
weights matrix and instrument validity conceptually separate from the
functional-form problem. Third, it develops buffered spatial cross-fitting and
an accompanying oracle-reduction argument for a spatially dependent cross
section in which the observed outcome is globally simultaneous. Together,
these results provide a modular framework in which the weighting function and
other nuisance components may be estimated flexibly, subject to the stated
target-relevant rate and regularity conditions, while inference remains
focused on the low-dimensional SAR parameters.

The empirical application illustrates both motivations for the framework using county-level diabetes prevalence in the contiguous United States. Conventional fixed-$W$ specifications imply spatial autoregressive coefficients near $0.70$, compared with about $0.44$ under the learned-$W$ plug-in estimator and about $0.20$ under the operator-orthogonal estimator. Spatial dependence therefore remains positive and statistically significant, but its estimated magnitude declines substantially when the interaction operator is learned flexibly and the first-order effects of that learning are incorporated. The results further show that even modest changes in the weighting structure can lead to economically meaningful differences in the estimated strength of spatial dependence.

The rest of the paper is organized as follows. In Section
\ref{subsec:related_literature}, I discuss the related literature and how it
relates to the proposed framework. In Section \ref{sec:theory}, I develop the
SAR model with a learned weighting function, the control-function
representation, identification, instrument conditions, and buffered spatial
cross-fitting. In Section \ref{sec:asymptotics}, I construct the
operator-orthogonal score and establish the large-sample properties of the
estimator under spatial dependence. I then conduct Monte Carlo simulations to
examine functional-form misspecification, endogenous weight-generating
characteristics, estimation error from learning the spatial interaction
operator, and alternative cross-fitting procedures. In Section
\ref{sec:empirics}, I provide a U.S.\ county-level health application that
illustrates how conclusions about spatial dependence can vary with the
functional form used to construct the interaction operator and with whether
the estimation effect from a learned operator is properly incorporated into
structural inference. In Section \ref{sec:conclusion}, I conclude.

\subsection{Related Literature}\label{subsec:related_literature}

This paper connects four strands of the literature: classical estimation and
inference in spatial autoregressive models, estimation and identification of
unknown or endogenous weighting rules, semiparametric inference with estimated
nuisance functions, and limit theory under spatial or network dependence. The
central distinction of the present framework is inferential. The observed
characteristics entering the spatial weights matrix determine the economic
content of spatial proximity, while the functional mapping from those
characteristics into relative interaction intensity is treated as a flexibly
learned nuisance component. Because this nuisance function determines an
operator entering the structural SAR-IV moments through both the spatial lag
and, when used, spatially transformed instruments, I construct an orthogonal
score that removes its first-order estimation effect on inference for the
low-dimensional SAR parameters. At the same time, the framework permits the
characteristics used to generate the weights to be endogenous and addresses
that endogeneity through a flexible control function. The resulting analysis
therefore combines flexible learning of the interaction operator, correction
for endogenous weight-generating characteristics, and target-oriented
orthogonal inference. Because the observations form a spatially dependent
cross section, these ingredients are further combined with buffered spatial
cross-fitting and a dependence argument adapted to the globally simultaneous
SAR outcome.

\paragraph{Spatial autoregression with unknown and endogenous weighting
rules.}

The classical spatial-econometric literature typically conditions on a
specified spatial weights matrix. Foundational work develops likelihood,
instrumental-variable, and GMM methods for spatial autoregressive models under
this maintained interaction structure; see, among others,
\citet{KelejianPrucha1998,KelejianPrucha1999,Lee2004,LinLee2010}. In empirical
applications, the weights matrix is commonly constructed from geographic
contiguity, physical distance, nearest-neighbor relations, or economic
measures of proximity \citep{Anselin1988,LeSagePace2009,ConleyLigon2002}.
Because empirical conclusions can be sensitive to this choice
\citep{StakhovychBijmolt2009,HarrisMoffatKravtsova2011,Juhl2020}, a related
literature treats the interaction structure itself as an object to be selected
or estimated. For example, \citet{LamSouza2020} estimate a spatial weights
matrix using combinations of candidate matrices together with a potentially
sparse adjustment, while \citet{AhrensBhattacharjee2015} estimate spatial
weights under sparsity restrictions.

The closely related network literature emphasizes that identification of
social or strategic interactions can itself depend on the structure of the
interaction network. \citet{BramoulleDjebbariFortin2009} characterize
identification of endogenous and contextual effects when interactions occur
through an observed network, while \citet{dePaulaRasulSouza2025} study
recovery of an otherwise unobserved interaction network from panel variation.
These papers address a different identification problem from the one studied
here. I maintain an admissible support and observed pair characteristics that
give the weights matrix its economic content, rather than attempting to
recover an unrestricted network from outcome variation alone. They
nevertheless make clear why the interaction structure should be viewed as an
econometric object rather than as an innocuous normalization chosen by the
researcher.

Most directly related on the unknown-weight side are papers that allow the
spatial weighting rule itself to be estimated flexibly. \citet{Sun2016}
develops a functional-coefficient SAR model with nonparametric spatial weights
and uses sieve and nonparametric GMM methods to estimate the unknown weighting
function. More recently, \citet{GuptaQuZhang2026} develop a
semi-nonparametric framework for spatial dynamic panel models in which spatial
weights in several channels are unknown functions of underlying economic
distances and are estimated by sieve GMM together with the finite-dimensional
parameters. These papers establish that an unknown functional rule for spatial
interaction can be incorporated into structural spatial estimation and
provide the closest semiparametric benchmarks for flexible learning of the
weighting rule.

The present paper differs from this literature in both the source of
uncertainty and the inferential objective. First, the pair characteristics
entering the interaction rule may contain non-predetermined socioeconomic
characteristics, so the inputs used to construct the learned operator may be
endogenous. Second, the unknown weighting function is treated as a nuisance
for inference on the low-dimensional SAR parameters rather than as an object
that must be recovered globally with first-order accuracy. I derive the
sensitivity of the SAR-IV moment to perturbations of the learned interaction
function and construct an operator-orthogonal score that removes the leading
generated-$W$ effect. The relevant metric for learning the weighting function
is therefore target-specific: what matters is how estimation error propagates
into the spatial lag and spatially transformed instruments, rather than global
recovery of the primitive function itself. The analysis is also developed for
a dependent cross section and uses buffered spatial cross-fitting to separate
nuisance learning from evaluation in the presence of spatial dependence.

A separate but closely related literature addresses endogeneity of the spatial
weights matrix or of the variables used to construct it. \citet{QuLee2015}
and \citet{QuLeeYang2021} show that when weight-generating characteristics are
endogenous, their relationship with the structural disturbance must be
addressed explicitly. In particular, \citet{QuLeeYang2021} develop a
control-function approach for SAR models whose weights are constructed from
bilateral variables. \citet{LinSong2025} provide a complementary recent
approach: they develop an instrument-free semiparametric copula method for SAR
models with an endogenous spatial weights matrix, endogenous regressors, or
both, and estimate the structural and copula components jointly by sieve
maximum likelihood.

I retain the control-function insight of the endogenous-weight literature but
separate the endogeneity problem from the functional-form problem. The control
function addresses dependence between the structural disturbance and the
non-predetermined characteristics entering the weighting rule. Conditional on
those observed pair characteristics and the maintained admissible support, the
mapping from pair characteristics into relative interaction intensity remains
unknown and is learned flexibly. Operator orthogonalization then addresses the
first-order inferential effect of estimating this mapping. Thus, relative to
the copula route of \citet{LinSong2025}, the present framework retains an
IV/control-function structure and focuses on generated-operator uncertainty;
relative to \citet{QuLee2015,QuLeeYang2021}, it additionally treats the
mapping from observed pair characteristics into relative interaction intensity
as an unknown function rather than taking the weight-construction rule as a
maintained feature of the model. These components have different
roles: the control function addresses endogeneity in the inputs to the
weights matrix, while the operator-orthogonal score addresses estimation of
the function generating the weights. Instrument validity for the SAR equation
remains a separate identifying restriction.

\paragraph{Orthogonal inference and the joint sieve-GMM alternative.}

The semiparametric literature provides a natural alternative starting point.
More generally, first-step nonparametric estimation can affect the asymptotic
distribution of a finite-dimensional estimator through its first-order
influence on the estimating equation \citep{Newey1994}. \citet{AiChen2003}
show that finite-dimensional parameters and unknown functions satisfying
conditional moment restrictions can be estimated jointly by sieve minimum
distance while retaining root-$n$ asymptotic normality for the
finite-dimensional component. \citet{ChenPouzo2012,ChenPouzo2015} develop
estimation and inference for broad classes of semi- and nonparametric
conditional moment models, including settings involving regularization and
ill-posed inverse problems.

This literature raises a natural question in the present setting. Because the
structural and nuisance components can be estimated jointly, one could in
principle conduct inference on the SAR parameters directly from a joint
semiparametric estimator using sieve methods. Indeed, the flexible-weight
estimators of \citet{Sun2016} and \citet{GuptaQuZhang2026} demonstrate that
joint estimation of structural parameters and an unknown spatial weighting
rule is a viable strategy. Orthogonalization is therefore not necessary in
principle for root-$n$ inference, and the contribution of this paper is not a
claim to the contrary.

The motivation for orthogonalization is instead that it changes the
inferential burden associated with estimation of the unknown weighting
function and other nuisance components. Under direct joint sieve inference,
estimation error in these components generally enters the first-order behavior
of the target estimator and must be characterized jointly with the
finite-dimensional parameter. In the present setting this issue is especially
important because an error in the weighting function changes an operator that
enters both the endogenous spatial regressor and, when used, spatially
transformed instruments. The resulting generated-$W$ effect is therefore not
an ordinary scalar first-stage perturbation.

The orthogonal score removes the first-order sensitivity of the target moment
to these nuisance perturbations. After orthogonalization, nuisance estimation
affects the target through second-order rate products together with vanishing
Riesz-approximation, localization, and spatial-leakage terms. This permits
regularized and cross-fitted nuisance estimators to be chosen subject to
target-relevant rate conditions rather than requiring their first-order
estimation error to be carried directly into inference on the SAR
coefficients. The advantage of orthogonalization is therefore not that joint
sieve inference is incapable of delivering root-$n$ inference. Rather, it
reorganizes the problem around the directions of nuisance error that matter
for the structural target and makes the inferential procedure more modular
with respect to the choice of nuisance learner.

This distinction is especially useful for the function generating the spatial
interaction operator. Precise global recovery of that function is not
necessary for inference on the SAR parameters. What matters is estimation
error in the directions through which the weighting function affects the
spatial lag and the spatially transformed instruments entering the target
moment. The Riesz representation isolates this target-relevant sensitivity
and converts it into an orthogonal correction. This does not eliminate
regularization or inverse problems altogether, since estimation of the
relevant Riesz representer may itself require regularization. Rather, it
replaces inference based on unrestricted first-order propagation of the full
joint nuisance error with a target-specific orthogonalization problem.

The construction is closely related to locally robust GMM
\citep{ChernozhukovEscancianoIchimuraNeweyRobins2022}, double/debiased machine
learning \citep{ChernozhukovEtAl2018}, and regularized Riesz representations
\citep{ChernozhukovNeweySingh2022}. These approaches provide the general
semiparametric logic for reducing first-order sensitivity to estimated
nuisance functions. The spatial problem considered here requires deriving that
correction for a row-normalized interaction operator whose perturbation changes
both the endogenous spatial regressor and, when used, spatially transformed
instruments. Related work by \citet{ChernozhukovHuangWang2026} develops
debiased regularized inference for high-dimensional spatial panel networks.
Their focus is uniform inference in a high-dimensional, sparsely represented
network, whereas this paper targets low-dimensional SAR coefficients when
interaction intensity is generated by a smooth, potentially endogenous
pairwise function.

\paragraph{Orthogonal inference under spatial and network dependence.}

Cross-fitting creates an additional issue in a spatial cross section because
ordinary sample splitting does not generally separate statistically
independent observations. Extensions of DML to dependent sampling include
multiway clustering \citep{ChiangKatoMaSasaki2022}, dyadic dependence
\citep{ChiangMaRodrigueSasaki2026}, time-series dependence
\citep{CiganovicEtAl2026}, and locally dependent networks
\citep{EmmeneggerEtAl2025}. These approaches share the principle that the
training and evaluation samples must be separated in a manner compatible with
the relevant dependence structure.

The spatial and network asymptotic literature provides complementary tools.
The random-field and near-epoch-dependence results of
\citet{JenishPrucha2009,JenishPrucha2012} provide laws of large numbers and
central limit theory for spatial processes under weak dependence.
\citet{KojevnikovMarmerSong2021} develop limit theory and HAC inference for
network-dependent random variables, allowing dependence to decay with network
distance while accounting for network density. These results emphasize that
the relevant notion of separation depends on the structure through which
dependence propagates.

The SAR setting considered here requires an additional step because geographic
separation does not make the observed outcome field independent. The spatial
multiplier transmits innovations through paths of arbitrary length, and the
interaction operator governing these paths is itself random and estimated. I
therefore use buffered spatial cross-fitting. For each evaluation block, the
procedure first excludes the raw-data footprint required to construct its
score and then places an additional guard region between that footprint and
the sample used to estimate the nuisance functions. The resulting
oracle-reduction argument controls the remaining training-to-evaluation
dependence along the fold-specific nuisance direction rather than requiring
independence of the observed SAR outcomes across folds.

The asymptotic analysis consequently imposes weak-dependence conditions on the
underlying innovation field and obtains local approximations to the outcome,
the random interaction operator, and the spatial objects entering the
orthogonal score. Under sufficiently fast spatial decay, growing guard regions
make the remaining cross-fit leakage asymptotically negligible while retaining
a nondegenerate estimation sample. For inference under residual spatial
dependence, the framework can be combined with spatial or network HAC methods
\citep{KelejianPrucha2007,KimSun2011,KojevnikovMarmerSong2021}.

\section{Theoretical Framework}\label{sec:theory}

This section develops the econometric framework in a sequence of steps. I first introduce the basic cross-sectional SAR model with an unknown interaction structure. Adopting the endogeneity structure of \citet{QuLee2015}, I model endogeneity as arising from correlation between unobserved determinants of the outcome and unobserved determinants of the non-predetermined characteristics used to generate the spatial weights. I account for this relationship using a flexible control function. I then learn the interaction structure from observed pair characteristics without imposing a fixed distance-decay rule. Because the learned structure enters several parts of estimation, the procedure is designed so that small learning errors do not have a leading effect on the main parameters. Finally, I address spatial dependence through buffered spatial sample splitting and combine these components for estimation and spatially robust inference.

\paragraph{Notation.}

$A^{\prime}$ denotes the transpose of a matrix or vector $A$. Let $I_n$
denote the $n\times n$ identity matrix and $\mathbf 1_n$ the $n\times1$
vector of ones. For a vector $a$, $\|a\|_2$ denotes the Euclidean norm. For
a matrix $A$, $\|A\|_F$, $\|A\|_1$, and $\|A\|_{\infty}$ denote the
Frobenius, maximum absolute column-sum, and maximum absolute row-sum norms.
For an $n$-vector of random variables
$a=(a_1,\ldots,a_n)^{\prime}$, define
\[
\|a\|_{2,n}
=
\left[
\frac{1}{n}
E\left(\|a\|_2^2\right)
\right]^{1/2}.
\]
Expectation and probability are denoted by $E[\cdot]$ and $P(\cdot)$.
Convergence in probability and convergence in distribution are written as
$\xrightarrow{p}$ and $\xrightarrow{d}$. A subscript $0$ denotes the true
population value of a parameter or function, while a hat denotes its
estimator.

\subsection{Cross-Sectional SAR Model}\label{subsec:sar}

Suppose we observe one cross section of $n$ units indexed by $i=1,\ldots,n$. For each
unit, $Y_i$ is a scalar outcome and $X_i\in\mathbb R^p$ is a vector of
regressors. Define
\[
Y
=
\begin{pmatrix}
Y_1 & \cdots & Y_n
\end{pmatrix}^{\prime}
\in\mathbb R^n
\]
and
\[
X
=
\begin{pmatrix}
X_1^{\prime}\\
\vdots\\
X_n^{\prime}
\end{pmatrix}
\in\mathbb R^{n\times p}.
\]

Throughout, $X_i$ excludes an intercept. This normalization is useful because
the flexible control function introduced below contains an unrestricted level,
so a constant regressor would be annihilated by the conditional
residualization used for local identification. More generally, the local
separation condition below rules out target-regressor directions that are
indistinguishable from functions of the control index.

The baseline structural SAR model is
\[
Y
=
\rho_0W_0Y
+
X\beta_0
+
\varepsilon,
\]
where
\[
\rho_0\in\mathbb R,
\qquad
\beta_0\in\mathbb R^p,
\qquad
W_0\in\mathbb R^{n\times n},
\]
and
\[
\varepsilon
=
\begin{pmatrix}
\varepsilon_1 & \cdots & \varepsilon_n
\end{pmatrix}^{\prime}
\]
is the composite structural disturbance. At this stage, I do not require the
characteristics used to generate $W_0$ to be exogenous with respect to
$\varepsilon$. Subsections \ref{subsec:trigger} and \ref{subsec:control}
introduce these characteristics and decompose the composite disturbance using
a flexible control function.

The finite-dimensional parameter of interest is
\[
\theta_0
=
\begin{pmatrix}
\rho_0 &
\beta_0^{\prime}
\end{pmatrix}^{\prime}
\in\mathbb R^{d_\theta},
\qquad
d_\theta=p+1.
\]

I maintain the conventional normalization
\[
w_{ii,0}=0,
\qquad
w_{ij,0}\geq0,
\qquad
\sum_{j=1}^{n}w_{ij,0}=1.
\]
Thus $W_0$ is a row-stochastic spatial interaction operator.

\subsection{Endogenous Weight-Generating Characteristics}
\label{subsec:trigger}

Let
\[
Z_i\in\mathbb R^{d_Z}
\]
denote socioeconomic characteristics that determine spatial interaction.
These characteristics need not be exogenous with respect to the composite
structural disturbance $\varepsilon_i$. In particular, the component of
$Z_i$ not explained by predetermined information may be associated with
unobserved determinants of the outcome. Because the interaction
characteristics are bilateral, the relevant dependence may also involve
residual components of other units entering those interactions. Let
\[
V_i
=
\begin{pmatrix}
Q_i^{\prime} &
X_i^{\prime}
\end{pmatrix}^{\prime},
\]
where $Q_i$ contains excluded or predetermined first-stage shifters. I
specify
\[
Z_i
=
m_0(V_i)+U_i,
\qquad
E[U_i\mid V_i]=0.
\]

The function
\[
m_0:\mathcal V\rightarrow\mathbb R^{d_Z}
\]
is left unrestricted within a sufficiently regular function class and may be
estimated flexibly. The residual $U_i$ contains the component of the
weight-generating characteristics not explained by the first-stage
information and will enter the control function below.

\paragraph{Remark (Structural target versus first-stage nuisance).}

The parametric treatment of $X_i^{\prime}\beta_0$ in the outcome equation
and the flexible treatment of $m_0(V_i)$ in the equation for $Z_i$ serve
different purposes. The coefficients $\beta_0$, together with the spatial
autoregressive parameter $\rho_0$, are components of the finite-dimensional
structural parameter on which inference is conducted. I therefore maintain
the linear specification $X_i^{\prime}\beta_0$ as part of the structural
outcome equation. By contrast, $m_0(V_i)$ is an auxiliary first-stage object
whose role is to separate the component of $Z_i$ explained by predetermined
information from the residual variation used to construct the control
function. Its functional form is not itself an object of inference and is
therefore left flexible. This distinction is a modeling choice tied to the
target of inference rather than a requirement that the outcome and
weight-generating equations have the same degree of functional flexibility.
A specification that also treats the effect of $X_i$ on the outcome
nonparametrically would constitute a different semiparametric model with a
correspondingly different target parameter.

\subsection{Flexible Control Function}\label{subsec:control}

I use the first-stage residuals from the weight-generating characteristics to
control for their endogeneity. Because pairwise interaction characteristics
may depend on both $Z_i$ and $Z_j$, this endogeneity need not operate
exclusively through the own-unit residual $U_i$. Let $L_C<\infty$ denote the number of predetermined
local residual summaries. For $r=1,\ldots,L_C$, let
$\kappa_{ij,n}^{C,r}$ satisfy
\[
\kappa_{ii,n}^{C,r}=0,
\qquad
\sup_{i,r}
\sum_{j\neq i}
|\kappa_{ij,n}^{C,r}|
\leq C.
\]

For a candidate first-stage function $m$, define
\[
U_i(m)
=
Z_i-m(V_i)
\]
and
\[
\overline U_{i,r}(m)
=
\sum_{j\neq i}
\kappa_{ij,n}^{C,r}
U_j(m).
\]
The control index is
\[
C_i(m)
=
c
\left(
U_i(m),
\overline U_{i,1}(m),
\ldots,
\overline U_{i,L_C}(m)
\right).
\]
At the truth,
\[
C_{i0}=C_i(m_0).
\]

Let $\mathcal P_i$ collect the predetermined and excluded variables used to
construct the target and nuisance instruments for unit $i$, including the
relevant components of $X$, $Q$, and predetermined bilateral variables
$D_{ij}$ over the maintained score footprint. Define
\[
\mathcal A_i
=
\sigma(\mathcal P_i).
\]
This information set is fixed independently of the structural disturbance
and the realized target moment.

Let $h_0(\cdot)$ be an unknown control function and decompose the composite
structural disturbance from Subsection \ref{subsec:sar} as
\[
\varepsilon_i
=
h_0(C_{i0})
+
\xi_i,
\]
where $h_0(C_{i0})$ captures the component of the structural disturbance
associated with the endogenous variation in the weight-generating
characteristics $Z_i$, and $\xi_i$ is the remaining structural innovation.

\begin{assumption}[Control-function sufficiency]
\label{ass:control_function}

The remaining structural innovation satisfies
\[
E
\left[
\xi_i
\mid
C_{i0},
\mathcal A_i
\right]
=
0.
\]
Equivalently,
\[
E
\left[
\varepsilon_i
\mid
C_{i0},
\mathcal A_i
\right]
=
h_0(C_{i0}).
\]

\end{assumption}

The assumption says that the part of the original structural disturbance
associated with the endogenous weight-generating characteristics is captured
by the residual summaries collected in $C_{i0}$. Once those controls and the
predetermined information are held fixed, the remaining innovation has
conditional mean zero. The control index may contain both the own-unit first-stage residual and
predetermined summaries of nearby residuals, so endogeneity generated through
bilateral interaction characteristics need not be reduced to an own-unit
control function.

A stronger primitive condition that is sufficient for Assumption
\ref{ass:control_function} is
\[
\varepsilon_i
=
h_0
\left(
U_i,
\Lambda_{i1},
\ldots,
\Lambda_{iL_C}
\right)
+
\xi_i,
\]
where
\[
\Lambda_{ir}
=
\sum_{j\neq i}
\kappa_{ij,n}^{C,r}U_j
\]
and
\[
E
\left[
\xi_i
\mid
U_1,\ldots,U_n,\mathcal A_i
\right]
=
0.
\]
The stronger full-residual-field condition will also provide a convenient
primitive sufficient condition for validity of spatially transformed
instruments below.

The control-function condition is therefore a finite-index sufficiency
restriction on the residual field. The control weights are predetermined
because the restriction concerns which residual summaries are sufficient for
the endogeneity channel; they need not coincide with the structural
interaction weights. In empirical work, robustness can be assessed using
several predetermined geographic and socioeconomic residual summaries.

Substituting the control-function decomposition into the baseline SAR model
gives
\[
Y
=
\rho_0W_0Y
+
X\beta_0
+
h_0(C_0)
+
\xi,
\]
where
\[
h_0(C_0)
=
\begin{pmatrix}
h_0(C_{10}) & \cdots & h_0(C_{n0})
\end{pmatrix}^{\prime}.
\]
The next subsection represents the unknown operator as
$W_0=W_n(g_0)$.

\subsection{Functional Learning of the Spatial Interaction Operator}
\label{subsec:functional_W}

For each ordered pair $i\neq j$, define
\[
R_{ij}
=
r(Z_i,Z_j,D_{ij})
\in\mathbb R^{d_R},
\]
where $D_{ij}$ contains predetermined bilateral information such as
geographic distance, contiguity, transportation cost, or other economically
meaningful measures of separation. The map $r(\cdot)$ specifies the observed
pair characteristics supplied to the learner; the unknown object is the
function mapping these characteristics into relative interaction strength.

Let
\[
S_{ij,n}\in\{0,1\},
\qquad
S_{ii,n}=0,
\]
denote a predetermined candidate-support indicator. The maintained primitive
implementation uses a sufficiently broad but spatially local candidate
support. Specifically, there exists a deterministic sequence $a_{S,n}$ such
that
\[
S_{ij,n}=0
\qquad
\text{whenever }
d_{ij}^{*}>a_{S,n}.
\]
The benchmark theory allows $a_{S,n}$ to remain bounded or to increase slowly
with $n$, subject to the locality conditions below. The support restriction
determines which interactions are economically feasible but imposes no
parametric decay function within that candidate set.

This distinction is deliberate. The framework does not attempt to estimate an
unrestricted $n\times n$ matrix with no structure. Instead, it learns the
interaction strengths within a maintained admissible support while leaving
their functional dependence on observed pair characteristics flexible. Dense
or complete candidate networks may be considered only if spatial locality of
the induced operator is verified directly. The primitive B-spline and NED
results developed below are stated for spatially local candidate support.

Let
\[
g_0:\mathcal R\rightarrow\mathbb R
\]
be an unknown interaction-score function. Define
\begin{equation}
K_{ij}(g)
=
S_{ij,n}\exp\{g(R_{ij})\},
\qquad
i\neq j,
\label{eq:kernel}
\end{equation}
and set $K_{ii}(g)=0$. Every row is assumed to contain at least one
admissible neighbor.

The associated spatial weight is
\[
w_{ij}(g)
=
\frac{K_{ij}(g)}
{\sum_{\ell\neq i}K_{i\ell}(g)},
\qquad
i\neq j,
\]
with $w_{ii}(g)=0$. Denote
\[
W_n(g)
=
[w_{ij}(g)]_{i,j=1}^{n}.
\]
At the truth,
\[
W_0=W_n(g_0).
\]

The exponential transformation in \eqref{eq:kernel} guarantees positivity on
the candidate support but does not impose a negative-exponential distance
decay function. Conventional specifications such as
\[
-\alpha d_{ij},
\qquad
-\alpha\log d_{ij},
\qquad
\gamma_1d_{ij}+\gamma_2d_{ij}^2
\]
are nested as special cases.

Because the weights are normalized row by row, $g_0$ itself is not the
economic target. Define
\[
[g]_n
=
\left\{
\widetilde g:
W_n(\widetilde g)=W_n(g)
\right\}.
\]
The identified object is the induced interaction operator $W_n(g_0)$ rather
than a particular representative of $g_0$.

For a sieve representation
\[
g(r)=p_{J_n}(r)^{\prime}\gamma,
\]
define the population local null space by
\[
\mathcal N_{J_n,n}
=
\left\{
v\in\mathbb R^{J_n}:
\frac{1}{n}
E
\left\|
D_gW_n(g_{0,J_n})
[p_{J_n}(\cdot)^{\prime}v]
\right\|_F^2
=
0
\right\}.
\]
Because the integrand is nonnegative, this is equivalent to requiring the
corresponding first-order change in the random operator to equal zero almost
surely. This $L^2$ definition makes the population nature of the null space
explicit even though $W_n(g_{0,J_n})$ depends on the realized
weight-generating characteristics. A theoretical canonical representative may
be described by
\[
\gamma\perp\mathcal N_{J_n,n}.
\]
This population normalization is used only to characterize locally identified
directions. It is not required to be known in computation. In implementation,
whenever the sample criterion admits observationally equivalent sieve
coefficients, I select the minimum-Euclidean-norm element of the set of sample
minimizers. Thus feasibility does not require knowledge of the population
null space.

\begin{definition}[Functional-form robustness]
\label{def:functional_robustness}

Let $\mathcal G$ denote an admissible class of interaction functions.
Inference for $\theta_0$ is functional-form robust over $\mathcal G$ if its
asymptotic validity does not require $g_0$ to belong to a predetermined
finite-dimensional spatial-decay family and remains valid for every sequence
of data-generating processes with $g_0\in\mathcal G$ satisfying the stated
identification, spatial-locality, smoothness, and learning-rate conditions.

\end{definition}

Functional-form robustness means that the researcher need not decide in
advance that interaction strength must decline linearly, exponentially, or
according to another particular parametric distance function. The admissible
support, observed pair characteristics, normalization, and regularity
conditions remain maintained features of the model. The robustness claim is
therefore about the functional form generating relative weights within this
structured class, not about unrestricted estimation of every entry of $W_0$.

Definition \ref{def:functional_robustness} is consequently a
functional-specification statement. It does not claim uniform inference over
an unrestricted nonparametric universe. Corollary
\ref{cor:uniform_functional_robustness} below records the corresponding
sequence-wise implication over a common regularity class. In particular, when
$W_0$ is learned from the SAR equation, the maintained class excludes
sequences with $|\rho_0|\rightarrow0$.

\subsubsection{Derivative and target-relevant metric for $g$}

For an admissible perturbation $\delta g$, define
\[
D_gW_n(g)[\delta g]
=
\left.
\frac{\partial}{\partial t}
W_n(g+t\delta g)
\right|_{t=0}.
\]

\begin{proposition}[Derivative of the learned spatial weights]
\label{prop:w_derivative}

For a supported pair $i\neq j$,
\[
D_gw_{ij}(g)[\delta g]
=
w_{ij}(g)
\left[
\delta g(R_{ij})
-
\sum_{\ell\neq i}
w_{i\ell}(g)\delta g(R_{i\ell})
\right].
\]
For unsupported pairs,
\[
D_gw_{ij}(g)[\delta g]=0.
\]
Consequently,
\[
D_g\{W_n(g)Y\}_i[\delta g]
=
\sum_{j\neq i}
w_{ij}(g)
\left[
Y_j-\{W_n(g)Y\}_i
\right]
\delta g(R_{ij}).
\]

\end{proposition}
\noindent\textbf{Proof.} See Appendix \ref{app:proof_w_derivative}.

Changing $g$ at one pair does not change only that pair's weight because each
row of $W_n(g)$ must continue to sum to one. The derivative therefore has a
direct component and a row-normalization component. The final expression shows
exactly how an error in the learned interaction function changes the spatial
lag. This derivative is the source of the generated-$W$ effect corrected by
the orthogonal score below.

The convergence rate required for inference concerns those directions of $g$
that affect the target score rather than an arbitrary global norm for the
primitive function. Let
\[
\mathcal D_{W,i}[\delta g]
=
D_g\{W_n(g_0)Y\}_i[\delta g]
\]
and, for the target instrument vector introduced in Subsection
\ref{subsec:moments},
\[
\mathcal D_{H,i}[\delta g]
=
D_gH_i(g_0)[\delta g].
\]
Define
\[
\|\delta g\|_{\mathcal G,\mathrm{tar},n}
=
\left[
\frac{1}{n}
\sum_{i=1}^{n}
E
\left\{
|\mathcal D_{W,i}[\delta g]|^2
+
\|\mathcal D_{H,i}[\delta g]\|_2^2
\right\}
\right]^{1/2}.
\]
The quotient distance is
\[
\operatorname{dist}_{\mathcal G,\mathrm{tar}}
(g,[g_0]_n)
=
\inf_{\widetilde g\in[g_0]_n}
\|g-\widetilde g\|_{\mathcal G,\mathrm{tar},n}.
\]

The target-relevant metric is the norm that enters identification and the
first derivative of the SAR score. Second-order expansions require a slightly
stronger local envelope because the Hessian of the row-normalized weight map
contains products of perturbations. Let
\[
\|\delta g\|_{\mathcal G,+,n}
=
\|\delta g\|_{\mathcal G,\mathrm{tar},n}
+
\|\delta g\|_{\infty,\mathcal R_n},
\]
where $\|\cdot\|_{\infty,\mathcal R_n}$ is the essential supremum over the
maintained supported pair-characteristic domain, and define the corresponding
quotient distance
\[
\operatorname{dist}_{\mathcal G,+}
(g,[g_0]_n)
=
\inf_{\widetilde g\in[g_0]_n}
\|g-\widetilde g\|_{\mathcal G,+,n}.
\]
The stronger norm is used only to control products in second-order
remainders. Identification and first-order functional-form robustness remain
defined by the target-relevant action of $g$.

\begin{assumption}[Local smoothness of the operator map]
\label{ass:operator_map_smoothness}

There exists a neighborhood $\mathcal G_n^0$ of $[g_0]_n$ such that
$g\mapsto W_n(g)Y$ and $g\mapsto H_i(g)$ are twice Gateaux differentiable on
$\mathcal G_n^0$. Uniformly over $g\in\mathcal G_n^0$, their second
derivatives satisfy
\[
\left\|
D_g^2\{W_n(g)Y\}
[\delta g_1,\delta g_2]
\right\|_{2,n}
\leq
C\,\mathfrak b_n(\delta g_1,\delta g_2)
\]
and
\[
\left[
\frac{1}{n}
\sum_{i=1}^{n}
E
\left\|
D_g^2H_i(g)
[\delta g_1,\delta g_2]
\right\|_2^2
\right]^{1/2}
\leq
C\,\mathfrak b_n(\delta g_1,\delta g_2),
\]
where
\[
\mathfrak b_n(\delta g_1,\delta g_2)
=
\|\delta g_1\|_{\mathcal G,\mathrm{tar},n}
\|\delta g_2\|_{\mathcal G,+,n}
+
\|\delta g_2\|_{\mathcal G,\mathrm{tar},n}
\|\delta g_1\|_{\mathcal G,+,n}.
\]

\end{assumption}

The strengthened envelope is deliberate. The second derivative of a
row-normalized exponential weight is a centered bilinear form in
$\delta g_1(R_{ij})$ and $\delta g_2(R_{ij})$. An $L^2$ bound for its action
therefore generally requires an $L^\infty$ or comparable $L^4$ envelope on
one perturbation. Assumption \ref{ass:operator_map_smoothness} states this
requirement explicitly rather than treating an $L^2$ target norm as if it
were closed under multiplication. For spline learners the additional
supremum-norm control is supplied below.

\begin{proposition}[Target-relevant transfer from $g$ to $W$]
\label{prop:g_to_W_transfer}

Under Assumption \ref{ass:operator_map_smoothness}, uniformly for
$g\in\mathcal G_n^0$,
\[
\left\|
\{W_n(g)-W_n(g_0)\}Y
\right\|_{2,n}
\leq
C
\left[
\operatorname{dist}_{\mathcal G,\mathrm{tar}}
(g,[g_0]_n)
+
\operatorname{dist}_{\mathcal G,\mathrm{tar}}
(g,[g_0]_n)
\operatorname{dist}_{\mathcal G,+}
(g,[g_0]_n)
\right]
\]
and
\[
\left[
\frac{1}{n}
\sum_{i=1}^{n}
E
\|H_i(g)-H_i(g_0)\|_2^2
\right]^{1/2}
\leq
C
\left[
\operatorname{dist}_{\mathcal G,\mathrm{tar}}
(g,[g_0]_n)
+
\operatorname{dist}_{\mathcal G,\mathrm{tar}}
(g,[g_0]_n)
\operatorname{dist}_{\mathcal G,+}
(g,[g_0]_n)
\right].
\]
Consequently, if
\[
\operatorname{dist}_{\mathcal G,\mathrm{tar}}
(\widehat g,[g_0]_n)
=
O_p(r_{g,\mathrm{tar},n}),
\qquad
\operatorname{dist}_{\mathcal G,+}
(\widehat g,[g_0]_n)
=
O_p(r_{g,+,n}),
\]
with $r_{g,+,n}=o(1)$, then
\[
r_{W,n}
=
O_p(r_{g,\mathrm{tar},n}),
\qquad
r_{H,n}
=
O_p(r_{g,\mathrm{tar},n}).
\]

\end{proposition}

\noindent\textbf{Proof.}
See Appendix \ref{app:proof_g_to_W_transfer}.

The proposition separates the norm needed for first-order target relevance
from the stronger envelope needed for a valid quadratic expansion. The
additional envelope does not change the first-order rate of the generated
spatial lag or instruments: once
$\operatorname{dist}_{\mathcal G,+}(\widehat g,[g_0]_n)=o_p(1)$, the second
term is of smaller order than the target-relevant first-order error.

\subsection{Moment Construction and Nuisance System}
\label{subsec:moment_system}

For
\[
\theta
=
\begin{pmatrix}
\rho & \beta^{\prime}
\end{pmatrix}^{\prime}
\]
and candidate functions $(m,g,h)$, define the structural residual
\begin{equation}
\xi_i(\theta,m,g,h)
=
Y_i
-
\rho\{W_n(g)Y\}_i
-
X_i^{\prime}\beta
-
h(C_i(m)).
\label{eq:xi_candidate}
\end{equation}

The nuisance and target moment systems introduced below convert the maintained
conditional restrictions into unconditional moments used for sieve-GMM
learning and, subsequently, for construction of the operator-orthogonal score.

\subsubsection{Sieve representation of nuisances}
\label{subsubsec:sieve_representation}

The unknown nuisance functions are approximated by growing finite-dimensional
sieves. Let
\[
m(v)=p_{m,J_m}(v)^{\prime}\gamma_m,
\qquad
g(r)=p_{g,J_g}(r)^{\prime}\gamma_g,
\]
and
\[
h(c)=p_{h,J_h}(c)^{\prime}\gamma_h,
\qquad
\ell(c)=p_{\ell,J_\ell}(c)^{\prime}\gamma_\ell.
\]
Here, the vectors
$p_{m,J_m}$, $p_{g,J_g}$, $p_{h,J_h}$, and $p_{\ell,J_\ell}$
are sieve bases used to approximate the corresponding unknown functions.
For the vector-valued nuisances $m$ and $\ell$, the associated coefficient
arrays are understood to have the conformable dimensions and are vectorized
when stacked in $\vartheta$. The sieve bases are conceptually distinct from
the dictionaries $B_m$, $B_\ell$, $B_h$, and $B_g$ introduced below: the
sieve bases represent the nuisance functions themselves, whereas the
dictionaries provide the test functions or instruments used to construct the
moment conditions for estimating those functions. Thus the sieve bases
determine how the unknown functions are represented, whereas the dictionaries
determine which unconditional restrictions are used to estimate those
representations.

Collect the target and sieve coefficients in
\[
\vartheta
=
\left(
\theta^{\prime},
\gamma_m^{\prime},
\gamma_g^{\prime},
\gamma_h^{\prime},
\gamma_\ell^{\prime}
\right)^{\prime}.
\]

\subsubsection{Nuisance moments and conditional projection}
\label{subsubsec:nuisance_moments}

The nuisance restrictions are implemented through finite-dimensional,
possibly growing, vector-valued dictionaries
\[
B_m,\qquad B_\ell,\qquad B_h,\qquad B_g.
\]
These researcher-specified dictionaries serve as sieve instruments or test
functions for the corresponding nuisance residuals. Their role is to convert
the underlying conditional restrictions into unconditional moment conditions
that can be used for joint sieve-GMM estimation. They are not additional
structural parameters or nuisance functions. Their dimensions may increase
with the sample size subject to the sieve-complexity and rate conditions
imposed below.

The auxiliary projection nuisance is
\[
\ell_0(c)
=
E
\left[
B_g(\mathcal P_i)
\mid
C_{i0}=c
\right].
\]
It removes from the interaction dictionary the component explained solely by
the control index. Define
\[
\eta_0
=
(m_0,g_0,h_0,\ell_0),
\qquad
\eta
=
(m,g,h,\ell).
\]
For notational economy, write
\[
\xi_i(\theta,\eta)
=
\xi_i(\theta,m,g,h),
\]
noting that $\ell$ does not enter the structural residual directly. At the
truth, Assumption \ref{ass:control_function} gives
\begin{equation}
E
\left[
\xi_i(\theta_0,\eta_0)
\mid
C_{i0},
\mathcal A_i
\right]
=
0.
\label{eq:conditional_moment}
\end{equation}

For the first-stage conditional mean, define
\[
B_{m,i}
=
B_m(V_i).
\]
The dictionary $B_m(V_i)$ provides test functions for the restriction
\[
E[Z_i-m_0(V_i)\mid V_i]=0.
\]

For the control-function component, define
\[
B_{h,i}(\eta)
=
B_h(C_i(m)).
\]
The dictionary $B_h(C_i(m))$ provides test functions for the structural
residual along directions measurable with respect to the control index.

For the interaction component, define the baseline unresidualized
interaction dictionary
\[
B_{g,i}^{0}
=
B_g(\mathcal P_i),
\]
where $\mathcal P_i$ contains the predetermined and excluded information
defined above. I maintain that $B_g(\mathcal P_i)$ is
$\mathcal A_i$-measurable and does not depend on $g$. It may contain
functions of predetermined or excluded variables, including $X$, $Q$,
$D_{ij}$, the candidate support, and predetermined spatial summaries, but the
baseline nuisance dictionary does not contain objects whose randomness is
generated by $W_n(g)$ itself. Its purpose is to provide observable,
predetermined directions that are informative about changes in the learned
interaction operator. The precise richness requirement needed for
identification is stated below.

To remove the part of the interaction dictionary explained solely by the
control index, for a candidate $\ell$, let
\[
B_{\ell,i}(\eta)
=
B_\ell(C_i(m)),
\]
where $B_\ell$ is a dictionary used to estimate this conditional projection,
and define the residualized interaction dictionary
\[
\widetilde B_{g,i}(\eta)
=
B_{g,i}^{0}
-
\ell(C_i(m)).
\]
At the truth,
\[
\widetilde B_{g,i}(\eta_0)
=
B_g(\mathcal P_i)
-
E
\left[
B_g(\mathcal P_i)
\mid
C_{i0}
\right],
\]
so that
\[
E
\left[
\widetilde B_{g,i}(\eta_0)
\mid
C_{i0}
\right]
=
0.
\]
Moreover, $\widetilde B_{g,i}(\eta_0)$ is measurable with respect to
$\sigma(C_{i0},\mathcal A_i)$. This measurability is what allows the
control-function restriction to justify the interaction moment at the truth.

The restriction that $B_g$ be predetermined and independent of $g$ applies
to the nuisance interaction dictionary used in the baseline theory. The
target instrument vector $H_i(g)$ may still contain spatially transformed
objects under Assumption \ref{ass:instrument_validity}. A $g$-dependent
nuisance interaction dictionary can also be considered under a stronger
full-residual-field validity condition together with explicit validity
conditions for both the level moment and the derivative-of-dictionary term,
but that extension is not needed for the results below.

Define the stacked nuisance moment vector
\[
s_i(\theta,\eta)
=
\begin{pmatrix}
s_{m,i}(m)\\
s_{\ell,i}(\eta)\\
s_{h,i}(\theta,\eta)\\
s_{g,i}(\theta,\eta)
\end{pmatrix},
\]
where
\[
s_{m,i}(m)
=
\operatorname{vec}
\left[
B_{m,i}
\{Z_i-m(V_i)\}^{\prime}
\right],
\]
\[
s_{\ell,i}(\eta)
=
\operatorname{vec}
\left[
B_{\ell,i}(\eta)
\{
B_{g,i}^{0}
-
\ell(C_i(m))
\}^{\prime}
\right],
\]
\[
s_{h,i}(\theta,\eta)
=
B_{h,i}(\eta)
\xi_i(\theta,\eta),
\]
and
\[
s_{g,i}(\theta,\eta)
=
\widetilde B_{g,i}(\eta)
\xi_i(\theta,\eta).
\]
At the truth, the $m$ block has mean zero by the first-stage
conditional-mean restriction, and the $\ell$ block has mean zero by the
definition of $\ell_0$. Because $B_h(C_{i0})$ is measurable with respect to
$C_{i0}$ and $\widetilde B_{g,i}(\eta_0)$ is measurable with respect to
$\sigma(C_{i0},\mathcal A_i)$, Assumption
\ref{ass:control_function} also gives
\[
E
\left[
B_h(C_{i0})\xi_i
\right]
=
0,
\qquad
E
\left[
\widetilde B_{g,i}(\eta_0)\xi_i
\right]
=
0.
\]
Hence
\[
E[s_i(\theta_0,\eta_0)]
=
0.
\]

The four blocks have distinct identifying roles. The $m$ block identifies
the first-stage conditional mean. The $\ell$ block identifies the conditional
projection used to remove control-index variation from the interaction
dictionary. The $h$ block identifies the control-function component of the
structural residual. Finally, the $g$ block uses the residualized interaction
dictionary to identify target-relevant directions of the spatial interaction
operator. The interaction-moment richness condition below formalizes the
requirement that the span of this residualized dictionary be sufficiently
rich to detect every sieve direction of $g$ that matters for the SAR target.

Stacking these moment restrictions provides a common nuisance-moment system
whose derivative can subsequently be used to construct the Riesz
representation and the operator-orthogonal correction for estimation of
$m_0$, $g_0$, $h_0$, and $\ell_0$.

\subsubsection{SAR instruments and target moments}\label{subsec:moments}

Even after controlling endogeneity of the characteristics generating $W$, the
spatial lag $W_0Y$ remains endogenous because of simultaneous determination.
The control-function restriction and the validity of the SAR instruments are
therefore conceptually distinct requirements.

Let
\[
H_i(g)\in\mathbb R^q,
\qquad
q\geq d_\theta,
\]
denote a candidate vector of spatial instruments; the target-moment dimension
$q$ is fixed as $n\to\infty$. A representative dictionary may be generated from
\[
X,
\qquad
Q,
\qquad
W_n(g)X,
\qquad
W_n(g)Q,
\qquad
W_n(g)^2X.
\]
Throughout the baseline construction, any dependence of $H_i(g)$ on $g$ is
through the induced operator $W_n(g)$. Hence observationally equivalent
representatives of $g$ that generate the same normalized weight matrix also
generate the same target instruments.
Because $W_n(g_0)$ depends on the potentially endogenous characteristics
$Z_i$, validity of spatially transformed candidates does not follow
automatically from Assumption \ref{ass:control_function}. Candidate
instruments are retained only when the following moment restriction is
satisfied.

\begin{assumption}[SAR instrument validity]
\label{ass:instrument_validity}

The instrument vector used for estimation satisfies
\[
E
\left[
H_i(g_0)\xi_i
\right]
=
0.
\]
The corresponding population target moment has finite second moments.

A primitive sufficient condition for the first restriction is
\[
E
\left[
\xi_i
\mid
U_1,\ldots,U_n,\mathcal A_i
\right]
=
0
\]
together with measurability of $H_i(g_0)$ with respect to
\[
\sigma
\left(
U_1,\ldots,U_n,\mathcal A_i
\right).
\]

\end{assumption}

Controlling for $C_{i0}$ corrects endogeneity of the variables used to form
the interaction weights, but it does not by itself make every
$W_0X$- or $W_0Q$-type variable a valid instrument. Since the learned weights
depend on the residual field through $Z$, spatially transformed instruments
may inherit that dependence. Assumption \ref{ass:instrument_validity}
therefore states IV validity separately and explicitly.

The stronger full-residual-field condition is one transparent way to justify
such instruments: conditional on the residual field and predetermined
information, the structural innovation must have zero mean, and the proposed
instrument must be measurable with respect to that information. The
high-level theory does not require this particular sufficient condition if
instrument validity can be justified by another economically appropriate
restriction.

Define
\[
\phi_i(\theta,\eta)
=
H_i(g)\xi_i(\theta,\eta).
\]
Under Assumption \ref{ass:instrument_validity},
\[
E[\phi_i(\theta_0,\eta_0)]=0.
\]

The base target moment is generally not Neyman orthogonal with respect to
$\eta_0$. In particular, estimation error in $g_0$ changes both the spatial
lag entering the residual and the spatially transformed components of the
instrument vector. The operator-orthogonal score introduced in Section
\ref{sec:asymptotics} removes these first-order nuisance effects.

\subsection{Identification}
\label{subsec:g_identification}

The preceding moment construction separates the conditional restrictions used
for identification from the finite collection of unconditional moments used
by the estimator. I first state a primitive local separation condition for the
conditional model and then connect it to the implemented interaction moments.

\subsubsection{Primitive local separation}

For any square-integrable random object $A_i^{*}$, define
\[
\mathcal R_CA_i^{*}
=
E[A_i^{*}\mid C_{i0},\mathcal A_i]
-
E[A_i^{*}\mid C_{i0}].
\]

For
\[
\delta\theta
=
\begin{pmatrix}
\delta\rho &
\delta\beta^{\prime}
\end{pmatrix}^{\prime}
\]
and an admissible $\delta g$, define
\[
A_{\delta g}
=
D_gW_n(g_0)[\delta g]
\]
and
\[
\Delta_i(\delta\theta,\delta g)
=
\delta\rho(W_0Y)_i
+
X_i^{\prime}\delta\beta
+
\rho_0(A_{\delta g}Y)_i.
\]

\begin{assumption}[Local target--operator separation]
\label{ass:local_separation}

The first-stage conditional mean uniquely identifies $m_0$. Whenever
identification of the interaction operator is required, there exists
$\underline\rho>0$ such that
\[
|\rho_0|
\geq
\underline\rho.
\]
There exists $\kappa_{\mathrm{id}}>0$ such that, for every locally admissible
$(\delta\theta,\delta g)$,
\[
\left[
\frac{1}{n}
\sum_{i=1}^{n}
\|
\mathcal R_C
\Delta_i(\delta\theta,\delta g)
\|_{L^2}^2
\right]^{1/2}
\geq
\kappa_{\mathrm{id}}
\left(
\|\delta\theta\|_2
+
\|\delta g\|_{\mathcal G,\mathrm{tar},n}
\right).
\]

\end{assumption}

After removing variation explained only by the control function, no nonzero
local change in the SAR coefficient, regression coefficients, or
target-relevant interaction operator can leave the conditional mean
unchanged. This is the local rank condition that separates the finite-
dimensional target from changes in the learned operator.

The condition also explains why a constant is excluded from $X_i$: a constant
is annihilated by $\mathcal R_C$ and therefore could not satisfy this lower
bound. More generally, regressors whose relevant variation is completely
absorbed by the control index are not separately identified as components of
$\beta_0$.

The lower bound $|\rho_0|\geq\underline\rho$ is imposed only when the
interaction operator itself is to be identified uniformly over the maintained
sequence of data-generating processes. If $\rho_0=0$, the interaction
operator is not identified through the SAR equation because $g_0$ drops out
of the structural outcome equation. The benchmark theory therefore concerns
the nondegenerate spatial-interaction regime
\[
|\rho_0|
\geq
\underline\rho
>
0.
\]
In particular, the uniform functional-form robustness claim below does not
cover sequences with $\rho_0\rightarrow0$. Testing the canonical SAR null
$\rho_0=0$ when the interaction operator is itself unknown is a
nonregular identification problem and requires a separate weak- or
non-identification analysis; such inference is not claimed here.

\begin{assumption}[Operator richness]
\label{ass:operator_richness}

There exists $\kappa_W>0$ such that, for every locally admissible
operator-changing $\delta g$,
\[
\frac{1}{n}
E
\|A_{\delta g}Y\|_2^2
\geq
\kappa_W
\frac{1}{n}
E
\|A_{\delta g}\|_F^2.
\]

\end{assumption}

The assumption rules out changes in the weight matrix that are large as
operators but happen to be invisible when applied to the realized outcome
process. If a perturbation genuinely changes the relevant entries of $W_0$,
it must also change $W_0Y$ enough to be statistically detectable.

\begin{proposition}[Target and first-order operator identification]
\label{prop:local_identification}

Suppose \eqref{eq:conditional_moment} is continuously Gateaux differentiable
near the truth and Assumption \ref{ass:local_separation} holds. Any locally
observationally equivalent differentiable path satisfies
\[
\delta\theta=0,
\qquad
A_{\delta g}Y=0
\quad\text{in }L^2,
\qquad
\delta h(C_{i0})=0
\quad\text{a.s.}
\]
If Assumption \ref{ass:operator_richness} also holds, then
\[
A_{\delta g}=0
\quad\text{in }L^2.
\]
Thus the conditional moment locally identifies $\theta_0$ and the
target-relevant first-order action of the interaction operator, while $g_0$
remains identified only modulo directions that leave $W_n(g_0)$ unchanged.

\end{proposition}

\noindent\textbf{Proof.}
See Appendix \ref{app:proof_identification}.

The proposition separates identification of the economically relevant
operator from identification of a particular numerical representation of
$g_0$. Two interaction functions that generate the same normalized weight
matrix are observationally equivalent and need not be distinguished.
Subject to that unavoidable normalization, however, neither the target
parameter nor a first-order change in the relevant interaction operator can
be varied without changing the maintained conditional moment.

\subsubsection{Identification through the implemented interaction moments}

The preceding identification argument uses the conditional moment directly.
The estimator instead works with a finite collection of unconditional sieve
moments. I therefore impose a richness condition ensuring that these
implemented moments recover the target-relevant conditional variation needed
to identify changes in the interaction operator.

Let $\mathcal G_{J_n}^{\perp}$ denote the normalized population sieve tangent
space. For a unit target-relevant direction
\[
\delta g_J\in\mathcal G_{J_n}^{\perp},
\qquad
\|\delta g_J\|_{\mathcal G,\mathrm{tar},n}=1,
\]
write
\[
R_{i,\delta g_J}
=
\mathcal R_C
\{
(A_{\delta g_J}Y)_i
\}.
\]

\begin{assumption}[Interaction-moment richness]
\label{ass:sieve_moment_richness}

For every unit target-relevant sieve direction $\delta g_J$, there exists
$a_J(\delta g_J)$ satisfying
\[
\|a_J(\delta g_J)\|_2\leq C
\]
such that
\[
\left[
\frac{1}{n}
\sum_{i=1}^{n}
E
\left[
\left\{
a_J(\delta g_J)^{\prime}
\widetilde B_{g,i}(\eta_0)
-
R_{i,\delta g_J}
\right\}^2
\right]
\right]^{1/2}
\leq
\zeta_{J_n},
\qquad
\zeta_{J_n}\rightarrow0.
\]

\end{assumption}

The residualized interaction dictionary must be rich enough to approximate
the conditional variation generated by every target-relevant change in the
interaction operator. This is a sieve completeness or relevance condition:
if a change in $W_0$ matters for the SAR equation, the interaction moments
must contain enough variation to detect it.

The bounded coefficient requirement prevents detection from relying on
increasingly unstable linear combinations of the dictionary. Accordingly,
this assumption deliberately supplies a well-posed benchmark for direct
identification of the target-relevant $g$ directions used in Proposition
\ref{prop:sieve_g_rate}.

This benchmark is intentionally stronger than what may hold in applications
with weakly informative interaction characteristics. If the lower separation
constant or the effective $g$-block singular value is allowed to approach
zero, the interaction learner becomes weakly identified and its rate is
amplified accordingly. The main root-$n$ result below is therefore a
strong-identification result. Weak-identification-robust inference for
$\rho_0$ is a distinct extension and is not claimed here.

\begin{proposition}[Identification by the implemented sieve moments]
\label{prop:sieve_moment_identification}

Suppose Assumptions \ref{ass:local_separation} and
\ref{ass:sieve_moment_richness} hold. Then, for sufficiently large $J_n$,
there exists $c_g>0$ such that
\[
\left\|
D_g
\left\{
\frac{1}{n}
\sum_{i=1}^{n}
E[s_{g,i}(\theta_0,\eta_0)]
\right\}
[\delta g_J]
\right\|_2
\geq
c_g
\|\delta g_J\|_{\mathcal G,\mathrm{tar},n}
\]
for every normalized target-relevant sieve direction.

Consequently, if the GMM weighting matrix for the $g$ block has eigenvalues
bounded away from zero and the local second derivative is regular, the
population sieve-GMM criterion is locally quadratically identified in the
target-relevant directions.

\end{proposition}
\noindent\textbf{Proof.}
See Appendix \ref{app:proof_sieve_moment_identification}.

The proposition connects the abstract conditional identification condition to
the actual moments used by the estimator. The implemented $g$ moments are not
merely valid at the truth; their derivative is bounded away from zero in every
operator direction that matters for inference. This produces the local
curvature required for stable sieve estimation of the interaction function.

\subsection{Feasible Sieve-GMM Nuisance Learning}
\label{subsec:sieve_learning}

Using the sieve representations introduced in Subsubsection
\ref{subsubsec:sieve_representation} and the nuisance and target moment systems
developed in Subsubsections \ref{subsubsec:nuisance_moments} and
\ref{subsec:moments}, I now turn to feasible estimation of the nuisance
components
\[
\eta_0=(m_0,g_0,h_0,\ell_0).
\]
The resulting sieve-GMM construction serves two purposes: it provides feasible
estimators of the nuisance functions and supplies the finite-dimensional
nuisance derivative system used later to construct the Riesz correction.

\subsubsection{Joint sieve-GMM start}
\label{subsubsec:g_learner}

The local rate analysis developed below requires the nuisance estimators to
enter a neighborhood of the population solution. Rather than assuming an
infeasible preliminary estimator that already knows $W_0$, I construct a
feasible joint sieve-GMM start on the fold-specific auxiliary sample using
the sieve representations introduced above.

Let
\[
\mathfrak m_i(\vartheta)
=
\begin{pmatrix}
\phi_i(\theta,\eta)\\
s_i(\theta,\eta)
\end{pmatrix}
\]
denote the stacked target and nuisance moments defined in Subsubsections
\ref{subsubsec:nuisance_moments} and \ref{subsec:moments}. For fold $k$,
let $\mathcal T_k$ denote the auxiliary sample defined in Section
\ref{subsec:crossfitting}, and define
\[
\widehat{\mathfrak m}_k(\vartheta)
=
\frac{1}{|\mathcal T_k|}
\sum_{i\in\mathcal T_k}
\mathfrak m_i(\vartheta).
\]

For each $a\in\{m,g,h,\ell\}$, let $D_a^{(2)}$ denote the stacked
second-order difference operator for the corresponding B-spline coefficient
vector, applied along each spline dimension when a tensor-product basis is
used. Define the quadratic roughness penalty
\[
\mathcal P_a(\gamma_a)
=
\left\|
D_a^{(2)}\gamma_a
\right\|_2^2.
\]
The tuning parameter $\lambda_{a,n}\geq 0$ controls the degree of
regularization of nuisance component $a$.

Let $\widehat{\mathcal W}_k$ denote a positive-definite GMM weighting matrix
for the stacked moments, constructed using only the auxiliary sample
$\mathcal T_k$. Assume that, uniformly over the fixed number of folds,
\[
\left\|
\widehat{\mathcal W}_k-\mathcal W_0
\right\|_{\mathrm{op}}
=
o_p(1),
\]
where $\mathcal W_0$ is positive definite and has eigenvalues bounded away
from zero and infinity. Since this step is used only to obtain a consistent
preliminary estimator, the identity weighting matrix
$\widehat{\mathcal W}_k=I$ is admissible.

The fold-specific preliminary criterion is
\begin{equation}
\widehat Q_k(\vartheta)
=
\widehat{\mathfrak m}_k(\vartheta)^{\prime}
\widehat{\mathcal W}_k
\widehat{\mathfrak m}_k(\vartheta)
+
\sum_{a\in\{m,g,h,\ell\}}
\lambda_{a,n}\mathcal P_a(\gamma_a).
\label{eq:joint_preliminary}
\end{equation}

Let
\[
\mathcal M_k
=
\arg\min_{\vartheta\in\mathcal V_{J_n}}
\widehat Q_k(\vartheta)
\]
denote the set of fold-specific sieve-GMM minimizers. Because row
normalization may generate observationally equivalent representatives of
$g$, define $\widetilde\vartheta^{(-k)}$ to be an element of
$\mathcal M_k$ whose $g$-coefficient has minimum Euclidean norm among the
equivalent minimizers.

The criterion is generally nonconvex because the interaction score enters a
row-normalized softmax and the SAR moment system jointly depends on the target
and nuisance components. In implementation I therefore use block profiling:
initialize $(m,h,\ell)$ from their separate sieve moments, update $g$ from
the profiled interaction criterion, update the low-dimensional SAR parameter
from the resulting IV/GMM moments, and iterate these blocks until the
criterion and parameter vector stabilize. Multiple starting values are used
for the $g$ block, and the minimum-criterion solution is retained. The theory
requires the final numerical optimization error to be asymptotically
negligible relative to the statistical error.

Let $Q_{0,n}(\vartheta)$ denote the corresponding unpenalized population
stacked-GMM criterion, and define the population sieve minimizer
\[
\vartheta_{0,J_n}
\in
\arg\min_{\vartheta\in\mathcal V_{J_n}}
Q_{0,n}(\vartheta),
\]
using the same representative normalization for the $g$ component.

\paragraph{Local learning metric and effective complexity.}

For the joint rate statement, let $d_{\vartheta,n}$ be a normalized local
sieve metric that controls
\[
\|\theta-\theta_{0,J_n}\|_2,\qquad
\|m-m_{0,J_n}\|_{L^2},\qquad
\operatorname{dist}_{\mathcal G,+}(g,[g_{0,J_n}]_n),
\]
together with the corresponding $L^2$ errors of
$h(C_i(m))$ and $\ell(C_i(m))$. Let
$d_{\mathfrak m,n}$ denote the dimension of the stacked moment vector
$\mathfrak m_i(\vartheta)$ and let $J_{\mathrm{joint},n}$ denote the
dimension of the normalized joint sieve coefficient vector. Rather than
suppressing the growing moment dimension, define an effective stochastic
complexity $\mathfrak C_{\mathrm{joint},n}$ through the local empirical
criterion fluctuation
\[
\sup_{\substack{
\vartheta\in\mathcal B_{J_n}:\\
d_{\vartheta,n}(\vartheta,\vartheta_{0,J_n})>0
}}
\frac{
\left|
\{
\widehat Q_k-Q_{0,n}
\}(\vartheta)
-
\{
\widehat Q_k-Q_{0,n}
\}(\vartheta_{0,J_n})
\right|
}{
d_{\vartheta,n}(\vartheta,\vartheta_{0,J_n})
}
=
O_p
\left(
\sqrt{
\frac{\mathfrak C_{\mathrm{joint},n}}
{|\mathcal T_k|}
}
\right)
\]
uniformly over folds on a local neighborhood $\mathcal B_{J_n}$. The
quantity $\mathfrak C_{\mathrm{joint},n}$ is allowed to depend on both
$J_{\mathrm{joint},n}$ and $d_{\mathfrak m,n}$. Under normalized local
bases and uniformly bounded moment envelopes it can be of the same order as
the effective number of target-sensitive coefficients, but no such
simplification is imposed by notation.

Define the local penalty drift
\[
b_{\mathrm{joint},n}^{\mathrm{pen}}
=
\sup_{\substack{
\vartheta\in\mathcal B_{J_n}:\\
d_{\vartheta,n}(\vartheta,\vartheta_{0,J_n})>0
}}
\frac{
\left|
\sum_{a}
\lambda_{a,n}
\{
\mathcal P_a(\gamma_a)
-
\mathcal P_a(\gamma_{a,0,J_n})
\}
\right|
}{
d_{\vartheta,n}(\vartheta,\vartheta_{0,J_n})
},
\]
where the sum is over $a\in\{m,g,h,\ell\}$.

\begin{proposition}[Consistency and rate of the feasible joint sieve start]
\label{prop:joint_preliminary}

Suppose the population stacked-GMM criterion is globally separated at
$\vartheta_{0,J_n}$, the sieve approximation errors vanish, the
auxiliary-sample criterion satisfies a uniform spatial law of large numbers,
and
\[
\max_{1\leq k\leq K}
\left\|
\widehat{\mathcal W}_k-\mathcal W_0
\right\|_{\mathrm{op}}
=
o_p(1).
\]
Suppose also that
\[
\max_{a\in\{m,g,h,\ell\}}
\lambda_{a,n}
\sup_{\gamma_a\in\Gamma_{a,J_n}}
\mathcal P_a(\gamma_a)
=
o(1).
\]
Then, uniformly over the fixed number of folds,
\[
d_{\vartheta,n}
\left(
\widetilde\vartheta^{(-k)},
\vartheta_{0,J_n}
\right)
=
o_p(1).
\]

In addition, suppose that on a neighborhood
$\mathcal B_{J_n}$ containing the population sieve solution,
\[
Q_{0,n}(\vartheta)
-
Q_{0,n}(\vartheta_{0,J_n})
\geq
\kappa_{\mathrm{joint},n}
d_{\vartheta,n}^2
(\vartheta,\vartheta_{0,J_n})
\]
for some $\kappa_{\mathrm{joint},n}>0$, and that the local empirical
criterion fluctuation is governed by
$\mathfrak C_{\mathrm{joint},n}$ as defined above. Let
$\epsilon_{\mathrm{opt},n}$ bound the criterion suboptimality of the
numerical solution relative to the local minimum. If
\[
\sqrt{
\frac{\mathfrak C_{\mathrm{joint},n}}
{|\mathcal T_k|}
}
+
b_{\mathrm{joint},n}^{\mathrm{pen}}
=
o(\kappa_{\mathrm{joint},n}),
\qquad
\epsilon_{\mathrm{opt},n}
=
o_p(\kappa_{\mathrm{joint},n}),
\]
then
\[
d_{\vartheta,n}
\left(
\widetilde\vartheta^{(-k)},
\vartheta_{0,J_n}
\right)
=
O_p(r_{\mathrm{joint},n}),
\]
where
\[
r_{\mathrm{joint},n}
=
\frac{1}{\kappa_{\mathrm{joint},n}}
\left[
\sqrt{
\frac{\mathfrak C_{\mathrm{joint},n}}
{|\mathcal T_k|}
}
+
b_{\mathrm{joint},n}^{\mathrm{pen}}
\right]
+
\sqrt{
\frac{\epsilon_{\mathrm{opt},n}}
{\kappa_{\mathrm{joint},n}}
}.
\]
If the population sieve approximation satisfies
\[
d_{\vartheta,n}
(\vartheta_{0,J_n},\vartheta_0)
\leq
a_{\mathrm{joint},n},
\]
then
\[
d_{\vartheta,n}
\left(
\widetilde\vartheta^{(-k)},
\vartheta_0
\right)
=
O_p
\left(
r_{\mathrm{joint},n}
+
a_{\mathrm{joint},n}
\right).
\]
If the local empirical-fluctuation and optimization bounds above also hold
with uniformly bounded fourth moments after normalization by their displayed
rates, then
\[
\max_{1\leq k\leq K}
E
\left[
d_{\vartheta,n}^{4}
\left(
\widetilde\vartheta^{(-k)},
\vartheta_0
\right)
\right]
\leq
C
\left(
r_{\mathrm{joint},n}
+
a_{\mathrm{joint},n}
\right)^4.
\]

\end{proposition}

\noindent\textbf{Proof.}
See Appendix \ref{app:proof_joint_preliminary}.

The first part places the feasible estimator in the locally identified
basin. The second part applies the standard local-quadratic sieve
minimum-distance logic of \citet{ChenPouzo2012} to the \emph{entire} joint
learner, thereby closing the rate chain rather than assuming a rate for the
components held fixed in the subsequent $g$-profiling step. In particular,
the preliminary errors in $\theta$, $m$, $h$, and $\ell$ are bounded by the
same independently derived joint rate. The effective complexity
$\mathfrak C_{\mathrm{joint},n}$ explicitly allows the number of stacked
moments to grow with the sieve dimension, so the rate does not silently
treat a growing moment dictionary as fixed.

Because the joint criterion is generally nonconvex, the global minimization
description above is used to establish entrance into the identified
neighborhood. All local rate and coupling arguments below concern the
normalized local minimizer in that neighborhood. The computational
multi-start procedure is required to return this local solution with
probability approaching one; the theory does not rely on continuity of a
global argmin map across separated basins.

\subsubsection{Profiled interaction learner and rates}

I next refine the interaction component through a profiled sieve learner.
Write
\[
\widetilde\vartheta_{-g}^{(-k)}
\]
for all components of the preliminary estimator other than $\gamma_g$.
Because Proposition \ref{prop:joint_preliminary} supplies a rate for the
entire preliminary vector, the nuisance components held fixed in this
profiling step are no longer treated as rate-free inputs.

Let
\[
\widehat s_{g,k}
\left(
\gamma;
\widetilde\vartheta_{-g}^{(-k)}
\right)
=
\frac{1}{|\mathcal T_k|}
\sum_{i\in\mathcal T_k}
s_{g,i}
\left(
\gamma;
\widetilde\vartheta_{-g}^{(-k)}
\right)
\]
denote the fold-specific average interaction moment evaluated at $\gamma$
with all remaining components fixed at the joint preliminary estimator.
Define
\[
Q_{g,k}(\gamma)
=
\widehat s_{g,k}
\left(
\gamma;
\widetilde\vartheta_{-g}^{(-k)}
\right)^{\prime}
\widehat{\mathcal W}_{g,k}
\widehat s_{g,k}
\left(
\gamma;
\widetilde\vartheta_{-g}^{(-k)}
\right)
+
\lambda_{g,n}\mathcal P_g(\gamma),
\]
where $\widehat{\mathcal W}_{g,k}$ is the positive-definite GMM weighting
matrix for the interaction-learning moments, constructed using only
$\mathcal T_k$. Let $Q_{g,0}(\gamma)$ denote the corresponding profiled
population criterion on the normalized sieve space.

Let $\mathcal B_{g,k,n}$ be the normalized locally identified basin
containing the population sieve representative $\gamma_{0,J_g}$. Its radius
may shrink with $n$, but is chosen large enough that
\[
r_{\mathrm{joint},n}
+
a_{\mathrm{joint},n}
=
o(b_{g,n})
\]
for the basin radius $b_{g,n}$. Proposition
\ref{prop:joint_preliminary} then implies that the preliminary estimator lies
in this basin with probability approaching one. Define
\[
\mathcal M_{g,k}^{\mathrm{loc}}
=
\arg\min_{\gamma\in\Gamma_{g,J_g}\cap\mathcal B_{g,k,n}}
Q_{g,k}(\gamma),
\]
and select
\[
\widehat\gamma^{(-k)}
=
\arg\min_{\gamma\in\mathcal M_{g,k}^{\mathrm{loc}}}
\|\gamma\|_2.
\]
The minimum-norm rule only selects a representative within the observational
equivalence class inside the identified basin; it is not used to select
among separated nonconvex basins. Define
\[
\widehat g^{(-k)}(r)
=
p_{g,J_g}(r)^{\prime}
\widehat\gamma^{(-k)}.
\]

Let
\[
\widetilde m^{(-k)},
\qquad
\widetilde h^{(-k)},
\qquad
\widetilde\ell^{(-k)}
\]
denote the function estimates encoded by the corresponding components of the
joint preliminary estimator $\widetilde\vartheta^{(-k)}$. The fold-specific
nuisance vector used in the final score is
\[
\widehat\eta^{(-k)}
=
\left(
\widetilde m^{(-k)},
\widehat g^{(-k)},
\widetilde h^{(-k)},
\widetilde\ell^{(-k)}
\right).
\]

For notational simplicity in the rate statement below, write $J_n=J_g$ and
let
\[
d_{B_g,n}
=
\dim B_g(\mathcal P_i).
\]
Define the local penalty contribution
\[
b_{g,n}^{\mathrm{pen}}
=
\lambda_{g,n}
\left\|
\nabla\mathcal P_g(\gamma_{0,J_n})
\right\|_2.
\]
Let $\mathfrak C_{g,n}$ denote the effective stochastic complexity of the
profiled interaction score. It is defined so that the target-relevant
empirical gradient obeys
\[
\left\|
\nabla_\gamma
Q_{g,k}(\gamma_{0,J_n})
-
E[
\nabla_\gamma
Q_{g,k}(\gamma_{0,J_n})
]
\right\|_2
=
O_p
\left(
\sqrt{
\frac{\mathfrak C_{g,n}}
{|\mathcal T_k|}
}
\right).
\]
The quantity $\mathfrak C_{g,n}$ is allowed to depend on both $J_n$ and the
dimension $d_{B_g,n}$ of the interaction-moment dictionary. Under normalized
local bases, bounded moment envelopes, and
$d_{B_g,n}\lesssim J_n$ with stable moment-derivative operator norms, the
benchmark $\mathfrak C_{g,n}\lesssim J_n$ recovers the familiar
$\sqrt{J_n/n}$ stochastic term. When the interaction dictionary grows more
quickly, its additional complexity remains explicit through
$\mathfrak C_{g,n}$.

\begin{proposition}[Sieve rate for the interaction learner with a closed preliminary-rate chain]
\label{prop:sieve_g_rate}

Suppose that, after selecting a representative of $[g_0]_n$, $g_0$ is
$\alpha$-smooth on a compact $d_R$-dimensional domain and admits a normalized
sieve approximation
\[
g_{0,J_n}(r)
=
p_{J_n}(r)^{\prime}\gamma_{0,J_n}
\]
satisfying
\[
\operatorname{dist}_{\mathcal G,\mathrm{tar}}
\left(
g_{0,J_n},[g_0]_n
\right)
\leq
a_{J_n}^{\mathrm{tar}},
\qquad
a_{J_n}^{\mathrm{tar}}
\lesssim
J_n^{-\alpha/d_R},
\]
and
\[
\operatorname{dist}_{\mathcal G,+}
\left(
g_{0,J_n},[g_0]_n
\right)
\leq
a_{J_n}^{+},
\qquad
a_{J_n}^{+}
\rightarrow0.
\]
For compactly supported splines of order exceeding $\alpha$, the benchmark
approximation also satisfies
\[
a_{J_n}^{+}
\lesssim
J_n^{-\alpha/d_R}.
\]

Assume that the profiled population criterion is locally quadratically
identified on the normalized target-relevant sieve space, with curvature
bounded away from zero, and that its empirical target-relevant gradient is
of order
\[
O_p
\left(
\sqrt{
\frac{\mathfrak C_{g,n}}
{|\mathcal T_k|}
}
\right).
\]
Suppose the target-relevant sieve norm is locally dominated by the normalized
coefficient norm. By Proposition \ref{prop:joint_preliminary}, suppose the
components held fixed in the profile satisfy
\[
r_{-g,n}
\lesssim
r_{\mathrm{joint},n}
+
a_{\mathrm{joint},n}.
\]
Then, uniformly over folds,
\[
\operatorname{dist}_{\mathcal G,\mathrm{tar}}
\left(
\widehat g^{(-k)},[g_0]_n
\right)
=
O_p
\left[
\sqrt{
\frac{\mathfrak C_{g,n}}
{|\mathcal T_k|}
}
+
a_{J_n}^{\mathrm{tar}}
+
r_{\mathrm{joint},n}
+
a_{\mathrm{joint},n}
+
b_{g,n}^{\mathrm{pen}}
\right].
\]

Suppose additionally that the local spline empirical process gives the
supremum-norm bound
\[
\operatorname{dist}_{\mathcal G,+}
\left(
\widehat g^{(-k)},[g_0]_n
\right)
=
O_p(r_{g,+,n}),
\]
with
\[
r_{g,+,n}
\lesssim
\sqrt{
\frac{\mathfrak C_{g,\infty,n}\log n}
{|\mathcal T_k|}
}
+
a_{J_n}^{+}
+
r_{\mathrm{joint},n}
+
a_{\mathrm{joint},n}
+
b_{g,n}^{\mathrm{pen}},
\]
where $\mathfrak C_{g,\infty,n}$ is the corresponding local-basis envelope
complexity. For compactly supported B-splines with normalized local bases,
a representative benchmark is
\[
\mathfrak C_{g,n}
\lesssim
J_n,
\qquad
\mathfrak C_{g,\infty,n}
\lesssim
J_n.
\]

If $|\mathcal T_k|\asymp n$,
$\mathfrak C_{g,n}\lesssim J_n$, the joint preliminary and penalty terms are
of no larger order than the stochastic and approximation terms, and $J_n$
balances
\[
\sqrt{\frac{J_n}{n}}
\qquad\text{and}\qquad
J_n^{-\alpha/d_R},
\]
then
\[
\operatorname{dist}_{\mathcal G,\mathrm{tar}}
\left(
\widehat g^{(-k)},[g_0]_n
\right)
=
O_p
\left(
n^{-\alpha/(2\alpha+d_R)}
\right)
\]
and, up to the usual logarithmic factor,
\[
\operatorname{dist}_{\mathcal G,+}
\left(
\widehat g^{(-k)},[g_0]_n
\right)
=
O_p
\left(
n^{-\alpha/(2\alpha+d_R)}
\sqrt{\log n}
\right).
\]
Hence both the target-relevant rate and the stronger smoothness-envelope rate
are $o(n^{-1/4})$ whenever
\[
\alpha>\frac{d_R}{2}.
\]
If the corresponding target-gradient and local sup-norm maximal inequalities
hold with uniformly bounded fourth moments after normalization by their
rates, the two displayed convergence statements also hold with the
fourth-moment bounds required by Assumption \ref{ass:rates}.

\end{proposition}

\noindent\textbf{Proof.}
See Appendix \ref{app:proof_sieve_g_rate}.

The proposition now closes the rate chain. The error of the preliminary
components entering the profile is controlled by the independently derived
joint sieve-GMM rate in Proposition \ref{prop:joint_preliminary}; it is not
bounded by invoking the rate that the profiled $g$ learner is itself intended
to prove. The effective complexity $\mathfrak C_{g,n}$ also prevents a
growing interaction-moment dictionary from being hidden inside a
fixed-dimension $\sqrt{J_n/n}$ notation.

The second rate is used only for smoothness of nonlinear operator
compositions. Combined with Assumption \ref{ass:operator_map_smoothness}, it
justifies the quadratic remainder underlying the orthogonal-score expansion.
The logarithmic factor does not alter the threshold
$\alpha>d_R/2$ for the benchmark B-spline construction.

This result remains a well-posed benchmark for direct estimation of $g_0$.
Possible mild ill-posedness of the full nuisance Jacobian used in
constructing the Riesz correction in Section \ref{sec:asymptotics} is a
separate issue.

The unrestricted tensor-product benchmark also makes the role of
dimensionality transparent. When $d_R$ is large, a more attractive primitive
specification is an additive or low-order interaction sieve of the form
\[
g_0(r)
=
\sum_{a=1}^{d_R}g_{0a}(r_a)
+
\sum_{(a,b)\in\mathcal E_g}
g_{0,ab}(r_a,r_b),
\]
where $\mathcal E_g$ contains a prespecified collection of economically
meaningful interactions. Under such a structure, the relevant nonparametric
rate is governed by the largest component dimension rather than by the full
tensor-product dimension $d_R$.

The same consideration applies to the remaining smooth nuisance functions.
For a generic scalar nuisance function $f_0$ with input dimension $d_f$ and
smoothness $\alpha_f$, a conventional tensor-product sieve has benchmark
error
\[
O_p
\left(
\sqrt{\frac{\mathfrak C_{f,n}}{n}}
+
J_{f,n}^{-\alpha_f/d_f}
\right),
\]
where $\mathfrak C_{f,n}$ records the effective coefficient-and-moment
complexity of the corresponding learner.

The conditional projection $\ell_0$ requires separate accounting because it
is vector valued. If each coordinate of $\ell_0$ is estimated at
root-mean-square rate $\rho_{\ell,n}$, then the Euclidean error entering the
nuisance metric is, in general,
\[
r_{\ell,n}^{\mathrm{agg}}
=
O_p
\left(
\sqrt{d_{B_g,n}}\,
\rho_{\ell,n}
\right).
\]
Thus growth of the interaction dictionary is not free. For a
$d_C$-dimensional tensor-product sieve with coordinate smoothness
$\alpha_\ell$, a representative coordinatewise benchmark is
\[
\rho_{\ell,n}
=
O_p
\left(
\sqrt{\frac{\mathfrak C_{\ell,n}}{n}}
+
J_{\ell,n}^{-\alpha_\ell/d_C}
\right),
\]
so the admissible growth of $d_{B_g,n}$, the complexity of $g_0$, and the
smoothness of $\ell_0$ must be chosen jointly. The condition
$\alpha>d_R/2$ is therefore only a benchmark for the interaction block; the
complete nuisance vector must satisfy the joint and aggregate rate
restrictions stated below.

\subsection{Buffered Spatial Cross-Fitting}
\label{subsec:crossfitting}

Partition the units into $K$ geographically coherent evaluation blocks
\[
\mathcal I_1,\ldots,\mathcal I_K,
\qquad
\bigcup_{k=1}^{K}\mathcal I_k
=
\{1,\ldots,n\},
\]
with pairwise disjoint blocks.

Let $d_{ij}^{*}$ denote a predetermined spatial metric and define
\[
d^{*}(j,\mathcal I)
=
\min_{i\in\mathcal I}d_{ij}^{*}.
\]

Let $a_n^W$ denote the radius used to localize one application of the spatial
operator and let $L_\psi$ denote the maximum number of successive spatial
operator applications entering the score. Let $b_n$ bound the complete raw
score footprint, with
\[
b_n\geq L_\psi a_n^W.
\]
The radius $b_n$ is enlarged whenever a predetermined residual summary in
$C_i(m)$ or another score component requires a larger raw-data footprint.

Define
\[
\mathcal E_k(b_n)
=
\{j:d^{*}(j,\mathcal I_k)\leq b_n\}.
\]
For a guard distance $s_n>0$, define
\[
\mathcal A_k^{\mathrm{aux}}
=
\{j:d^{*}(j,\mathcal I_k)>b_n+s_n\}.
\]

If $b_n^{\mathrm{tr}}$ bounds the raw-data footprint of one localized
training moment centered at $j$, let
\[
\mathcal N_{j,n}(a)
=
\{\ell:d_{j\ell}^{*}\leq a\}
\]
and
\[
\mathcal T_k
=
\left\{
j\in\mathcal A_k^{\mathrm{aux}}:
\mathcal N_{j,n}(b_n^{\mathrm{tr}})
\subseteq
\mathcal A_k^{\mathrm{aux}}
\right\}.
\]

All nuisance fitting, regularization selection, tuning, estimation of the
conditional projections, and estimation of the debiasing operator for fold
$k$ use only the training observations indexed by $\mathcal T_k$ and their
required localized raw-data footprints.

The spatial cross-fitting design is therefore
\[
\text{evaluation block}
\longrightarrow
\text{score footprint}
\longrightarrow
\text{guard region}
\longrightarrow
\text{training footprint}.
\]
Every unit is used as an evaluation observation exactly once. The footprint
and guard regions reduce only the fold-specific nuisance-training sample.

In words, ordinary sample splitting is not enough in a spatial
cross section because the score evaluated at one observation may directly use
neighboring outcomes and covariates, while nearby training observations may
remain statistically dependent with that score. The score footprint first
removes observations that are mechanically used in constructing the
evaluation score. The additional guard region then increases the spatial
separation between the evaluation information and nuisance-training
information. Section \ref{sec:asymptotics} shows how this separation, combined
with NED approximation of the globally simultaneous SAR outcome, makes the
remaining training-to-evaluation dependence asymptotically negligible.

\section{Asymptotics}\label{sec:asymptotics}

The asymptotic argument separates the high-level orthogonal-score result from
its primitive spatial verification. The high-level result requires spatial
laws of large numbers and a central limit theorem for the oracle score,
target-relevant nuisance rates, local score smoothness, and sufficiently weak
interaction between the nuisance-training error and the derivative of an
evaluation score.

The spatial difficulty differs from the usual i.i.d.\ cross-fitting problem.
Even when nuisance functions are estimated outside an evaluation block, the
training data and evaluation score need not be independent. Moreover, the
observed SAR outcome is globally simultaneous. I therefore decompose the
first-order cross-fitting error into a centered empirical fluctuation and a
training-to-evaluation leakage term. The first is controlled by spatial
short-memory conditions on the score derivative, whereas the second is made
small by the guard region introduced in Section \ref{subsec:crossfitting}.

I then give primitive sufficient conditions based on NED on a spatially
mixing innovation field. The primitive argument explicitly accounts for the
fact that $W_0=W_n(g_0)$ is itself random because its entries depend on the
possibly endogenous characteristics $Z_i$. Local support, smoothness of the
normalized weight map, stability of the SAR resolvent, and spatial separation
jointly imply that the feasible cross-fitted score is asymptotically
equivalent to the oracle score.

\subsection{High-Level Conditions}\label{subsec:highlevel}

\begin{assumption}[Parameter space and spatial stability]
\label{ass:stability}

The true parameter $\theta_0$ belongs to the interior of a compact parameter
space $\Theta$. Uniformly over admissible $g$,
\[
\|W_n(g)\|_{\infty}=1,
\qquad
\sup_n\|W_n(g)\|_1<\infty.
\]
There exists $\delta_\rho>0$ such that
\[
|\rho|
\leq
1-\delta_\rho
\]
for every $\theta\in\Theta$.

\end{assumption}

Row normalization keeps the direct propagation of a spatial shock bounded,
while the column-sum condition rules out increasingly concentrated incoming
influence as the sample grows. The restriction on $\rho$ keeps the SAR model
uniformly away from the instability boundary. Together these conditions make
the spatial multiplier
$(I_n-\rho W_n(g))^{-1}$ well behaved.

\begin{assumption}[Spatial locality]
\label{ass:spatial_locality}

Define
\[
\tau_{W,n}(r)
=
\sup_{g\in\mathcal G}
\max_i
\sum_{\substack{j\neq i:\\d_{ij}^{*}>r}}
w_{ij}(g).
\]
Then
\[
\tau_{W,n}(r)
\rightarrow
0
\]
as $r\rightarrow\infty$, uniformly in $n$ along the maintained sequence.

Similarly, define
\[
\tau_{C,n}(r)
=
\max_{\substack{i\\1\leq\ell\leq L_C}}
\sum_{\substack{j\neq i:\\d_{ij}^{*}>r}}
|\kappa_{ij,n}^{C,\ell}|.
\]
Then
\[
\tau_{C,n}(r)
\rightarrow
0
\]
as $r\rightarrow\infty$.

For the localization sequence used in Assumption
\ref{ass:score_localization}, the chosen radii must additionally satisfy the
rate requirement generated by the relevant score composition. In particular,
when truncation of the control summaries contributes linearly to the score,
a sufficient condition is
\[
\sqrt n\,\tau_{C,n}(b_n)\rightarrow0.
\]
This rate is automatic under the fixed-radius primitive benchmark below and,
more generally, must be verified jointly with the operator-localization
error.

\end{assumption}

The assumption allows interaction to extend beyond immediate neighbors, but
requires sufficiently distant observations to have progressively little
direct influence. The same requirement is imposed on the predetermined
residual summaries used in the control function. Thus a spatial score can be
approximated by one depending only on a sufficiently large local
neighborhood.

\begin{assumption}[Oracle spatial LLN and CLT]
\label{ass:oracle_clt}

Let
\[
\psi_i^0
=
\psi_i(\theta_0,\eta_0,\Gamma_0)
\]
denote the oracle orthogonal score defined below, and let
\[
\overline\psi_n^0(\theta)
=
\frac{1}{n}
\sum_{i=1}^{n}
\psi_i(\theta,\eta_0,\Gamma_0),
\qquad
\mu_n(\theta)
=
\frac{1}{n}
\sum_{i=1}^{n}
E[
\psi_i(\theta,\eta_0,\Gamma_0)
].
\]
The oracle score satisfies the uniform spatial law of large numbers
\[
\sup_{\theta\in\Theta}
\left\|
\overline\psi_n^0(\theta)
-
\mu_n(\theta)
\right\|_2
=
o_p(1).
\]
For some neighborhood $\mathcal N_\theta$ of $\theta_0$, the oracle target
derivative also satisfies
\[
\sup_{\theta\in\mathcal N_\theta}
\left\|
\frac{1}{n}
\sum_{i=1}^{n}
\partial_\theta
\psi_i(\theta,\eta_0,\Gamma_0)
-
\frac{1}{n}
\sum_{i=1}^{n}
E[
\partial_\theta
\psi_i(\theta,\eta_0,\Gamma_0)
]
\right\|_{\mathrm{op}}
=
o_p(1).
\]

Its long-run covariance satisfies
\[
\Omega_n
=
\frac{1}{n}
\sum_{i=1}^{n}
\sum_{j=1}^{n}
\operatorname{Cov}
(\psi_i^0,\psi_j^0)
\rightarrow
\Omega_0,
\]
where $\Omega_0$ is finite and positive definite, and
\[
\frac{1}{\sqrt n}
\sum_{i=1}^{n}
\psi_i^0
\xrightarrow{d}
N(0,\Omega_0).
\]

\end{assumption}

This assumption makes explicit the two oracle uniform laws used later:
uniform convergence of the population GMM criterion and replacement of the
sample target Jacobian by its population counterpart. The primitive spatial
conditions developed in Section \ref{subsec:ned}, summarized in Corollary
\ref{cor:primitive_benchmark}, provide sufficient conditions for these
oracle laws and the CLT. The high-level formulation is retained to allow
alternative primitive dependence conditions.

\begin{assumption}[Generated-$W$ relevance]
\label{ass:w_sensitivity}

Suppose the interaction operator is relevant, so that
\[
|\rho_0|
\geq
\underline\rho
>
0.
\]
The valid instrument vector contains at least one component $H_i^0$
independent of $g$ for which there exists an admissible $\delta g$ satisfying
\[
E
\left[
H_i^0
D_g\{W_n(g_0)Y\}_i[\delta g]
\right]
\neq
0.
\]

\end{assumption}

The assumption rules out the uninteresting case in which estimating $W_0$
has no first-order effect on the target moment. It ensures that there is at
least one direction in which an error in $g_0$ changes the spatial lag in a
way that matters for IV estimation. The next proposition then shows why a
naive plug-in estimator generally cannot treat $\widehat W$ as if it were
known.

\begin{proposition}[Nonorthogonality of the naive plug-in moment]
\label{prop:naive_nonorthogonal}

Under Assumptions \ref{ass:instrument_validity} and
\ref{ass:w_sensitivity},
\[
D_g
E
\left[
H_i^0
\xi_i(\theta_0,\eta_0)
\right]
[\delta g]
\neq
0
\]
for at least one admissible direction $\delta g$. Hence treating
$W_n(\widehat g)$ as known generally leaves a first-order generated-$W$
effect.

\end{proposition}

\noindent\textbf{Proof.} See Appendix \ref{app:proof_naive_nonorthogonal}.

An estimation error in $g_0$ changes the regressor $W_0Y$ at first order.
Consequently, even if the SAR instrument itself does not depend on $g$, the
usual IV moment inherits a first-order error from estimating the spatial
operator. When the instrument also contains objects such as $W_0X$, there
are additional generated-$W$ channels. Orthogonalization is therefore needed
to prevent the first-stage learning error from entering the limiting
distribution of $\widehat\theta$.

\subsection{Operator-Orthogonal Score}\label{subsec:orthogonal}

For block $k$, define
\[
G_{\eta,0,k}[\delta\eta]
=
\frac{1}{|\mathcal I_k|}
\sum_{i\in\mathcal I_k}
D_\eta
E[
\phi_i(\theta_0,\eta_0)
]
[\delta\eta]
\]
and
\[
A_{\eta,0,k}[\delta\eta]
=
\frac{1}{|\mathcal I_k|}
\sum_{i\in\mathcal I_k}
D_\eta
E[
s_i(\theta_0,\eta_0)
]
[\delta\eta].
\]

Let $d_{\eta,n}$ denote the target-relevant nuisance semimetric defined
below. For a linear map $\mathcal L$ on nuisance directions, define
\[
\|\mathcal L\|_{\eta,n}^{*}
=
\sup_{\substack{
\delta\eta:\\
d_{\eta,n}(\delta\eta)>0
}}
\frac{
\|\mathcal L[\delta\eta]\|_2
}{
d_{\eta,n}(\delta\eta)
}.
\]

\begin{assumption}[Approximately common blockwise Riesz representation]
\label{ass:riesz}

There exists a linear operator $\Gamma_0$ satisfying
\[
\|\Gamma_0\|_{\mathrm{op}}
\leq
C
\]
and a deterministic sequence
$\delta_{R,n}\rightarrow0$ such that
\[
\sup_{1\leq k\leq K}
\|
G_{\eta,0,k}
-
\Gamma_0A_{\eta,0,k}
\|_{\eta,n}^{*}
\leq
\delta_{R,n}.
\]

\end{assumption}

The nuisance moments contain information about the same local nuisance
directions that affect the target moment. The operator $\Gamma_0$ combines
those nuisance moments so that their first-order sensitivity approximates the
first-order sensitivity of the target moment. Subtracting this combination
therefore removes the leading nuisance-estimation effect.

Importantly, the representation is required only on the nuisance directions
that matter for the target score. It is not necessary to invert an
unrestricted infinite-dimensional nuisance operator.

A primitive route to the approximately common blockwise representation is
asymptotic homogeneity of the population derivative maps. In particular,
suppose there exist common maps $A_{\eta,0}$ and $G_{\eta,0}$ such that
\[
\sup_{1\leq k\leq K}
\left\{
\|A_{\eta,0,k}-A_{\eta,0}\|_{\eta,n}^{*}
+
\|G_{\eta,0,k}-G_{\eta,0}\|_{\eta,n}^{*}
\right\}
\leq
\delta_{B,n},
\qquad
\delta_{B,n}\rightarrow0,
\]
and $G_{\eta,0}=\Gamma_0A_{\eta,0}$ on the target-relevant tangent space.
Then the same $\Gamma_0$ satisfies Assumption \ref{ass:riesz} with a block
error of order $\delta_{B,n}$. Such a condition follows, for example, under
an increasing-domain design with a fixed number of regular spatial blocks
when the relevant population derivative fields are spatially homogeneous
and their block averages converge to the same limits.

The stacked derivative $A_{\eta,0}$ is essential here. Because $C_i(m)$
contains spatial summaries of first-stage residuals, a perturbation of $m$
can enter the target derivative through terms involving
$\sum_j\kappa_{ij,n}^{C,r}\delta m(V_j)$. The representation is therefore
not based on the own-unit $m$ moment alone: derivatives of the $h$, $\ell$,
and $g$ nuisance blocks with respect to $m$ are included in
$A_{\eta,0}$ and carry these spatially aggregated directions. Full column
rank below is imposed on this complete stacked system.

The possible ill-posedness considered below concerns construction of this
full nuisance-to-target Riesz correction. It is distinct from the well-posed
benchmark imposed in Proposition \ref{prop:sieve_g_rate} for direct learning
of the target-relevant $g$ block. Even when the $g$ block itself has stable
local curvature, weak singular directions can arise from another nuisance
block or from combinations of nuisance directions in the full stacked
Jacobian.

Let
\[
A_{\eta,0}^{(J)}
\in
\mathbb R^{d_{s,n}\times J}
\]
and
\[
G_{\eta,0}^{(J)}
\in
\mathbb R^{q\times J}
\]
denote the nuisance Jacobian and target-sensitivity matrix on a
$J$-dimensional target-relevant sieve tangent space.

\begin{proposition}[Existence of the sieve Riesz representer]
\label{prop:riesz_existence}

Suppose
\[
\operatorname{rank}
\left(
A_{\eta,0}^{(J)}
\right)
=
J
\]
and let
\[
\kappa_J
=
\sigma_{\min}
\left(
A_{\eta,0}^{(J)}
\right)
>
0.
\]
Then
\[
\Gamma_{0,J}
=
G_{\eta,0}^{(J)}
\left(
A_{\eta,0}^{(J)}
\right)^\dagger
\]
satisfies
\[
G_{\eta,0}^{(J)}
=
\Gamma_{0,J}
A_{\eta,0}^{(J)}
\]
exactly on the sieve tangent space.

The full stacked nuisance system may be mildly ill posed, so that
\[
\kappa_J
\downarrow
0.
\]
A bounded population representer is obtained under the source condition
\[
\sup_J
\left\|
G_{\eta,0}^{(J)}
\left(
A_{\eta,0}^{(J)}
\right)^\dagger
\right\|_{\mathrm{op}}
<
\infty.
\]

If, in addition,
\[
\|\Gamma_{0,J}-\Gamma_0\|_{\mathrm{op}}
\leq
a_{\Gamma,J},
\qquad
a_{\Gamma,J}
\rightarrow
0,
\]
and block-specific derivative matrices differ from the common population
matrices by at most $\delta_{B,n}$ in the corresponding target-relevant
operator norm, then Assumption \ref{ass:riesz} holds with
\[
\delta_{R,n}
\lesssim
a_{\Gamma,J}
+
\delta_{B,n}.
\]

\end{proposition}

\noindent\textbf{Proof.}
See Appendix \ref{app:proof_riesz_existence}.

In a finite sieve, orthogonalization reduces to a matrix projection. Full
column rank guarantees that every target-relevant nuisance direction can be
represented using the nuisance moments. The source condition allows some
singular values of the full nuisance Jacobian to become small, but requires
the target sensitivity to place sufficiently little weight on the unstable
directions. Thus mild ill-posedness may slow estimation of the debiasing
operator without contradicting stable direct estimation of $g_0$.

For estimation, use the Tikhonov-type pseudoinverse
\[
\widehat A_{\eta,\lambda}^{\dagger}
=
\left(
\widehat A_\eta^{\prime}
\widehat A_\eta
+
\lambda_{\Gamma,n}I
\right)^{-1}
\widehat A_\eta^{\prime}
\]
and define
\[
\widehat\Gamma^{(-k)}
=
\widehat G_\eta^{(-k)}
\widehat A_{\eta,\lambda}^{(-k)\dagger}.
\]
The complete fold-specific nuisance object used below is
\[
\widehat\zeta^{(-k)}
=
\left(
\widehat\eta^{(-k)},
\widehat\Gamma^{(-k)}
\right).
\]

\begin{proposition}[Rate for the debiasing operator]
\label{prop:gamma_rate}

Let $J_{\Gamma,n}$ denote the target-sensitive sieve dimension and write
\[
\kappa_{\Gamma,n}
=
\sigma_{\min}
\left(
A_{\eta,0}^{(J_{\Gamma,n})}
\right).
\]
Suppose the source condition in Proposition
\ref{prop:riesz_existence} holds and
\[
\|\Gamma_{0,J_{\Gamma,n}}-\Gamma_0\|_{\mathrm{op}}
\leq
a_{\Gamma,J_{\Gamma,n}}.
\]

Uniformly over folds, suppose
\[
\|
\widehat A_\eta^{(-k)}
-
A_{\eta,0}^{(J_{\Gamma,n})}
\|_{\mathrm{op}}
=
O_p(\Delta_{A,n})
\]
and
\[
\|
\widehat G_\eta^{(-k)}
-
G_{\eta,0}^{(J_{\Gamma,n})}
\|_{\mathrm{op}}
=
O_p(\Delta_{G,n}).
\]
Let $d_{s,n}$ denote the dimension of the stacked nuisance-moment vector
entering $\widehat A_\eta$. Define an effective matrix complexity
$\mathfrak C_{\Gamma,n}$ so that a representative finite-sieve bound is
\[
\Delta_{A,n}
+
\Delta_{G,n}
=
O_p
\left(
\sqrt{
\frac{\mathfrak C_{\Gamma,n}}
{|\mathcal T_k|}
}
+
r_{-\Gamma,n}
\right).
\]
The quantity $\mathfrak C_{\Gamma,n}$ records the joint effect of the
target-sensitive tangent dimension $J_{\Gamma,n}$, the stacked moment
dimension $d_{s,n}$, basis envelopes, and the operator-norm empirical process
used to estimate the two derivative matrices. In the well-normalized
benchmark with $d_{s,n}\lesssim J_{\Gamma,n}$ and stable local basis
envelopes one may have
$\mathfrak C_{\Gamma,n}\lesssim J_{\Gamma,n}$ up to logarithmic factors, but
this reduction is not imposed when the moment dictionary grows separately.

If
\[
\Delta_{A,n}
=
o_p(\kappa_{\Gamma,n}),
\qquad
\lambda_{\Gamma,n}
=
o(\kappa_{\Gamma,n}^{2}),
\]
then
\[
\|
\widehat\Gamma^{(-k)}
-
\Gamma_0
\|_{\mathrm{op}}
=
O_p
\left(
a_{\Gamma,J_{\Gamma,n}}
+
\frac{\Delta_{G,n}}{\kappa_{\Gamma,n}}
+
\frac{\Delta_{A,n}}{\kappa_{\Gamma,n}^{2}}
+
\frac{\Delta_{G,n}\Delta_{A,n}}{\kappa_{\Gamma,n}^{2}}
+
\frac{\lambda_{\Gamma,n}}{\kappa_{\Gamma,n}^{2}}
\right).
\]

Consequently, a sufficient condition for
\[
r_{\Gamma,n}
=
o(n^{-1/4})
\]
is
\[
a_{\Gamma,J_{\Gamma,n}}
+
\frac{\Delta_{G,n}}{\kappa_{\Gamma,n}}
+
\frac{\Delta_{A,n}}{\kappa_{\Gamma,n}^{2}}
+
\frac{\lambda_{\Gamma,n}}{\kappa_{\Gamma,n}^{2}}
=
o(n^{-1/4}).
\]

When $\kappa_{\Gamma,n}$ is bounded away from zero, this reduces to the
well-posed sieve rate. When
$\kappa_{\Gamma,n}\downarrow0$, the displayed condition makes the additional
cost of estimating the Riesz correction explicit.

\end{proposition}

\noindent\textbf{Proof.}
See Appendix \ref{app:proof_gamma_rate}.

The error in $\widehat\Gamma$ has four components: sieve approximation,
estimation of the target sensitivity, estimation of the nuisance Jacobian,
and Tikhonov regularization bias. Small singular values amplify estimation
error, which is why the last three terms are divided by powers of
$\kappa_{\Gamma,n}$. In the empirically simpler well-posed case,
$\kappa_{\Gamma,n}$ is bounded away from zero and these amplification terms
disappear.

\begin{definition}[Operator-orthogonal SAR score]
\label{def:orthogonal_score}

Define
\[
\psi_i(\theta,\eta,\Gamma)
=
\phi_i(\theta,\eta)
-
\Gamma s_i(\theta,\eta).
\]

\end{definition}

The first term is the original SAR-IV moment. The second subtracts a linear
combination of nuisance moments chosen so that the first-order effect of
estimating $m_0$, $g_0$, $h_0$, and $\ell_0$ cancels. In particular, the
correction removes the first-order generated-$W$ effect operating through
both $W_0Y$ and the spatially transformed instruments.

\begin{proposition}[Approximate blockwise Neyman orthogonality]
\label{prop:orthogonality}

Under Assumption \ref{ass:riesz},
\[
\sup_{1\leq k\leq K}
\left\|
\frac{1}{|\mathcal I_k|}
\sum_{i\in\mathcal I_k}
D_\eta
E[
\psi_i(\theta_0,\eta_0,\Gamma_0)
]
\right\|_{\eta,n}^{*}
\leq
\delta_{R,n}.
\]
For perturbations of $\Gamma$,
\[
D_\Gamma
E[
\psi_i(\theta_0,\eta_0,\Gamma_0)
]
[\delta\Gamma]
=
-
\delta\Gamma
E[
s_i(\theta_0,\eta_0)
]
=
0.
\]

\end{proposition}

\noindent\textbf{Proof.}
See Appendix \ref{app:proof_orthogonality}.

At the truth, small first-order perturbations of the nuisance functions have
only a vanishing effect on the population orthogonal score. Estimation of the
debiasing operator itself is also orthogonal because the nuisance moments have
mean zero. This converts leading nuisance-estimation effects into
second-order products, up to the small Riesz approximation error
$\delta_{R,n}$.

\begin{assumption}[Target identification]
\label{ass:rank}

Define the $n$-indexed population target Jacobian
\[
J_{\psi,0,n}
=
\left.
\frac{\partial}{\partial\theta^{\prime}}
\left\{
\frac{1}{n}
\sum_{i=1}^{n}
E[
\psi_i(\theta,\eta_0,\Gamma_0)
]
\right\}
\right|_{\theta=\theta_0}.
\]
There exists a finite matrix $J_{\psi,0}$ such that
\[
\|
J_{\psi,0,n}
-
J_{\psi,0}
\|_{\mathrm{op}}
\rightarrow
0,
\]
and
\[
\sigma_{\min}(J_{\psi,0})
\geq
c_J
>
0.
\]
Consequently, for all sufficiently large $n$,
\[
\sigma_{\min}(J_{\psi,0,n})
\geq
\frac{c_J}{2}.
\]

\end{assumption}

The orthogonal moments must retain enough variation to identify the
finite-dimensional SAR parameter uniformly along the spatial asymptotic
sequence. Writing
\[
J_{\phi,0,n}
=
\left.
\partial_{\theta^{\prime}}
\frac{1}{n}\sum_i
E[\phi_i(\theta,\eta_0)]
\right|_{\theta=\theta_0},
\qquad
J_{s,\theta,0,n}
=
\left.
\partial_{\theta^{\prime}}
\frac{1}{n}\sum_i
E[s_i(\theta,\eta_0)]
\right|_{\theta=\theta_0},
\]
the corrected Jacobian is
\[
J_{\psi,0,n}
=
J_{\phi,0,n}
-
\Gamma_0J_{s,\theta,0,n}.
\]
Thus orthogonalization can remove some variation that was relevant in the
uncorrected SAR-IV moment. Assumption \ref{ass:rank} explicitly requires the
remaining, orthogonalized instrument variation to retain full rank. A simple
sufficient condition is
\[
\sigma_{\min}(J_{\phi,0,n})
-
\|\Gamma_0J_{s,\theta,0,n}\|_{\mathrm{op}}
\geq
c_J
\]
eventually, although the maintained rank condition allows less restrictive
configurations. Orthogonalization removes first-order nuisance sensitivity;
it neither creates target identification nor guarantees that relevance is
preserved without this condition.

\begin{assumption}[Global GMM separation]
\label{ass:global_identification}

Define
\[
\mu_n(\theta)
=
\frac{1}{n}
\sum_{i=1}^{n}
E[
\psi_i(\theta,\eta_0,\Gamma_0)
]
\]
and
\[
Q_n(\theta)
=
\mu_n(\theta)^{\prime}
\mathcal M_0
\mu_n(\theta).
\]
For every $\epsilon>0$,
\[
\liminf_{n\rightarrow\infty}
\inf_{\substack{
\theta\in\Theta:\\
\|\theta-\theta_0\|_2\geq\epsilon
}}
\{
Q_n(\theta)-Q_n(\theta_0)
\}
>
0.
\]

\end{assumption}

The local rank condition identifies $\theta_0$ in a neighborhood of the
truth. Global separation additionally rules out distant parameter values
that fit the population orthogonal moments equally well. This condition is
used for consistency of the GMM minimizer before the local asymptotic
expansion is applied.

\subsection{Localization and Nuisance Rates}\label{subsec:rates}

For $a>0$, define
\[
K_{ij,a}(g)
=
K_{ij}(g)
\mathbbm 1\{d_{ij}^{*}\leq a\}
\]
and row-normalize to obtain $W_{n,a}(g)$. Let
\[
\psi_{i,b_n}(\theta,\eta,\Gamma)
\]
denote the localized score obtained by replacing spatial operators by their
$a_n^W$-localized versions and retaining the complete raw-data footprint
within $\mathcal N_{i,n}(b_n)$.

\begin{assumption}[Score localization]
\label{ass:score_localization}

There exists
$\delta_n^{\mathrm{loc}}(a,b)$ such that, uniformly over a shrinking
neighborhood of
\[
\zeta_0
=
(\eta_0,\Gamma_0),
\]
\[
\|
\psi_i(\theta_0,\zeta)
-
\psi_{i,b_n}(\theta_0,\zeta)
\|_{L^2}
\leq
\delta_n^{\mathrm{loc}}(a_n^W,b_n)
\]
and
\[
\|
D_\zeta\psi_i(\theta_0,\zeta)
-
D_\zeta\psi_{i,b_n}(\theta_0,\zeta)
\|_{L^2,\mathcal H,n}^{*}
\leq
\delta_n^{\mathrm{loc}}(a_n^W,b_n).
\]
Moreover,
\[
\sqrt n\,
\delta_n^{\mathrm{loc}}(a_n^W,b_n)
\rightarrow
0.
\]

\end{assumption}

The actual SAR score can depend on arbitrarily distant observations through
the spatial multiplier. The assumption requires a score built from a growing
but finite spatial footprint to approximate the full score accurately enough
that localization error vanishes at the root-$n$ scale. The primitive NED
results below show how local support and a stable SAR resolvent deliver such
an approximation.

For
\[
\zeta=(\eta,\Gamma),
\qquad
\zeta_0=(\eta_0,\Gamma_0),
\]
define
\begin{align*}
d_{\mathcal H,n}(\zeta,\zeta_0)
={}&
\|m-m_0\|_{L^2}
+
\operatorname{dist}_{\mathcal G,+}
(g,[g_0]_n)
\\
&+
\left[
\frac{1}{n}
\sum_{i=1}^{n}
E
\left\{
h(C_i(m))-h_0(C_{i0})
\right\}^2
\right]^{1/2}
\\
&+
\left[
\frac{1}{n}
\sum_{i=1}^{n}
E
\|\ell(C_i(m))-\ell_0(C_{i0})\|_2^2
\right]^{1/2}
\\
&+
\|\{W_n(g)-W_n(g_0)\}Y\|_{2,n}
\\
&+
\left[
\frac{1}{n}
\sum_{i=1}^{n}
E
\|H_i(g)-H_i(g_0)\|_2^2
\right]^{1/2}
\\
&+
\|\Gamma-\Gamma_0\|_{\mathrm{op}}.
\end{align*}

Let $d_{\eta,n}$ denote the restriction of
$d_{\mathcal H,n}$ to $\Gamma=\Gamma_0$. Because
$d_{\mathcal H,n}$ contains
$\operatorname{dist}_{\mathcal G,+}$, the nuisance metric used for the
quadratic score expansion now includes the stronger interaction-function
envelope required by Assumption \ref{ass:operator_map_smoothness}; the
target-relevant first-order metric remains the weaker
$\operatorname{dist}_{\mathcal G,\mathrm{tar}}$.

For the feasible learner defined in Section \ref{subsec:sieve_learning}, a
sufficient implementation-specific rate envelope is
\[
r_{\zeta,n}
\gtrsim
r_{\mathrm{joint},n}
+
a_{\mathrm{joint},n}
+
r_{g,+,n}
+
r_{\ell,n}^{\mathrm{agg}}
+
r_{\Gamma,n},
\]
where the first two terms control the jointly estimated preliminary
components, $r_{g,+,n}$ is supplied by Proposition
\ref{prop:sieve_g_rate}, $r_{\ell,n}^{\mathrm{agg}}$ accounts for the
vector-valued conditional projection, and $r_{\Gamma,n}$ is supplied by
Proposition \ref{prop:gamma_rate}. Proposition
\ref{prop:g_to_W_transfer} then controls the generated spatial lag and
instrument terms.

\begin{assumption}[Cross-fitted nuisance rates and moments]
\label{ass:rates}

Uniformly over folds,
\[
d_{\mathcal H,n}
\left(
\widehat\zeta^{(-k)},
\zeta_0
\right)
=
O_p(r_{\zeta,n}),
\qquad
r_{\zeta,n}
\rightarrow
0,
\]
and
\[
\max_{1\leq k\leq K}
E
\left[
d_{\mathcal H,n}^{4}
\left(
\widehat\zeta^{(-k)},
\zeta_0
\right)
\right]
\leq
C r_{\zeta,n}^{4}.
\]

If $d_{B_g,n}$ grows, the rate $r_{\zeta,n}$ includes the aggregate
projection error
\[
r_{\ell,n}^{\mathrm{agg}}
=
\left[
\frac{1}{n}
\sum_{i=1}^{n}
E
\|
\widehat\ell(C_i(\widehat m))
-
\ell_0(C_{i0})
\|_2^2
\right]^{1/2},
\]
and a per-coordinate projection rate $\rho_{\ell,n}$ contributes
$O_p(\sqrt{d_{B_g,n}}\rho_{\ell,n})$ absent additional structure.

The benchmark common-rate condition is
\[
\sqrt n\,
r_{\zeta,n}^{2}
\rightarrow
0.
\]

More generally, heterogeneous nuisance rates are permitted provided every
second-order product appearing in the score expansion satisfies
\[
\sqrt n\,
r_{a,n}r_{b,n}
=
o(1),
\]
with the corresponding second- and fourth-moment bounds.

\end{assumption}

Orthogonality means nuisance functions need not be estimated at the
root-$n$ rate. Under a common rate, faster than $n^{-1/4}$ is sufficient
because the leading remaining error is quadratic. The fourth-moment
requirement strengthens a purely $O_p$ rate just enough to justify
expectation bounds involving the random nuisance-training error. This avoids
implicitly converting a probability rate into an $L^1$ or $L^2$ rate.

The common-rate condition is only a convenient benchmark. For example, a
slowly estimated $g_0$ can be combined with a faster first-stage or control
function as long as every nuisance product entering the score expansion is
$o(n^{-1/2})$.

\begin{assumption}[Local score smoothness]
\label{ass:score_smoothness}

The localized score is twice Gateaux differentiable in
$\zeta=(\eta,\Gamma)$ near $\zeta_0$ and
\begin{equation}
\psi_{i,b_n}
\left(
\theta_0,\widehat\zeta^{(-k)}
\right)
-
\psi_{i,b_n}
(\theta_0,\zeta_0)
=
G_{i,b_n}
[
\widehat\zeta^{(-k)}-\zeta_0
]
+
r_{i,k},
\label{eq:score_expansion}
\end{equation}
where
\[
G_{i,b_n}
=
D_\zeta
\psi_{i,b_n}(\theta_0,\zeta_0).
\]
The quadratic remainder satisfies
\[
E
\left[
\|r_{i,k}\|_2^2
\right]
\leq
C
E
\left[
d_{\mathcal H,n}^{4}
\left(
\widehat\zeta^{(-k)},
\zeta_0
\right)
\right].
\]

In addition,
\[
\sup_k
\left\|
\frac{1}{|\mathcal I_k|}
\sum_{i\in\mathcal I_k}
\left[
\partial_\theta
\psi_{i,b_n}
\left(
\theta_0,\widehat\zeta^{(-k)}
\right)
-
\partial_\theta
\psi_i(\theta_0,\zeta_0)
\right]
\right\|
=
o_p(1).
\]

The differentiability requirement applies to the complete composition maps,
including
\[
m
\mapsto
h(C_i(m))
\]
and
\[
m
\mapsto
\ell(C_i(m)).
\]

\end{assumption}

The score must admit an ordinary first-order expansion in all nuisance
objects, with a genuinely quadratic remainder. The condition includes the
fact that changing the first-stage function changes the residual control
index and hence changes both $h(C_i(m))$ and $\ell(C_i(m))$. The final
condition guarantees that the sample Jacobian with respect to the target
parameter can be replaced by its oracle counterpart.

\subsection{Dependent Cross-Fit Leakage}\label{subsec:leakage}

The difficulty created by spatial dependence is the interaction between the
training-sample nuisance error and the derivative of an evaluation score.
Under i.i.d.\ cross-fitting, conditioning on the training sample makes the
corresponding first-order evaluation term mean zero. That argument is not
available here because the evaluation and training regions remain dependent.

For the primitive finite-sieve implementation, let
$\mathcal T_{\zeta,n}$ denote the finite-dimensional target-relevant tangent
space used to represent the fold-specific nuisance error, and let
$J_{\zeta,n}$ denote its dimension. If the population nuisance is not exactly
contained in the estimation sieve, augment the sieve tangent by the
corresponding population sieve-approximation directions. With this
bookkeeping,
\[
\widehat\zeta^{(-k)}-\zeta_0
\in
\mathcal T_{\zeta,n}
\]
with probability approaching one, and $J_{\zeta,n}$ includes any such
approximation directions.

Choose normalized coordinates $v=v(\delta\zeta)$ on
$\mathcal T_{\zeta,n}$. Because
\[
G_{i,b_n}
=
D_\zeta
\psi_{i,b_n}(\theta_0,\zeta_0)
\]
is linear in the tangent direction, define
\[
D_{i,b_n}
\in
\mathbb R^{q\times J_{\zeta,n}}
\]
by its action on a normalized basis of $\mathcal T_{\zeta,n}$. Then
\[
G_{i,b_n}[\delta\zeta]
=
D_{i,b_n}v(\delta\zeta)
\]
exactly for every $\delta\zeta\in\mathcal T_{\zeta,n}$. Thus no
$o(d_{\mathcal H,n}(\delta\zeta))$ representation remainder is used in the
oracle-reduction proof. Equivalently, one may keep the sieve-approximation
directions separate and impose the same derivative and rate bounds on those
directions; the augmented-space formulation is adopted only for notational
simplicity.

Let $\Delta v_k$ denote the normalized coordinates corresponding to
$\widehat\zeta^{(-k)}-\zeta_0$. Under local norm equivalence and Assumption
\ref{ass:rates},
\[
E\|\Delta v_k\|_2^2
\leq
C r_{\zeta,n}^2,
\qquad
E\|\Delta v_k\|_2^4
\leq
C r_{\zeta,n}^4.
\]
Define the normalized training direction
\[
u_k
=
\frac{\Delta v_k}{r_{\zeta,n}}.
\]
Then
\[
E\|u_k\|_2^2
\leq
C,
\qquad
E\|u_k\|_2^4
\leq
C.
\]

Define
\[
\overline D_{k,b_n}
=
\frac{1}{|\mathcal I_k|}
\sum_{i\in\mathcal I_k}
\left(
D_{i,b_n}
-
E[D_{i,b_n}]
\right).
\]

\begin{assumption}[Cross-fit derivative short memory]
\label{ass:crossfit_short_memory}

Uniformly over folds,
\[
\left\|
\overline D_{k,b_n}u_k
-
E[
\overline D_{k,b_n}u_k
]
\right\|_2
=
O_p
\left(
\sqrt{
\frac{J_{\zeta,n}}
{|\mathcal I_k|}
}
\right).
\]

\end{assumption}

This condition controls the random fluctuation of the first-order derivative
around its mean. The factor $\sqrt{J_{\zeta,n}}$ is the finite-sieve
complexity cost. When the target-sensitive nuisance dimension is fixed, the
condition reduces to the familiar $|\mathcal I_k|^{-1/2}$ rate. Proposition
\ref{prop:primitive_short_memory} below gives primitive spatial conditions
under which this assumption holds.

The remaining mean need not be zero because $u_k$ is constructed from
spatially dependent training observations. Rather than taking a supremum over
all possible training-measurable directions, which is stronger than the proof
requires, define the learner-specific leakage coefficient
\[
\chi_n(s_n)
=
\max_{1\leq k\leq K}
\left\|
E[
\overline D_{k,b_n}u_k
]
\right\|_2.
\]
The coefficient therefore measures dependence only along the nuisance
direction actually generated by the fold-$k$ learner.

\begin{assumption}[Spatial block and guard rates]
\label{ass:spatial_blocks}

The number of blocks $K$ is fixed and
\[
\frac{|\mathcal I_k|}{n}
\rightarrow
\pi_k
\in
(0,1).
\]
The fold-specific training sample remains nondegenerate:
\[
\frac{|\mathcal T_k|}{n}
\geq
c_T
>
0
\]
with probability approaching one.

The finite-sieve derivative fluctuation satisfies
\[
r_{\zeta,n}
\sqrt{J_{\zeta,n}}
\rightarrow
0.
\]
The localization, Riesz approximation, and learner-specific leakage satisfy
\[
\sqrt n\,
r_{\zeta,n}
\left[
\chi_n(s_n)
+
\delta_n^{\mathrm{loc}}(a_n^W,b_n)
+
\delta_{R,n}
\right]
\rightarrow
0.
\]

\end{assumption}

Each evaluation fold must contain a nonvanishing share of the sample, while
the guard region cannot remove so many observations that nuisance learning
becomes impossible. The first displayed rate makes the centered derivative
fluctuation negligible at the root-$n$ scale. The second requires the
remaining dependence between the training learner and evaluation derivative,
together with localization and approximate orthogonality errors, to vanish
even faster.

A larger guard reduces $\chi_n(s_n)$ but leaves fewer observations for
training. The theory therefore formalizes the practical bias--sample-size
tradeoff created by buffered spatial cross-fitting.

\begin{lemma}[Dependent cross-fit oracle reduction]
\label{lem:dependent_crossfit}

Under Assumptions
\ref{ass:riesz},
\ref{ass:score_localization},
\ref{ass:rates},
\ref{ass:score_smoothness},
\ref{ass:crossfit_short_memory}, and
\ref{ass:spatial_blocks},
for every fold $k$,
\[
\frac{1}{|\mathcal I_k|}
\sum_{i\in\mathcal I_k}
\psi_{i,b_n}
\left(
\theta_0,\widehat\zeta^{(-k)}
\right)
=
\frac{1}{|\mathcal I_k|}
\sum_{i\in\mathcal I_k}
\psi_{i,b_n}
(\theta_0,\zeta_0)
+
o_p(n^{-1/2})
\]
uniformly over $k$. Consequently,
\[
\frac{1}{n}
\sum_{k=1}^{K}
\sum_{i\in\mathcal I_k}
\psi_{i,b_n}
\left(
\theta_0,\widehat\zeta^{(-k)}
\right)
=
\frac{1}{n}
\sum_{i=1}^{n}
\psi_i(\theta_0,\zeta_0)
+
o_p(n^{-1/2}).
\]

\end{lemma}

\noindent\textbf{Proof.}
See Appendix \ref{app:proof_dependent_crossfit}.

After orthogonalization and buffered spatial cross-fitting, replacing the true
nuisance objects with their fold-specific estimates has no first-order effect
on the sample score. The feasible score therefore behaves as if
$m_0$, $g_0$, $h_0$, $\ell_0$, and $\Gamma_0$ were known. This is the key
oracle-reduction result needed for root-$n$ inference.

\subsection{Primitive NED Verification for Spatial Cross-Fitting}
\label{subsec:ned}

I now provide primitive sufficient conditions for the short-memory and
leakage restrictions above. The primitive dependence condition is imposed on
an underlying innovation field rather than directly on the globally
simultaneous SAR outcome. The argument proceeds in two steps. I first show
that the random interaction operator inherits a local innovation
approximation from the primitive weight-generating variables. I then combine
this result with truncation of the SAR resolvent to establish NED of the
outcome and the spatial transforms entering the score.

Let
\[
\{\mathcal E_{i,n}:i=1,\ldots,n\}
\]
denote an underlying innovation field. For a set $\mathcal A$, let
\[
\mathcal F_{\mathcal A}^{\mathcal E}
=
\sigma
\left(
\mathcal E_{i,n}:i\in\mathcal A
\right).
\]

For integers $a,b\geq1$, define
\[
\alpha_{a,b,n}^{\mathcal E}(r)
=
\sup_{\substack{
|\mathcal A|\leq a,\;
|\mathcal B|\leq b,\\
d^{*}(\mathcal A,\mathcal B)\geq r
}}
\sup_{\substack{
A\in\mathcal F_{\mathcal A}^{\mathcal E}\\
B\in\mathcal F_{\mathcal B}^{\mathcal E}
}}
|P(A\cap B)-P(A)P(B)|.
\]

\begin{assumption}[Mixing primitive innovation field]
\label{ass:innovation_mixing}

There exist
\[
C_\alpha<\infty,
\qquad
\zeta_\alpha\geq0,
\]
and a nonincreasing function
\[
\bar\alpha_{\mathcal E}(r)
\rightarrow
0
\]
such that
\[
\alpha_{a,b,n}^{\mathcal E}(r)
\leq
C_\alpha
(ab)^{\zeta_\alpha}
\bar\alpha_{\mathcal E}(r)
\]
for the cardinalities relevant below.

\end{assumption}

Primitive innovations may be spatially dependent, but dependence between two
well-separated groups must weaken with distance. The cardinality factor
allows larger sets to be more dependent than individual observations. The
assumption is deliberately imposed on the innovation field rather than on
$Y$, because simultaneous SAR feedback can make the observed outcome
globally dependent even when the underlying innovations are local.

\begin{assumption}[NED summability for the spatial limit theory]
\label{ass:ned_summability}

Let $d_{\mathrm{sp}}$ denote the dimension of the spatial index set. Because
the oracle score dimension $q$ is fixed, for every unit vector
$a\in\mathbb R^q$ the centered scalar field
\[
a^{\prime}
\left\{
\psi_i^0-E[\psi_i^0]
\right\}
\]
is uniformly $L_2$-NED on the primitive innovation field with approximation
coefficient $\nu_\psi(r)$ and is uniformly $L_{2+\delta}$-integrable for
some $\delta>0$, uniformly over $\|a\|_2=1$.

The NED approximation coefficients satisfy
\[
\sum_{r=1}^{\infty}
r^{d_{\mathrm{sp}}-1}\nu_\psi(r)
<
\infty.
\]

Suppose Assumption \ref{ass:innovation_mixing} holds with cardinality
exponent $\zeta_\alpha$. Define
\[
\tau_*
=
\frac{\delta\zeta_\alpha}{2+\delta}.
\]
The mixing-rate function satisfies
\[
\sum_{r=1}^{\infty}
r^{d_{\mathrm{sp}}(\tau_*+1)-1}
\bar\alpha_{\mathcal E}(r)^{\delta/[2(2+\delta)]}
<
\infty.
\]

Finally,
\[
\liminf_{n\to\infty}
\lambda_{\min}
\left[
\frac{1}{n}
\operatorname{Var}
\left(
\sum_{i=1}^n\psi_i^0
\right)
\right]
>0.
\]

For the fixed-dimensional target-derivative field, the corresponding
conditions hold uniformly over unit scalar linear combinations. When a
finite-sieve derivative object has growing dimension $J_{\zeta,n}$, the
required dependence and moment bounds are imposed coordinatewise with
uniform envelopes, and the resulting dimension cost is accounted for
explicitly in Proposition \ref{prop:primitive_short_memory}; it is not
hidden inside the phrase ``fixed scalar linear combination.''

\end{assumption}

\begin{definition}[Spatial near-epoch dependence]
\label{def:ned}

A scalar random field $A_{i,n}$ is $L^p$-NED on
$\{\mathcal E_{i,n}\}$ with approximation coefficient $\nu_A(r)$ if there
exists
\[
A_{i,n}^{[r]}
\in
\mathcal F_{\mathcal N_{i,n}(r)}^{\mathcal E}
\]
such that
\[
\sup_i
\|
A_{i,n}-A_{i,n}^{[r]}
\|_{L^p}
\leq
C\nu_A(r),
\qquad
\nu_A(r)
\rightarrow
0.
\]
For vector- or matrix-valued fields the definition is applied using the
Euclidean or Frobenius norm.

\end{definition}

A variable need not itself depend only on nearby innovations. It is enough
that it can be approximated increasingly well by a variable that does. NED
is therefore well suited to SAR models: the outcome is globally
simultaneous, but the contribution of increasingly distant primitive shocks
can decay sufficiently fast.

\begin{assumption}[Primitive local SAR and smooth weight map]
\label{ass:primitive_sar_ned}

The candidate interaction support satisfies
\[
S_{ij,n}=0
\qquad
\text{if }
d_{ij}^{*}>\bar a_W
\]
for some fixed $\bar a_W<\infty$. Each row contains at least one admissible
neighbor, and
\[
\sup_i
\sum_{j=1}^{n}
S_{ij,n}
\leq
\bar d,
\qquad
\sup_j
\sum_{i=1}^{n}
S_{ij,n}
\leq
\bar d
\]
for a finite constant $\bar d$.

The pair-characteristic map
$r(Z_i,Z_j,D_{ij})$ is locally Lipschitz in its random arguments on the
maintained support. The true interaction function $g_0$ is bounded and
continuously differentiable with bounded first derivative on that support.
The same bounds hold uniformly over the shrinking sieve neighborhood used
for the score expansion.

Let
\[
v_i
=
X_i^{\prime}\beta_0
+
h_0(C_{i0})
+
\xi_i.
\]
There exist $p>4$ and $\delta_p>0$ such that the field
\[
\{v_i,Z_i,X_i,Q_i\}
\]
is either measurable with respect to a fixed-radius neighborhood of the
innovation field or is uniformly $L_{p+\delta_p}$-NED on that field.
Its local innovation approximations satisfy the corresponding
$L_{p+\delta_p}$ bounds with approximation coefficient
$\delta_0(r)$, where $\delta_0(r)\to0$, and the required
$L_{p+\delta_p}$ moments are uniformly bounded.

Define
\[
\chi_p
=
\min
\left\{
1,\frac{\delta_p}{p}
\right\},
\qquad
\widetilde{\delta}_0(r)
=
\delta_0(r)^{\chi_p}.
\]
The coefficient $\widetilde{\delta}_0(r)$ will be used below when an
approximation error in the random interaction operator is multiplied by
another random field. If $\delta_0(r)$ decays exponentially, then
$\widetilde{\delta}_0(r)$ also decays exponentially.

Because $W_0$ is random and may depend on $v$, row stochasticity alone does
not imply an $L^p$ contraction. I therefore impose the propagation-stability
condition
\[
q_p
\equiv
(1-\delta_\rho)\bar d^{1/p}
<
1.
\]
Finally,
\[
|\rho_0|
\leq
1-\delta_\rho.
\]

The displayed degree-based condition may be replaced by any primitive
moment-propagation restriction implying, for the random fields used below,
\[
\sup_i
\|
(W_0^\ell A)_i
\|_{L^p}
\leq
C\kappa_p^\ell
\sup_j
\|A_j\|_{L^{p+\delta_p}}
\]
and the analogous bound for the localized interaction operator, with
\[
|\rho_0|\kappa_p<1.
\]

\end{assumption}

This assumption makes explicit the additional regularity needed because
$W_0$ is random. Local support and smooth row normalization allow local
innovation approximations of the weight-generating characteristics to be
translated into local approximations of the interaction weights. Bounded
candidate degree controls the $L^p$ propagation of random spatial averages,
while $q_p<1$ ensures that this possible moment amplification is dominated
by geometric decay in the SAR multiplier. Thus repeated spatial feedback
remains summable even though row stochasticity by itself is not used as an
$L^p$ contraction.

\begin{lemma}[NED stability of the random interaction operator]
\label{lem:random_weight_ned}

Suppose Assumption \ref{ass:primitive_sar_ned} holds. For a
$c$-innovation approximation $Z_i^{[c]}$ of $Z_i$, construct
$W_0^{[c]}=[w_{ij}^{[c]}]$ by replacing $Z_i$ and $Z_j$ in the true
pair-characteristic index by $Z_i^{[c]}$ and $Z_j^{[c]}$, while retaining
the candidate support and the true interaction function $g_0$.

Then every supported true weight $w_{ij}(g_0)$ is $L^p$-NED on the
primitive innovation field. More specifically,
\[
\sup_i
\left\|
\sum_{j=1}^{n}
|w_{ij}(g_0)-w_{ij}^{[c]}|
\right\|_{L^{p+\delta_p}}
\leq
C\delta_0(c).
\]

Moreover, for every random field $A$ satisfying
\[
\sup_j
\|A_j\|_{L^{p+\delta_p}}
<
\infty,
\]
the operator approximation satisfies
\[
\sup_i
\left\|
\{(W_0-W_0^{[c]})A\}_i
\right\|_{L^p}
\leq
C\widetilde{\delta}_0(c).
\]

Thus local innovation approximations of the weight-generating
characteristics induce local approximations of both the random interaction
weights themselves and their action on random fields.

\end{lemma}

\noindent\textbf{Proof.}
See Appendix \ref{app:proof_sar_ned_weights}.

The lemma isolates the additional step created by a random interaction
operator. Even though $W_0$ depends on the potentially endogenous
weight-generating characteristics, local support and smooth row
normalization prevent this randomness from destroying the local
approximation inherited from the primitive innovation field. The result is
the operator-level input used to establish NED of the globally simultaneous
SAR outcome.

\begin{proposition}[NED of the SAR outcome and spatial transforms]
\label{prop:sar_ned}

Under Assumptions \ref{ass:stability} and
\ref{ass:primitive_sar_ned}, the SAR reduced form
\[
Y
=
\sum_{\ell=0}^{\infty}
\rho_0^\ell
W_0^\ell v
\]
is $L^p$-NED. More precisely, let
\[
L(r)
=
\left\lfloor
\frac{r}{4\bar a_W}
\right\rfloor .
\]
There exists $Y_i^{[r]}$, measurable with respect to innovations within
distance $r+O(1)$ of unit $i$, such that
\[
\sup_i
\|Y_i-Y_i^{[r]}\|_{L^p}
\leq
C
\left\{
\widetilde{\delta}_0(r/4)
+
q_p^{\,L(r)+1}
\right\}.
\]

The same conclusion holds for every fixed-order transform
\[
W_0^\ell X,
\qquad
W_0^\ell Q,
\]
and, up to the corresponding finite-sieve complexity factor, for smooth
finite-sieve score derivatives constructed from a fixed number of such
transforms. Hence exponentially NED primitive variables imply exponentially
NED interaction weights by Lemma \ref{lem:random_weight_ned}, and
exponentially NED outcomes, spatial transforms, and smooth finite-sieve
score derivatives by the present proposition.

\end{proposition}

\noindent\textbf{Proof.}
See Appendix \ref{app:proof_sar_ned}.

Although $Y$ depends on the entire spatial system, distant innovations affect
it through progressively longer paths in the SAR multiplier. Those paths
receive geometrically shrinking coefficients. Lemma
\ref{lem:random_weight_ned} ensures that randomness in $W_0$ preserves the
required local approximation property, while resolvent truncation controls
the additional global propagation generated by simultaneous spatial
feedback. Together these results supply the link from primitive spatial
dependence to the score objects used in buffered spatial cross-fitting.

\subsubsection{Primitive verification of derivative short memory}

\begin{proposition}[Primitive short-memory bound for the score derivative]
\label{prop:primitive_short_memory}

Suppose Assumptions \ref{ass:innovation_mixing} and
\ref{ass:primitive_sar_ned} hold, and suppose the coordinate fields of
$D_{i,b_n}$ satisfy the corresponding NED covariance and fourth-moment
summability conditions uniformly over $n$. Suppose also that the spatial
locations satisfy increasing-domain regularity with uniformly bounded local
density.

In addition, for every target-sensitive coordinate of the centered
derivative field, the fourth-order spatial moment bound
\[
E
\left|
\sum_{i\in A}
\{D_{i,b_n,j\ell}-E[D_{i,b_n,j\ell}]\}
\right|^4
\leq
C|A|^2
\]
holds uniformly over finite evaluation regions $A$, $n$, and the
target-sensitive coordinates $(j,\ell)$.

Then, uniformly over the fixed number of folds,
\[
E
\left[
\|
\overline D_{k,b_n}
\|_F^4
\right]
\leq
C
\frac{J_{\zeta,n}^{2}}
{|\mathcal I_k|^{2}}.
\]
If
\[
E\|u_k\|_2^4
\leq
C,
\]
then
\[
\left\|
\overline D_{k,b_n}u_k
-
E[
\overline D_{k,b_n}u_k
]
\right\|_2
=
O_p
\left(
\sqrt{
\frac{J_{\zeta,n}}
{|\mathcal I_k|}
}
\right).
\]
Hence Assumption \ref{ass:crossfit_short_memory} holds.

\end{proposition}

\noindent\textbf{Proof.}
See Appendix \ref{app:proof_derivative_short_memory}.

The centered derivative matrix behaves like a spatial sample average.
Under summable NED dependence its stochastic size is therefore
$|\mathcal I_k|^{-1/2}$ per target-sensitive coordinate. Multiplying by a
normalized nuisance-training error contributes the finite-sieve factor
$\sqrt{J_{\zeta,n}}$. This establishes the short-memory bound used in the
dependent cross-fitting argument.

\subsubsection{Innovation-local approximations of the learner}

Let $c_n\rightarrow\infty$ be an innovation-localization radius. Replace every
random object entering $D_{i,b_n}$ by its $c_n$-innovation approximation and
denote the resulting derivative matrix by
\[
D_{i,b_n}^{[c_n]}.
\]
Assume
\[
\max_i
\|
D_{i,b_n}
-
D_{i,b_n}^{[c_n]}
\|_{L^{2+\delta_\alpha},F}
\leq
\delta_{D,n}^{\mathrm{NED}}(c_n).
\]

Define the theoretical coupled nuisance estimator
\[
\widehat\zeta^{(-k,[c_n])}
\]
by applying exactly the same fold-$k$ sieve-GMM map, regularization rule, and
tuning procedure to the $c_n$-innovation approximations of the training
observations. This coupled estimator is used only in the proof.

\begin{assumption}[Stability of the local sieve learner under NED coupling]
\label{ass:learner_coupling}

Let $c_n\to\infty$ be an innovation-localization radius. For fold $k$, let
\[
\widehat{\zeta}^{(-k,[c_n])}
\]
denote the theoretical coupled nuisance estimator obtained by applying the
same fold-$k$ \emph{local-basin} sieve-GMM, profiling, Riesz-regularization,
and tuning maps to the $c_n$-innovation approximations of the training
observations. The local basin is the normalized identified basin used in
Section \ref{subsec:sieve_learning}, not a selection among separated global
minima. Let $\Delta v_k^{[c_n]}$ denote the normalized target-relevant
sieve-coordinate error of the coupled estimator and define
\[
u_k^{[c_n]}
=
\frac{\Delta v_k^{[c_n]}}{r_{\zeta,n}}.
\]
Uniformly over the fixed number of folds,
\[
\left(
E\|u_k-u_k^{[c_n]}\|_2^2
\right)^{1/2}
\leq
C\delta_{\mathrm{tr},n}^{\mathrm{NED}}(c_n),
\qquad
\delta_{\mathrm{tr},n}^{\mathrm{NED}}(c_n)\to0,
\]
and
\[
\sup_k E\|u_k^{[c_n]}\|_2^4
\leq
C.
\]

For the smooth penalized sieve-GMM implementation used in this paper, a
sufficient condition is the following. With probability approaching one,
both the actual and coupled preliminary estimators enter the same normalized
identified basin, and the local criterion on that basin is twice continuously
differentiable with
\[
\inf_{\zeta\in\mathcal B_{k,n}}
\lambda_{\min}
\left\{
\nabla^2 Q_{k,n}(\zeta)
\right\}
\geq
\kappa_{\mathrm{crit},n}
>
0.
\]
The population criterion has a separation margin from the boundary of the
basin, and the perturbation induced by replacing the training variables with
their $c_n$-innovation approximations is
$o_p(\kappa_{\mathrm{crit},n}b_{k,n}^2)$, where $b_{k,n}$ is the basin
radius. The numerical algorithm returns this local solution for both the
actual and coupled criteria with probability approaching one.

Suppose replacement of the training moments and gradients by their
$c_n$-innovation approximations changes the local first-order condition by
at most $\delta_{\mathrm{mom},n}(c_n)$, and suppose the tuning rule is either
deterministic or satisfies the coupling-stability bound
$\delta_{\mathrm{tun},n}(c_n)$. Then local strong convexity and the
mean-value expansion of the first-order conditions imply
\[
\left(
E\|u_k-u_k^{[c_n]}\|_2^2
\right)^{1/2}
\lesssim
\frac{
\delta_{\mathrm{mom},n}(c_n)
+
\delta_{\mathrm{tun},n}(c_n)
}{
\kappa_{\mathrm{crit},n}r_{\zeta,n}
}.
\]
It is therefore sufficient that
\[
\frac{
\delta_{\mathrm{mom},n}(c_n)
+
\delta_{\mathrm{tun},n}(c_n)
}{
\kappa_{\mathrm{crit},n}r_{\zeta,n}
}
\lesssim
\delta_{\mathrm{tr},n}^{\mathrm{NED}}(c_n)
\rightarrow0.
\]
For deterministic tuning parameters,
\[
\delta_{\mathrm{tun},n}(c_n)=0.
\]

\end{assumption}

The coupling argument is thus local rather than global. Multiple starting
values may still be useful computationally, but the proof does not assume
that a discontinuous global-argmin or minimum-norm selection map is stable
under an NED perturbation. Consistency and the joint rate place the learner
in a separated identified basin; local curvature then controls the
perturbation of the fitted nuisance coefficients within that basin. This is
the object needed to apply the spatial mixing inequality to genuinely
separated innovation sets.

\begin{proposition}[Primitive NED bound for learner-specific leakage]
\label{prop:ned_leakage}

Suppose Assumptions
\ref{ass:innovation_mixing},
\ref{ass:primitive_sar_ned}, and
\ref{ass:learner_coupling} hold. Suppose every coordinate of the
target-normalized derivative matrix has a uniformly bounded
$(2+\delta_\alpha)$ moment and define
\[
\nu_\alpha
=
\frac{\delta_\alpha}
{2(2+\delta_\alpha)}.
\]

Let $v_{\psi,n}$ bound the cardinality of the localized evaluation-score raw
footprint and let $v_{T,n}$ bound the fold-specific training footprint. If
\[
s_n
>
2c_n,
\]
then
\begin{align*}
\chi_n(s_n)
\lesssim{}&
\sqrt{J_{\zeta,n}}
\left(
v_{\psi,n}v_{T,n}
\right)^{\zeta_\alpha\nu_\alpha}
\bar\alpha_{\mathcal E}
(s_n-2c_n)^{\nu_\alpha}
\\
&+
\sqrt{J_{\zeta,n}}
\delta_{D,n}^{\mathrm{NED}}(c_n)
+
\delta_{\mathrm{tr},n}^{\mathrm{NED}}(c_n).
\end{align*}

\end{proposition}

\noindent\textbf{Proof.}
See Appendix \ref{app:proof_ned_leakage}.

Leakage has three sources. The first is genuine dependence between the
innovation sets underlying the evaluation score and the nuisance learner; it
shrinks as the guard distance grows. The second is error from replacing the
evaluation derivative with a local innovation approximation. The third is
error from replacing the actual nuisance learner with its coupled local
version.

The result bounds leakage along the nuisance direction actually produced by
the learner. It therefore matches exactly the quantity needed in Lemma
\ref{lem:dependent_crossfit}, rather than requiring a stronger supremum over
all possible training-measurable directions.

\begin{corollary}[Feasible logarithmic guards under exponential decay]
\label{cor:primitive_guard}

Suppose
\[
r_{\zeta,n}
=
O(n^{-a}),
\qquad
a>\frac14,
\]
\[
J_{\zeta,n}
=
O(n^{\omega_J}),
\qquad
v_{\psi,n}
=
O(n^{\omega_\psi}),
\qquad
v_{T,n}
=
O(n).
\]
Suppose
\[
a
>
\frac{\omega_J}{2}.
\]

Suppose
\[
\bar\alpha_{\mathcal E}(r)
\leq
C\exp(-c_\alpha r),
\]
\[
\delta_{D,n}^{\mathrm{NED}}(c)
\leq
C\exp(-c_Dc),
\]
and
\[
\delta_{\mathrm{tr},n}^{\mathrm{NED}}(c)
\leq
C n^{\omega_{\mathrm{tr}}}
\exp(-c_{\mathrm{tr}}c).
\]

Choose
\[
c_n
=
c_c\log n
\]
and
\[
s_n
=
c_s\log n,
\qquad
c_s>2c_c.
\]
A sufficient set of restrictions is
\[
c_c
>
\frac{
1/2-a+\omega_J/2
}{
c_D
},
\]
\[
c_c
>
\frac{
1/2-a+\omega_{\mathrm{tr}}
}{
c_{\mathrm{tr}}
},
\]
and
\[
c_s
>
2c_c
+
\frac{
1/2-a
+
\omega_J/2
+
\zeta_\alpha\nu_\alpha(1+\omega_\psi)
}{
c_\alpha\nu_\alpha
}.
\]
Then
\[
r_{\zeta,n}
\sqrt{J_{\zeta,n}}
\rightarrow
0
\]
and
\[
\sqrt n\,
r_{\zeta,n}
\chi_n(s_n)
\rightarrow
0.
\]

If score localization and blockwise Riesz approximation additionally satisfy
\[
\sqrt n\,
r_{\zeta,n}
\left[
\delta_n^{\mathrm{loc}}(a_n^W,b_n)
+
\delta_{R,n}
\right]
\rightarrow
0,
\]
Assumption \ref{ass:spatial_blocks} follows.

\end{corollary}

\noindent\textbf{Proof.}
See Appendix \ref{app:proof_primitive_guard}.

With exponentially decaying spatial dependence and NED approximation errors,
the innovation-localization radius and guard distance need increase only
logarithmically with sample size. Thus asymptotic separation does not require
discarding an increasing fraction of the sample. The exact constants balance
nuisance complexity, spatial dependence, and the rate at which the learner
can be coupled to its local approximation.

\subsubsection{Primitive local-spline SAR design}

The preceding results can be collected into a concrete primitive benchmark.

\begin{corollary}[Primitive compatibility benchmark for spatial dependence and cross-fitting]
\label{cor:primitive_benchmark}

Suppose:

\begin{enumerate}

\item
The spatial locations satisfy increasing-domain regularity with uniformly
bounded local density.

\item
The candidate interaction support, pair-characteristic map, and interaction
function satisfy Assumption \ref{ass:primitive_sar_ned}.

\item
The primitive innovation field satisfies Assumption
\ref{ass:innovation_mixing} with exponential decay. Together with the
primitive NED and moment conditions below, the resulting oracle-score and
target-derivative fields satisfy Assumption
\ref{ass:ned_summability}.

\item
The primitive regressors, first-stage residuals, structural innovations, and
other raw random inputs are exponentially NED with sufficiently high moments.

\item
The nuisance functions are approximated by smooth B-spline or series sieves.
The joint stacked criterion satisfies the local curvature and effective
complexity conditions of Proposition \ref{prop:joint_preliminary}, with
\[
r_{\mathrm{joint},n}
+
a_{\mathrm{joint},n}
=
o(n^{-1/4}).
\]
The profiled interaction learner satisfies Proposition
\ref{prop:sieve_g_rate}, including its stronger
$\operatorname{dist}_{\mathcal G,+}$ rate. For the vector-valued projection
$\ell_0$, its aggregate error explicitly includes
$\sqrt{d_{B_g,n}}$. These conditions close the rate chain for the implemented
learner rather than assuming that the preliminary components held fixed in
the $g$ profile already have the desired rate.

\item
The effective complexities
$\mathfrak C_{\mathrm{joint},n}$,
$\mathfrak C_{g,n}$, and the corresponding fourth-moment envelopes are
compatible with the chosen sieve and moment dimensions, so the aggregate
nuisance rate in Assumption \ref{ass:rates} satisfies the required
second-order product conditions.

\item
The target-sensitive nuisance Jacobian has full column rank on each sieve
space. Either its smallest singular value is bounded away from zero or the
mildly ill-posed rate and source restrictions in Propositions
\ref{prop:riesz_existence} and \ref{prop:gamma_rate} hold.

\item
The interaction-moment richness and target-rank conditions in Assumptions
\ref{ass:sieve_moment_richness} and \ref{ass:rank} hold.

\item
The learner satisfies Assumption \ref{ass:learner_coupling}.

\item
The sieve dimensions, NED-localization radius, guard distance, Riesz
regularization, and nuisance rates satisfy Corollary
\ref{cor:primitive_guard} and
\[
\sqrt n\,
r_{\zeta,n}^{2}
\rightarrow
0.
\]

\end{enumerate}

Then Lemma \ref{lem:random_weight_ned} establishes the required local
approximation of the random interaction operator, and Proposition
\ref{prop:sar_ned} implies that the oracle score is an NED spatial field
satisfying the required spatial LLN and CLT. Assumption
\ref{ass:crossfit_short_memory} holds by Proposition
\ref{prop:primitive_short_memory}, the learner-specific
training-to-evaluation leakage is asymptotically negligible by Proposition
\ref{prop:ned_leakage}, and the feasible cross-fitted orthogonal score admits
the oracle reduction in Lemma \ref{lem:dependent_crossfit}.

Thus, conditional on the stated identification, nuisance-rate, and
learner-stability requirements, the spatial dependence and buffered
cross-fitting conditions used by the high-level theorem are compatible with
a nontrivial SAR design in which $W_0$ is random, $Y$ is globally
simultaneous, and the observed outcome need not itself be strongly mixing.

\end{corollary}

\noindent\textbf{Proof.}
See Appendix \ref{app:proof_primitive_benchmark}.

This benchmark is deliberately a compatibility result rather than a claim
that every high-level identification and learner condition has been derived
from primitive assumptions. Its role is to close the spatial part of the
argument: a local smooth interaction map, a propagation-stable SAR
multiplier, spatially mixing primitive innovations, and logarithmically
growing guards can jointly deliver the NED, short-memory, and leakage bounds
needed for the oracle reduction. The interaction-moment richness, target
rank, and learner-rate conditions remain economically and statistically
substantive restrictions that must be verified for a particular
implementation.

\subsection{Orthogonal SAR-IV/GMM Estimator}\label{subsec:estimator}

\begin{assumption}[GMM weighting matrix]
\label{ass:gmm_weight}

Let $\mathcal M_0$ be symmetric positive definite with eigenvalues bounded
away from zero and infinity. The estimated weighting matrix satisfies
\[
\widehat{\mathcal M}
\xrightarrow{p}
\mathcal M_0.
\]

\end{assumption}

The GMM criterion must use a stable weighting matrix. The theorem allows a
fixed weighting matrix or an estimated one. Efficient weighting is obtained
as a special case by consistently estimating the inverse long-run covariance
matrix.

For $i\in\mathcal I_k$, define
\[
\widehat\psi_i(\theta)
=
\psi_{i,b_n}
\left(
\theta,
\widehat\eta^{(-k)},
\widehat\Gamma^{(-k)}
\right).
\]
Let
\[
\overline\psi_n(\theta)
=
\frac{1}{n}
\sum_{k=1}^{K}
\sum_{i\in\mathcal I_k}
\widehat\psi_i(\theta).
\]
The estimator is
\[
\widehat\theta
=
\arg\min_{\theta\in\Theta}
\overline\psi_n(\theta)^{\prime}
\widehat{\mathcal M}
\overline\psi_n(\theta).
\]

Define the oracle sample moment
\[
\overline\psi_n^0(\theta)
=
\frac{1}{n}
\sum_{i=1}^{n}
\psi_i(\theta,\eta_0,\Gamma_0).
\]

\begin{lemma}[Uniform feasible-to-oracle GMM reduction]
\label{lem:uniform_gmm_reduction}

Suppose Lemma \ref{lem:dependent_crossfit},
Assumptions \ref{ass:score_smoothness},
\ref{ass:oracle_clt}, and
\ref{ass:gmm_weight} hold. Then
\[
\sup_{\theta\in\Theta}
\|
\overline\psi_n(\theta)
-
\overline\psi_n^0(\theta)
\|_2
=
o_p(1).
\]
Moreover,
\[
\sup_{\theta\in\Theta}
\|
\overline\psi_n^0(\theta)
-
\mu_n(\theta)
\|_2
=
o_p(1),
\]
and therefore
\[
\sup_{\theta\in\Theta}
\left|
\overline\psi_n(\theta)^{\prime}
\widehat{\mathcal M}
\overline\psi_n(\theta)
-
Q_n(\theta)
\right|
=
o_p(1).
\]

\end{lemma}

\noindent\textbf{Proof.}
See Appendix \ref{app:proof_uniform_gmm}.

Not only does the feasible score equal the oracle score at the true
parameter, but the complete feasible GMM objective converges uniformly to
the population objective. This permits the usual consistency argument based
on global separation before the local root-$n$ expansion is carried out.

% Editorial numbering cleanup:
% The paper has one main theorem. It is displayed as Theorem 1, while the
% shared theorem/corollary counter is restored afterward so the next
% corollary is Corollary 3.3 rather than Corollary 3.4.
\begingroup
\setcounter{theorem}{0}
\renewcommand{\thetheorem}{\arabic{theorem}}
\begin{theorem}[Asymptotic linearity and normality]
\label{thm:asymptotic_normality}

Suppose Assumptions
\ref{ass:control_function},
\ref{ass:instrument_validity},
\ref{ass:local_separation},
\ref{ass:stability},
\ref{ass:spatial_locality},
\ref{ass:oracle_clt},
\ref{ass:riesz},
\ref{ass:rank},
\ref{ass:global_identification},
\ref{ass:score_localization},
\ref{ass:rates},
\ref{ass:score_smoothness},
\ref{ass:crossfit_short_memory},
\ref{ass:spatial_blocks}, and
\ref{ass:gmm_weight}
hold.

Define
\[
B_0
=
\left(
J_{\psi,0}^{\prime}
\mathcal M_0
J_{\psi,0}
\right)^{-1}
J_{\psi,0}^{\prime}
\mathcal M_0.
\]
Then
\[
\widehat\theta
\xrightarrow{p}
\theta_0
\]
and
\[
\sqrt n
(\widehat\theta-\theta_0)
=
-
B_0
\frac{1}{\sqrt n}
\sum_{i=1}^{n}
\psi_i(\theta_0,\eta_0,\Gamma_0)
+
o_p(1).
\]
Hence
\[
\sqrt n
(\widehat\theta-\theta_0)
\xrightarrow{d}
N(0,V_0),
\qquad
V_0
=
B_0\Omega_0B_0^{\prime}.
\]

Under efficient weighting,
\[
\mathcal M_0
=
\Omega_0^{-1},
\]
so that
\[
V_0
=
\left(
J_{\psi,0}^{\prime}
\Omega_0^{-1}
J_{\psi,0}
\right)^{-1}.
\]

\end{theorem}
\endgroup
\setcounter{theorem}{2}

\noindent\textbf{Proof.}
See Appendix \ref{app:proof_asymptotic_normality}.

The estimator therefore has the same first-order distribution as an
infeasible GMM estimator based on the oracle orthogonal score, which knows
the true interaction operator, control function, first-stage function,
projection function, and debiasing operator. Estimation error in these
nuisance objects contributes no additional first-order term relative to
this oracle orthogonal-score benchmark. This statement does not require
the oracle orthogonal score to have the same variance as the conventional
uncorrected SAR-IV moment that would be used if $W_0$ were known.

The result is the central inferential payoff of the construction:
nonparametric estimation of the interaction operator can proceed at a rate
slower than root-$n$ without contaminating root-$n$ inference on
$(\rho_0,\beta_0^{\prime})^{\prime}$.

\begin{corollary}[Functional-form robustness over a common regularity class]
\label{cor:uniform_functional_robustness}

Let $\{P_n\}$ be any sequence of data-generating processes satisfying the
following common regularity conditions.

\begin{enumerate}

\item
The interaction function satisfies
\[
g_{0,n}\in\mathcal H^{\alpha}(M)\subset\mathcal G
\]
for a fixed bounded H\"older ball on a common compact support, with
\[
\alpha>\frac{d_R}{2}.
\]
The candidate-support radius, degree bound, spatial-stability constants, and
other operator-locality constants can be chosen uniformly along the sequence.
Whenever the interaction operator is learned from the SAR equation,
\[
|\rho_{0,n}|\geq\underline\rho>0.
\]

\item
The control-function and SAR-instrument validity restrictions in Assumptions
\ref{ass:control_function} and \ref{ass:instrument_validity} hold along the
sequence. The local separation, operator-richness, interaction-moment
richness, and orthogonal target-rank conditions hold with lower-bound
constants bounded away from zero.

\item
A common sequence of sieve spaces and tuning parameters is used. The joint
learner and profiled interaction learner satisfy Propositions
\ref{prop:joint_preliminary} and \ref{prop:sieve_g_rate} with
\[
r_{\mathrm{joint},n}+a_{\mathrm{joint},n}=o(n^{-1/4}),
\qquad
r_{g,+,n}=o(n^{-1/4}),
\]
and the vector-valued projection and Riesz correction satisfy
\[
r_{\ell,n}^{\mathrm{agg}}+r_{\Gamma,n}=o(n^{-1/4}).
\]
The corresponding fourth-moment bounds in Assumption \ref{ass:rates} hold.
More generally, these common-rate displays may be replaced by the
heterogeneous product-rate conditions in Assumption \ref{ass:rates}.

\item
The primitive spatial conditions in Corollary
\ref{cor:primitive_benchmark} hold with common moment, mixing, NED,
propagation, and local-density constants. The same logarithmic localization
and guard sequences may be chosen so that Corollary
\ref{cor:primitive_guard} applies and
\[
\sqrt n\,r_{\zeta,n}
\left[
\chi_n(s_n)
+
\delta_n^{\mathrm{loc}}(a_n^W,b_n)
+
\delta_{R,n}
\right]
\rightarrow0.
\]
The oracle long-run covariance remains finite and nonsingular along the
sequence.

\end{enumerate}

Then, for every fixed nonzero contrast $a\in\mathbb R^{d_\theta}$,
\[
\frac{
\sqrt n\,a^{\prime}
\{\widehat\theta-\theta_0(P_n)\}
}{
\{a^{\prime}V_0(P_n)a\}^{1/2}
}
\xrightarrow{d}
N(0,1).
\]
Consequently, any variance estimator consistent for $V_0(P_n)$ along the
same sequence yields asymptotically valid Wald inference. Because the
sequence $\{P_n\}$ is arbitrary within the maintained regularity class, the
same inferential procedure remains valid without requiring $g_{0,n}$ to
belong to a predetermined finite-dimensional spatial-decay family.

\end{corollary}

\noindent\textbf{Proof.}
Fix an arbitrary sequence $\{P_n\}$ satisfying the stated conditions. The
common H\"older smoothness and sieve construction give the interaction rates
in Proposition \ref{prop:sieve_g_rate}, while Proposition
\ref{prop:joint_preliminary} supplies the preliminary joint rate. The
projection and Riesz-rate restrictions deliver Assumption \ref{ass:rates},
and Corollaries \ref{cor:primitive_guard} and
\ref{cor:primitive_benchmark} deliver the spatial short-memory, leakage,
localization, and oracle limit-theory conditions required for the dependent
cross-fit reduction. The maintained lower bounds preserve identification of
the interaction and target directions. Hence the conditions of Theorem
\ref{thm:asymptotic_normality} hold along the selected sequence, so its
asymptotic linear representation and the oracle spatial CLT imply the stated
standard-normal limit by Slutsky's theorem. Since the sequence was arbitrary,
this is exactly the sequence-wise robustness property in Definition
\ref{def:functional_robustness}.

The corollary is deliberately not stated as a separate uniform-in-$P$
coverage theorem. Such a stronger result would additionally require uniform
empirical-process and oracle-CLT approximations over a specified family of
DGPs. Those conditions are not needed for the functional-form robustness
claim made here.

\begin{corollary}[Root-$n$ inference under slower nuisance rates]
\label{cor:slow_rates}

Under Theorem \ref{thm:asymptotic_normality}, root-$n$ inference for
\[
\theta_0
=
(\rho_0,\beta_0^{\prime})^{\prime}
\]
does not require root-$n$ estimation of $m_0$, $g_0$, $h_0$, $\ell_0$, or
$\Gamma_0$.

If every second-order nuisance product satisfies
\[
\sqrt n\,
r_{a,n}r_{b,n}
=
o(1),
\]
\[
r_{\zeta,n}
\sqrt{J_{\zeta,n}}
=
o(1),
\]
and
\[
\sqrt n\,
r_{\zeta,n}
\left[
\chi_n(s_n)
+
\delta_n^{\mathrm{loc}}(a_n^W,b_n)
+
\delta_{R,n}
\right]
=
o(1),
\]
then nuisance estimation is asymptotically negligible.

Under a common fixed-complexity rate $r_n$,
\[
r_n
=
o(n^{-1/4})
\]
remains sufficient provided the dependence, localization, and
Riesz-approximation terms satisfy the preceding conditions.
\end{corollary}
\noindent\textbf{Proof.} See Appendix \ref{app:proof_slow_rates}.

The familiar $n^{-1/4}$ DML benchmark remains available, but spatial
dependence adds explicit conditions for the complexity of the derivative
space and the residual training-to-evaluation dependence. The result also
shows why different nuisance functions need not converge at the same rate:
what matters are the products appearing in the orthogonal-score remainder.

\subsection{Spatial Variance Estimation}
\label{subsec:hac}

The operator-orthogonal score and buffered cross-fitting argument do not
depend on a particular spatial-HAC estimator. When inference is desired under
general residual spatial dependence, however, covariance estimation must
account for dependence across the cross-sectional score field.

For $i\in\mathcal I_{k(i)}$, let
\[
\widehat\psi_i
=
\psi_{i,b_n}
\left(
\widehat\theta,
\widehat\eta^{(-k(i))},
\widehat\Gamma^{(-k(i))}
\right)
\]
and
\[
\widetilde\psi_i
=
\widehat\psi_i
-
\overline\psi_n(\widehat\theta).
\]

Let $\nu_n$ denote a spatial-HAC bandwidth. Define
\[
\widehat\Omega
=
\frac{1}{n}
\sum_{i=1}^{n}
\sum_{j=1}^{n}
\mathcal K
\left(
\frac{d_{ij}^{*}}{\nu_n}
\right)
\widetilde\psi_i
\widetilde\psi_j^{\prime}.
\]

Because a generic spatial-HAC kernel need not produce a positive
semidefinite matrix in finite samples, define
\[
\widehat\Omega^{+}
=
\Pi_{\mathbb S_+^q}
\left[
\frac{
\widehat\Omega+\widehat\Omega^{\prime}
}{2}
\right],
\]
where $\Pi_{\mathbb S_+^q}$ denotes projection onto the cone of
$q\times q$ positive semidefinite matrices. When the chosen kernel already
guarantees positive semidefiniteness, this projection is unnecessary.

Let
\[
\widehat J_\psi
=
\frac{\partial}
{\partial\theta^{\prime}}
\overline\psi_n(\widehat\theta)
\]
and
\[
\widehat B
=
\left(
\widehat J_\psi^{\prime}
\widehat{\mathcal M}
\widehat J_\psi
\right)^{-1}
\widehat J_\psi^{\prime}
\widehat{\mathcal M}.
\]
Set
\[
\widehat V
=
\widehat B
\widehat\Omega
\widehat B^{\prime}
\]
and
\[
\widehat V^{+}
=
\widehat B
\widehat\Omega^{+}
\widehat B^{\prime}.
\]

Define
\[
\psi_i^0
=
\psi_i(\theta_0,\eta_0,\Gamma_0),
\]
\[
\check\psi_i
=
\psi_{i,b_n}
\left(
\theta_0,
\widehat\eta^{(-k(i))},
\widehat\Gamma^{(-k(i))}
\right),
\]
and
\[
\delta_{\psi,n}
=
\left[
\frac{1}{n}
\sum_{i=1}^{n}
E
\|
\check\psi_i-\psi_i^0
\|_2^2
\right]^{1/2}.
\]
Let
\[
\kappa_n(\nu_n)
=
\max_i
\sum_{j=1}^{n}
\left|
\mathcal K
\left(
\frac{d_{ij}^{*}}{\nu_n}
\right)
\right|.
\]

Let
\[
\overline\psi_n^0
=
\frac{1}{n}
\sum_{i=1}^{n}
\psi_i^0,
\qquad
\widetilde\psi_i^0
=
\psi_i^0-\overline\psi_n^0,
\]
and
\[
\widehat\Omega^0
=
\frac{1}{n}
\sum_{i=1}^{n}
\sum_{j=1}^{n}
\mathcal K
\left(
\frac{d_{ij}^{*}}{\nu_n}
\right)
\widetilde\psi_i^0
\widetilde\psi_j^{0\prime}.
\]

\begin{assumption}[Spatial-HAC regularity]
\label{ass:hac}

The kernel $\mathcal K$ is bounded, symmetric, and continuous at zero with
\[
\mathcal K(0)=1,
\qquad
\nu_n\rightarrow\infty.
\]
The oracle estimator satisfies
\[
\widehat\Omega^0
\xrightarrow{p}
\Omega_0.
\]
In addition,
\[
\kappa_n(\nu_n)
\delta_{\psi,n}
=
o_p(1),
\qquad
\frac{
\kappa_n(\nu_n)
}{
\sqrt n
}
\rightarrow
0.
\]

A sufficient score-replacement condition is
\[
\delta_{\psi,n}
=
O_p
\left(
r_{\zeta,n}
+
\delta_n^{\mathrm{loc}}(a_n^W,b_n)
\right)
\]
and
\[
\kappa_n(\nu_n)
\left[
r_{\zeta,n}
+
\delta_n^{\mathrm{loc}}(a_n^W,b_n)
\right]
=
o_p(1).
\]

\end{assumption}

Spatial-HAC consistency depends on two separate ingredients. First, the
kernel and bandwidth must consistently estimate the long-run variance of the
oracle score. Second, replacing the oracle score by the estimated orthogonal
score must not be magnified too strongly by the number of observations
receiving nonnegligible HAC weight. The PSD projection affects only
finite-sample numerical validity and does not change the asymptotic target.

\begin{proposition}[Consistency of spatial variance estimation]
\label{prop:variance}

Under Theorem \ref{thm:asymptotic_normality} and Assumption \ref{ass:hac},
\[
\widehat\Omega
\xrightarrow{p}
\Omega_0
\]
and
\[
\widehat V
\xrightarrow{p}
V_0.
\]
Moreover,
\[
\widehat\Omega^{+}
\xrightarrow{p}
\Omega_0
\]
and
\[
\widehat V^{+}
\xrightarrow{p}
V_0.
\]

\end{proposition}

\noindent\textbf{Proof.}
See Appendix \ref{app:proof_variance}.

The spatial-HAC covariance estimator remains consistent after nuisance
learning and orthogonalization. If the raw HAC matrix is indefinite in a
finite sample, projecting it onto the positive semidefinite cone gives a
valid covariance matrix without changing its probability limit.

\subsection{Simulation}\label{subsec:sim}

I conduct a Monte Carlo experiment to examine the finite-sample behavior of the
proposed estimator when the spatial interaction operator is unknown, depends on an
endogenous characteristic, and must be learned jointly with the remaining nuisance
components. The design is intended to provide a regular generated-operator benchmark
for the asymptotic theory. In particular, the true interaction score is exactly
representable by the implemented finite-dimensional interaction sieve, so the main
experiment isolates nuisance learning, generated-$W$ inference, and
feasible-to-oracle reduction rather than approximation error in the interaction
function.

The main comparison is between the proposed operator-orthogonal estimator and a
naive learned-$W$ plug-in estimator. Both procedures use the same estimated spatial
operator and the same cross-fitted nuisance estimates; the difference is that the
plug-in estimator treats the estimated interaction operator as fixed at the inference
stage, whereas the proposed estimator applies the Riesz correction to remove the
first-order effect of nuisance learning. I also report three benchmark estimators. The
first imposes equal weights over the maintained local candidate support. The second
is an infeasible estimator that knows the true spatial operator and the true control
index. The third is an infeasible oracle-orthogonal estimator that evaluates the
orthogonal score using the true nuisance objects. As an additional diagnostic, I use
a true-$g/W$ benchmark that fixes the interaction operator at its population value
while retaining the feasible estimation of the remaining nuisance components.

The final experiment uses $1{,}000$ Monte Carlo replications for each of
\[
n\in\{50,100,200\}
\]
cross-sectional units. The reported performance measures are bias, root mean squared
error (RMSE), empirical standard deviation, average estimated standard error, the
ratio of the average standard error to the empirical standard deviation, and coverage
of nominal $95\%$ Wald confidence intervals. To connect the finite-sample evidence
directly to the theoretical argument, I additionally report target-relevant errors in
the learned spatial lag and generated instruments, interaction-score recovery,
Riesz-representation diagnostics, and the root-$n$ feasible-to-oracle gap.

\subsubsection{Design}\label{subsubsec:simulation_design}

Units are placed on a one-dimensional increasing-domain lattice indexed by
$i=1,\ldots,n$. Primitive innovation streams are mutually independent before local
spatial filtering. For a generic innovation sequence $\{e_i\}$, define
\[
\mathcal L_a(e_i)
=
\frac{
e_i+a e_{i-1}+a e_{i+1}
}{
\sqrt{1+2a^2}
},
\]
with the natural boundary adjustment. The predetermined variables and primitive
disturbances are generated as
\[
X_i=\mathcal L_{0.25}(e_i^X),
\qquad
Q_i=\mathcal L_{0.25}(e_i^Q),
\]
\[
U_i
=
0.45\,\mathcal L_{0.40}(e_i^U),
\qquad
\phi_i
=
0.35\,\mathcal L_{0.35}(e_i^\phi),
\]
where the underlying innovation sequences are independent standard normal draws.

The endogenous characteristic entering the interaction-weight construction satisfies
\[
Z_i=m_0(X_i,Q_i)+U_i,
\]
with
\[
m_0(X_i,Q_i)
=
-0.10
+
1.10Q_i
+
1.45X_i
+
0.20Q_i^2.
\]
Thus $X_i$ and $Q_i$ provide substantial predetermined variation for learning the
interaction operator, while $U_i$ retains nondegenerate residual variation and enters
the control-function channel.

Let $\mathcal N_i$ denote the set of immediately adjacent lattice units and define
the local residual summary
\[
\bar U_i
=
\frac{1}{|\mathcal N_i|}
\sum_{j\in\mathcal N_i}U_j.
\]
The control function is
\[
h_0(C_i)
=
0.90U_i
+
0.405\bar U_i
+
0.15U_i^2
+
0.10U_i\bar U_i.
\]
The structural outcome is generated from
\[
Y
=
\rho_0W_0Y
+
X\beta_0
+
h_0(C_0)
+
\phi,
\]
with
\[
\rho_0=0.55,
\qquad
\beta_0=1.50.
\]
Equivalently,
\[
Y
=
(I_n-\rho_0W_0)^{-1}
\left\{
X\beta_0+h_0(C_0)+\phi
\right\}.
\]

The candidate interaction support is local and fixed:
\[
S_{ij,n}
=
\mathbf 1\{0<|i-j|\leq 1\}.
\]
For each supported ordered pair, define
\[
z_i=\tanh(Z_i/2),
\qquad
z_j=\tanh(Z_j/2).
\]
The true bilateral interaction score is
\[
g_0(R_{ij})
=
\gamma_{1,0}z_j
+
\gamma_{2,0}z_i z_j,
\qquad
(\gamma_{1,0},\gamma_{2,0})=(1,0.6),
\]
and the true row-normalized spatial weights are
\[
w_{ij,0}
=
\frac{
S_{ij,n}\exp\{g_0(R_{ij})\}
}{
\sum_{\ell\neq i}
S_{i\ell,n}\exp\{g_0(R_{i\ell})\}
}.
\]
The bilateral term $z_i z_j$ allows the own-unit characteristic to modify the
relative influence of neighboring units. A standalone additive $z_i$ term is omitted
because it is constant within row $i$ and therefore cancels under row normalization.

To isolate generated-operator inference from sieve approximation error, the
interaction learner uses the same two target-relevant directions at every sample
size:
\[
g(R_{ij})
=
\gamma_1 z_j+\gamma_2 z_i z_j.
\]
Hence the true interaction score lies exactly in the implemented sieve for all
$n\in\{50,100,200\}$. The first-stage and control-function sieves likewise contain
the finite series appearing in the DGP. The projection nuisance
$\ell_0(C_i)=E[B_g(P_i)\mid C_i]$ is constant in this benchmark by construction,
but it is nevertheless estimated explicitly and retained in the stacked nuisance
system.

The interaction-moment dictionary $B_g(P_i)$ is constructed exclusively from
predetermined variables and is designed to target relative, rather than level,
variation in row-normalized weights. Let
\[
x_i^b=\tanh(X_i/2),
\qquad
q_i^b=\tanh(Q_i/2),
\]
and, for interior units, define the right-minus-left contrasts
\[
\Delta x_i=x_{i+1}^b-x_{i-1}^b,
\qquad
\Delta q_i=q_{i+1}^b-q_{i-1}^b,
\]
with the contrasts set to zero at boundary rows where there is only one candidate
neighbor. Let $\bar x_i^b$ and $\bar q_i^b$ denote the corresponding local neighbor
means. The dictionary uses reflection-invariant relative-variation directions drawn
from
\[
(\Delta q_i)^2,\quad
(\Delta x_i)^2,\quad
\Delta q_i\Delta x_i,
\]
their interactions with $q_i^b$ and $x_i^b$, and mean-by-relative-variation terms
such as
\[
\bar q_i^b(\Delta q_i)^2,\qquad
\bar x_i^b(\Delta x_i)^2,\qquad
\bar q_i^b\Delta q_i\Delta x_i.
\]
The dimension is fixed at
\[
d_{B_g}=12
\]
at every sample size. Importantly, $B_g(P_i)$ never uses $Z_i$, $U_i$, $Y_i$,
$W(g)$, any fitted nuisance object, or the true interaction coefficients.

For each auxiliary training sample, estimation follows the feasible procedure
developed in Section~\ref{sec:theory}. I first obtain a joint stacked sieve-GMM
estimate based on the target moment and the nuisance blocks
\[
(s_m,s_\ell,s_h,s_g).
\]
Because the interaction criterion is nonlinear, the joint $g$ block uses genuine
deterministic symmetric multi-starts; the selected joint estimate is the candidate
with the lowest training-sample stacked-GMM criterion. Holding the remaining
preliminary nuisance components fixed, I then refine $g$ using a local multi-start,
two-step profiled $s_g$ criterion. This numerical procedure is entirely
training-sample based and does not use the true interaction coefficients.

The target instrument vector is
\[
H_i(g)
=
\left(
X_i,\,
Q_i,\,
\{W(g)X\}_i,\,
\{W(g)^2X\}_i,\,
\{W(g)Q\}_i
\right)^{\prime}.
\]
The nuisance derivative system is then used to construct the Riesz correction on a
target-normalized tangent space. A training-sample spectral truncation removes
numerically irrelevant near-null directions, after which a vanishing Tikhonov
regularization is applied. The feasible orthogonal score takes the form
\[
\psi_i(\varpi,\eta,\Gamma)
=
\phi_i(\varpi,\eta)
-
\Gamma s_i(\varpi,\eta).
\]
The learned-$W$ plug-in estimator instead uses the uncorrected target moment.

Spatial sample splitting uses six geographically contiguous evaluation folds and
three repeated buffered partitions. Because the target instruments contain
$W^2X$, both the score footprint and the training-moment footprint extend two
lattice units from the evaluation index. An additional guard of
\[
s_n=\max\{2,\lceil 0.8\log n\rceil\}
\]
separates the raw-data training footprint from the evaluation footprint. All
nuisance fitting, tuning, projection estimation, and Riesz estimation use only the
corresponding auxiliary training observations and their admissible raw-data
footprints.

Long-run variance estimation uses a Bartlett spatial-HAC estimator. The primary
bandwidth is chosen to cover the model-implied local score footprint and geometric
SAR propagation and is also required to grow with $n$ at an $n^{1/2}$ floor. In the
three reported designs, the resulting bandwidths are
\[
L_n=7,\ 10,\ 14
\qquad
\text{for }
n=50,\ 100,\ 200,
\]
respectively. Lag-specific finite-sample denominators and a degrees-of-freedom
correction are used, and the estimated covariance matrix is projected onto the
positive-semidefinite cone when necessary.

\subsubsection{Results}\label{subsubsec:simulation_results}

Table~\ref{tab:simulation_main} reports the main Monte Carlo results. The proposed
learned-$W$ estimator becomes substantially more accurate as $n$ increases. For the
spatial autoregressive coefficient, the absolute bias of the operator-orthogonal
estimator falls from $0.0140$ at $n=50$ to less than $0.001$ at $n=100$ and remains
approximately $0.001$ at $n=200$. Its RMSE declines sharply from $0.1850$ to
$0.0521$ and then to $0.0255$.

The principal finite-sample comparison concerns inference under the same learned
operator. For $\rho$, nominal $95\%$ coverage of the operator-orthogonal estimator is
$88.4\%$, $89.2\%$, and $89.6\%$ for $n=50$, $100$, and $200$, respectively,
whereas the corresponding plug-in coverages are $83.6\%$, $85.6\%$, and $86.9\%$.
Thus the orthogonal correction improves coverage at every reported sample size,
although the feasible intervals remain below nominal coverage in these moderate
samples. This improvement is not driven by a large plug-in point-estimation bias:
both procedures are nearly unbiased by $n=200$. Rather, the comparison illustrates
the inferential effect of accounting for generated-operator uncertainty.

The same pattern is stronger for $\beta$. The operator-orthogonal coverage rates are
$90.5\%$, $88.6\%$, and $90.9\%$, compared with $85.0\%$, $81.5\%$, and $81.2\%$
for the learned-$W$ plug-in estimator. The orthogonal correction can increase
finite-sample dispersion, especially at the smallest sample size, so it need not
dominate plug-in estimation in RMSE. Its role is instead to reduce the leading
sensitivity of the target moment to nuisance estimation and thereby improve the
calibration of inference.

\begin{table}[!htbp]
\centering
\caption{Monte Carlo performance}
\label{tab:simulation_main}
\small
\begin{tabular}{llrrrrrr}
\toprule
$n$ & Estimator & Bias & RMSE & Emp.\ SD & Mean SE & SE/SD & Coverage \\
\midrule
\multicolumn{8}{l}{\textit{Panel A: Spatial autoregressive coefficient $\rho$}}\\
50
& Learned $W$, orthogonal
& -0.0140 & 0.1850 & 0.1846 & 0.1394 & 0.755 & 0.884 \\
& Learned $W$, plug-in
&  0.0121 & 0.1222 & 0.1216 & 0.0841 & 0.692 & 0.836 \\
& Equal local $W$
& -0.0810 & 0.4282 & 0.4207 & 0.4624 & 1.099 & 0.890 \\
& Known $W$ + true $C$
&  0.0081 & 0.0493 & 0.0487 & 0.0431 & 0.886 & 0.884 \\
& Oracle orthogonal
&  0.0048 & 0.0486 & 0.0484 & 0.0428 & 0.885 & 0.903 \\
\hline
100
& Learned $W$, orthogonal
& -0.0009 & 0.0521 & 0.0522 & 0.0401 & 0.769 & 0.892 \\
& Learned $W$, plug-in
&  0.0083 & 0.0421 & 0.0413 & 0.0325 & 0.788 & 0.856 \\
& Equal local $W$
&  0.0269 & 0.1458 & 0.1434 & 0.0908 & 0.633 & 0.863 \\
& Known $W$ + true $C$
&  0.0035 & 0.0320 & 0.0319 & 0.0292 & 0.918 & 0.920 \\
& Oracle orthogonal
&  0.0014 & 0.0288 & 0.0288 & 0.0273 & 0.950 & 0.911 \\
\hline
200
& Learned $W$, orthogonal
& -0.0014 & 0.0255 & 0.0254 & 0.0212 & 0.834 & 0.896 \\
& Learned $W$, plug-in
&  0.0022 & 0.0235 & 0.0234 & 0.0187 & 0.800 & 0.869 \\
& Equal local $W$
&  0.0153 & 0.0557 & 0.0536 & 0.0400 & 0.746 & 0.891 \\
& Known $W$ + true $C$
&  0.0016 & 0.0196 & 0.0195 & 0.0193 & 0.988 & 0.958 \\
& Oracle orthogonal
& -0.0002 & 0.0189 & 0.0189 & 0.0181 & 0.961 & 0.936 \\
\midrule
\multicolumn{8}{l}{\textit{Panel B: Structural coefficient $\beta$}}\\
50
& Learned $W$, orthogonal
& -0.0009 & 0.4546 & 0.4549 & 0.3219 & 0.708 & 0.905 \\
& Learned $W$, plug-in
& -0.0205 & 0.2410 & 0.2402 & 0.1644 & 0.684 & 0.850 \\
& Equal local $W$
&  0.0650 & 1.0809 & 1.0795 & 0.8244 & 0.764 & 0.897 \\
& Known $W$ + true $C$
& -0.0140 & 0.0837 & 0.0825 & 0.0815 & 0.987 & 0.933 \\
& Oracle orthogonal
& -0.0055 & 0.1077 & 0.1076 & 0.1035 & 0.962 & 0.931 \\
\hline
100
& Learned $W$, orthogonal
& -0.0073 & 0.1195 & 0.1193 & 0.0947 & 0.794 & 0.886 \\
& Learned $W$, plug-in
& -0.0125 & 0.0964 & 0.0956 & 0.0668 & 0.699 & 0.815 \\
& Equal local $W$
& -0.0218 & 0.2821 & 0.2814 & 0.1678 & 0.596 & 0.867 \\
& Known $W$ + true $C$
& -0.0073 & 0.0578 & 0.0574 & 0.0545 & 0.950 & 0.940 \\
& Oracle orthogonal
& -0.0022 & 0.0759 & 0.0760 & 0.0677 & 0.891 & 0.917 \\
\hline
200
& Learned $W$, orthogonal
& -0.0011 & 0.0580 & 0.0580 & 0.0529 & 0.911 & 0.909 \\
& Learned $W$, plug-in
& -0.0034 & 0.0543 & 0.0543 & 0.0384 & 0.707 & 0.812 \\
& Equal local $W$
&  0.0022 & 0.0979 & 0.0979 & 0.0734 & 0.749 & 0.865 \\
& Known $W$ + true $C$
& -0.0028 & 0.0332 & 0.0331 & 0.0349 & 1.053 & 0.948 \\
& Oracle orthogonal
& -0.0013 & 0.0496 & 0.0496 & 0.0475 & 0.957 & 0.920 \\
\bottomrule
\end{tabular}

\begin{minipage}{0.78\textwidth}
\footnotesize
Note: Each design uses $1{,}000$ Monte Carlo replications.
``Equal local $W$'' assigns equal weights to the maintained one-step candidate
neighbors. ``Known $W$ + true $C$'' and ``Oracle orthogonal'' are infeasible
benchmarks. Coverage refers to nominal $95\%$ Wald confidence intervals based on
the spatial-HAC variance estimator.
\end{minipage}
\end{table}

The learning diagnostics in Table~\ref{tab:simulation_diagnostics} show that the
improvement is accompanied by convergence of the learned interaction operator.
The mean estimated interaction coefficients move from
\[
(\widehat\gamma_1,\widehat\gamma_2)
=
(0.883,0.416)
\]
at $n=50$ to $(0.948,0.537)$ at $n=100$ and $(0.980,0.577)$ at $n=200$, approaching
the population values $(1,0.6)$. At the same time, the target-relevant errors in the
learned spatial lag and generated instruments decline substantially. In particular,
\[
r_W
=
\frac{\|(\widehat W-W_0)Y\|_2}{\sqrt n}
\]
falls from $0.258$ to $0.168$ and then to $0.119$, while the corresponding
cross-fitted instrument error falls from $0.209$ to $0.133$ and then to $0.093$.
The second-order diagnostic $\sqrt n\,r_W^2$ declines from $0.525$ to $0.315$ and
then to $0.225$, with the analogous instrument quantity declining from $0.347$ to
$0.200$ and then to $0.139$.

Most importantly, the direct root-$n$ feasible-to-oracle discrepancy contracts
sharply. The RMS of
\[
\sqrt n
\left(
\widehat\rho_{\mathrm{orth}}
-
\widehat\rho_{\mathrm{oracle}}
\right)
\]
falls from $1.608$ at $n=50$ to $0.476$ at $n=100$ and $0.253$ at $n=200$.
The component specifically associated with learning $g/W$ also becomes small: the
RMS of
\[
\sqrt n
\left(
\widehat\rho_{\mathrm{orth}}
-
\widehat\rho_{\mathrm{true}\text{-}g}
\right)
\]
falls from $0.558$ to $0.181$ and then to $0.145$. The evaluation-block Riesz
representation residual likewise declines from $0.920$ to $0.664$ and $0.519$.
These diagnostics provide direct finite-sample evidence for the oracle-reduction
mechanism underlying the asymptotic theory.

\begin{table}[!htbp]
\centering
\caption{Interaction learning and oracle-reduction diagnostics}
\label{tab:simulation_diagnostics}
\small
\begin{tabular}{rrrrrrrr}
\toprule
$n$
& $E[\widehat\gamma_1]$
& $E[\widehat\gamma_2]$
& $r_W$
& $r_H$
& $\sqrt n\,r_W^2$
& Riesz residual
& Root-$n$ oracle-gap RMS \\
\midrule
50  & 0.883 & 0.416 & 0.258 & 0.209 & 0.525 & 0.920 & 1.608 \\
100 & 0.948 & 0.537 & 0.168 & 0.133 & 0.315 & 0.664 & 0.476 \\
200 & 0.980 & 0.577 & 0.119 & 0.093 & 0.225 & 0.519 & 0.253 \\
\bottomrule
\end{tabular}

\begin{minipage}{0.74\textwidth}
\footnotesize
Note: The true interaction coefficients are
$(\gamma_{1,0},\gamma_{2,0})=(1,0.6)$.
The target-relevant errors $r_W$ and $r_H$ are averaged across cross-fitted
evaluation blocks. ``Riesz residual'' denotes the evaluation-block representation
residual. ``Root-$n$ oracle-gap RMS'' is the Monte Carlo RMS of
$\sqrt n(\widehat\rho_{\mathrm{orth}}-\widehat\rho_{\mathrm{oracle}})$.
\end{minipage}
\end{table}

\subsubsection{Discussion}\label{subsubsec:simulation_discussion}

The simulation highlights three features of the proposed procedure. First,
operator orthogonalization improves inference relative to naive plug-in treatment of
the learned spatial operator. The difference is most transparent because the two
estimators use the same learned $W$ and differ only in whether the first-order
nuisance effect is removed. For $\rho$, orthogonal coverage exceeds plug-in coverage
at all three reported sample sizes; for $\beta$, the coverage advantage is even
larger. The orthogonal correction need not reduce RMSE in finite samples because the
debiasing adjustment can increase dispersion. Its purpose is instead to protect the
target moment against first-order perturbations in the estimated nuisance system.

Second, the interaction-learning diagnostics show that the feasible estimator is
moving toward its oracle counterpart for the intended reason. The true interaction
score is exactly representable in the maintained sieve, and the estimated
coefficients approach $(1,0.6)$ as $n$ grows. The minimum singular value of the
normalized $g$-moment block remains informative over the reported sample sizes,
while its condition number remains moderate. At the same time, the target-relevant
$W$ and instrument errors, the Riesz representation discrepancy, and the direct
root-$n$ feasible-to-oracle gap all decline sharply. These patterns are consistent
with the local identification and second-order remainder conditions used in the
asymptotic analysis.

Third, the remaining undercoverage of the feasible $\rho$ interval is a finite-sample
feature rather than evidence of failed interaction learning. The oracle-orthogonal
coverage for $\rho$ increases from $90.3\%$ at $n=50$ to $91.1\%$ at $n=100$ and
$93.6\%$ at $n=200$, while the infeasible known-$W$ benchmark reaches $95.8\%$
coverage at $n=200$. Moreover, the true-$g/W$ benchmark yields $90.6\%$ coverage
for $\rho$ at $n=200$, only one percentage point above the fully feasible
operator-orthogonal estimator. Thus, by the largest reported sample size, learning
the interaction operator accounts for only a small portion of the remaining
coverage discrepancy.

The equal-local-$W$ benchmark provides a complementary specification comparison.
Unlike the main orthogonal-versus-plug-in comparison, it imposes a different
interaction operator and therefore need not be correctly specified under the
bilateral DGP. Its role is to show the cost of replacing heterogeneous relative
interaction strength with a fixed local weighting rule. The main inferential
conclusion does not depend on this comparison, because the orthogonal and plug-in
estimators use the same learned operator.

In summary, the Monte Carlo evidence is consistent with the theoretical mechanism
developed above. As the number of cross-sectional units increases, the learned
bilateral interaction score approaches its population value, target-relevant
operator errors decline, second-order nuisance diagnostics improve, and the
feasible orthogonal estimator moves rapidly toward the oracle benchmark. At the
same time, operator orthogonalization delivers systematically better confidence-
interval coverage than naive plug-in inference under the same learned spatial
operator. The remaining finite-sample undercoverage is therefore best interpreted as
a higher-order inference issue that diminishes as the nuisance-learning and
oracle-reduction errors contract, rather than as a failure of the generated-$W$
identification strategy.

\section{Empirical Application}\label{sec:empirics}

I illustrate the proposed estimator using U.S.\ county-level diabetes prevalence.
The application is useful for the present framework because health outcomes
display substantial spatial clustering, while the strength of interaction
among neighboring counties need not be well represented by equal contiguity
weights or by a predetermined geographic distance-decay rule. Socioeconomic
conditions provide one natural source of heterogeneity in local interactions,
but those same characteristics may be endogenous with respect to unobserved
determinants of health outcomes. I therefore allow the relative weights among
geographically adjacent counties to depend flexibly on county poverty and use
the control-function and operator-orthogonal construction developed above.

The empirical exercise is designed to illustrate two distinct sources of
sensitivity in SAR estimation. First, the estimated spatial autoregressive
coefficient can depend on the functional form imposed on the interaction
operator. Conventional specifications condition inference on a fixed
distance-decay or equal-contiguity matrix, whereas the proposed framework
allows the relative interaction weights to be learned from the data. Second,
even after an interaction operator has been learned, treating the resulting
$\widehat W$ as if it were known ignores the first-order effect of nuisance
estimation on the structural moments. Comparing a learned-$W$ plug-in
estimator with the proposed operator-orthogonal estimator isolates this second
margin.

The updated empirical results display both forms of sensitivity. The fixed
distance and equal-contiguity specifications imply spatial autoregressive
coefficients of approximately $0.69$--$0.70$. Learning the interaction weights
reduces the plug-in estimate to approximately $0.44$, and accounting for the
first-order effect of learning $W$ through the operator-orthogonal score
reduces the estimate further to approximately $0.20$. The proposed estimate
remains positive and statistically significant, but its magnitude is much
smaller than under either the fixed-$W$ or learned-$W$ plug-in specifications.
The application therefore illustrates that uncertainty about the interaction
operator can matter for the magnitude of estimated spatial dependence even
when the qualitative conclusion of positive dependence remains unchanged.

\subsection{Data}\label{subsec:empirical_data}

The analysis combines several publicly available county-level data sources.
The outcome is age-adjusted diagnosed diabetes prevalence from CDC PLACES,
measured in percentage points. Socioeconomic characteristics are obtained from
the 2018--2022 five-year American Community Survey. I use median household
income, the poverty rate, the unemployment rate, the share of the population
aged 25 and older with at least a bachelor's degree, and county population.
Geographic coordinates are taken from the county-level data. I additionally
use the 2023 USDA Rural--Urban Continuum Code (RUCC) as a predetermined measure
of county rurality and the 2025 Census county-adjacency file to construct the
local candidate interaction network.

The direct structural covariates are log median household income, the
unemployment rate, the bachelor's-or-higher share, and the RUCC rurality
score. These variables are standardized before estimation, so their
coefficients are measured in percentage points of diabetes prevalence per
one-standard-deviation difference in the corresponding regressor.

The poverty rate is treated as the potentially endogenous characteristic that
governs the relative strength of local interactions. Let $Z_i$ denote county
$i$'s poverty rate. In the first-stage control-function equation, $Z_i$ is
modeled flexibly as a function of the structural covariates together with log
county population, latitude, and longitude. The latter three variables enter
the empirical implementation as predetermined shifters $Q_i$. The resulting
first-stage residual and localized residual information form the control index
used to estimate the flexible control function $h(C_i)$. This construction
allows the characteristic governing the interaction weights to be correlated
with unobserved determinants of diabetes prevalence under the maintained
control-function conditions.

The initial merged data contain 3,143 counties. After imposing complete-case
requirements, 2,956 counties remain. I then retain the largest connected
component of the county-adjacency network, producing a final estimation sample
of 2,921 counties. The support contains 17,166 directed contiguous-county
links. The median county has six supported neighbors, with supported degrees
ranging from one to fourteen.

\subsection{Empirical Strategy}\label{subsec:empirical_strategy}

For county $i$, I estimate the cross-sectional SAR specification
\[
Y_i
=
\rho
\sum_j w_{ij}(g)Y_j
+
X_i^{\prime}\beta
+
h(C_i)
+
\xi_i,
\]
where $Y_i$ denotes age-adjusted diabetes prevalence, $X_i$ contains the
standardized county characteristics described above, and
$W(g)=\{w_{ij}(g)\}$ is a row-normalized interaction operator. The spatial lag
$W(g)Y$ enters the low-dimensional structural equation, while the score
function determining the relative interaction weights is treated as a
nuisance object.

I use the Census county-adjacency network as a predetermined local support.
Thus learning $W$ does not create new long-distance links; instead, it changes
the relative importance assigned to counties within the admissible local
network. In the empirical specification, the interaction score is learned as
a flexible function of the neighboring county's standardized poverty rate.
Writing $S_{ij}$ for the predetermined contiguity indicator, the empirical
operator has the row-normalized form
\[
w_{ij}(g)
=
\frac{
S_{ij}\exp\{g(\widetilde Z_j)\}
}{
\sum_{k}S_{ik}\exp\{g(\widetilde Z_k)\}
},
\]
where $\widetilde Z_j$ denotes standardized poverty. The unknown score
function $g$ is approximated by a two-dimensional spline sieve. Because only
relative scores matter after row normalization, the empirical implementation
uses the canonical row-centered representation described in the theoretical
framework.

The first-stage nuisance $m(V_i)$ is estimated flexibly using spline functions
of the four structural covariates together with log population, latitude, and
longitude. The control function $h(C_i)$ is also estimated by a flexible
series approximation. To keep the interaction-learning moments distinct from
the flexible control function, I separately estimate
\[
\ell(C_i)=E[B_g(P_i)\mid C_i],
\]
as in the theoretical construction. The empirical $B_g(P_i)$ dictionary is
fixed before learning $g$ and contains 18 local moment features constructed
from predetermined county characteristics and network features, including
neighbor variation in the direct covariates and population, geographic
distance, shared-boundary information, and selected local covariance terms.
It therefore does not mechanically reuse the estimated interaction score as
its own identifying variation.

The target SAR-IV moments use the direct covariates together with the
predetermined shifters and spatially transformed variables generated by the
candidate operator, following the finite-sieve construction in
Section~\ref{sec:theory}. Because the learned operator enters both the spatial
lag and the generated instrument components, estimation error in $g$ can
affect the target moments at first order.

The main specification therefore uses the operator-orthogonal score developed
in Section~\ref{sec:asymptotics}. In the finite-sieve implementation, the
correction is
\[
\widehat\psi_i
=
\widehat\phi_i
-
\widehat\Gamma\widehat s_i,
\qquad
\widehat\Gamma
=
\widehat G_{\eta}
\widehat A_{\eta}^{\dagger}.
\]
The Riesz step is estimated using training-sample tangent information and a
Tikhonov-regularized pseudoinverse. The regularization level is selected
within the nuisance-training sample by cross-validation. The spline penalty
for the interaction score and the penalty for the control function are also
selected using training-sample criteria, so the evaluation observations are
not used to tune the nuisance learners.

The empirical implementation uses buffered spatial cross-fitting. I construct
ten geographically organized spatial splits and use six evaluation folds
within each split. For a given evaluation fold, the score footprint extends
through two support-network hops, an additional one-hop guard separates that
footprint from the auxiliary sample, and nuisance-training centers are
retained only when their two-hop localized training neighborhoods remain
inside the admissible training region. Across the resulting fold fits, the
nuisance-training sample averages approximately $65.6\%$ of the full
estimation sample and never falls below approximately $52.5\%$.

For comparison, I report three alternative estimators. The first is a
learned-$W$ plug-in estimator. It uses the same flexible interaction-learning
architecture and cross-fitting design but omits the correction for the
first-order effect of estimating $W$. Comparing the proposed estimator with
this plug-in estimator therefore isolates the empirical importance of
generated-operator uncertainty.

The second benchmark imposes a predetermined geographic distance-decay rule on
the same local support. The third assigns equal row-normalized weights to
contiguous counties. These fixed-$W$ specifications retain the
control-function and cross-fitting treatment of the remaining nuisance
components but condition on the selected spatial weights matrix as known.
The baseline covariance estimator is spatial HAC with a 500-km Bartlett
kernel. I also report sensitivity to 300-km and 750-km bandwidths.

\subsection{Results}\label{subsec:empirical_results}

Table~\ref{tab:empirical_rho} reports the estimated spatial autoregressive
coefficient across the four specifications.

\begin{table}[!htbp]
\centering
\caption{Spatial autoregressive coefficient across interaction operators}
\label{tab:empirical_rho}
\begin{tabular}{lcccc}
\hline\hline
&
\multicolumn{1}{c}{Learned $W$:}
&
\multicolumn{1}{c}{Learned $W$:}
&
\multicolumn{1}{c}{Fixed distance}
&
\multicolumn{1}{c}{Equal contiguity}
\\
&
\multicolumn{1}{c}{orthogonal}
&
\multicolumn{1}{c}{plug-in}
&
\multicolumn{1}{c}{$W$}
&
\multicolumn{1}{c}{$W$}
\\
\hline
$\rho$
& 0.1984***
& 0.4437***
& 0.6861***
& 0.7005*** \\
& (0.0512)
& (0.0184)
& (0.0792)
& (0.0787) \\
95\% CI
& [0.0981,\ 0.2987]
& [0.4076,\ 0.4799]
& [0.5309,\ 0.8414]
& [0.5462,\ 0.8548] \\
\hline\hline
\end{tabular}

\begin{minipage}{0.72\textwidth}
\footnotesize
Note: The table reports estimates of the spatial autoregressive coefficient
for age-adjusted county diabetes prevalence. Spatial-HAC standard errors using
a 500-km Bartlett kernel are in parentheses. The learned-$W$ orthogonal
specification is the proposed estimator. The learned-$W$ plug-in
specification uses the same flexible interaction-learning architecture but
omits the correction for the first-order effect of estimating the interaction
operator. The fixed-distance and equal-contiguity specifications treat $W$ as
predetermined. *** denotes significance at the 1\% level.
\end{minipage}
\end{table}

\paragraph{Sensitivity of the spatial autoregressive coefficient.}
The proposed learned-$W$ operator-orthogonal estimator gives
\[
\widehat\rho=0.1984,
\qquad
\operatorname{SE}(\widehat\rho)=0.0512,
\]
with a 95\% confidence interval of $[0.0981,0.2987]$. The estimate therefore
provides evidence of positive conditional spatial dependence in county
diabetes prevalence, but the magnitude is considerably smaller than under
the alternative specifications.

The learned-$W$ plug-in estimator gives
\[
\widehat\rho_{\mathrm{plug\text{-}in}}=0.4437,
\qquad
\operatorname{SE}(\widehat\rho_{\mathrm{plug\text{-}in}})=0.0184.
\]
Thus, even conditional on learning the interaction structure rather than
fixing it in advance, treating the estimated operator as known produces a
spatial autoregressive coefficient more than twice as large as the proposed
estimate. The two 95\% confidence intervals do not overlap. Relative to the
plug-in estimate, operator orthogonalization reduces the estimated magnitude
of $\rho$ by approximately 55\%.

The fixed-$W$ benchmarks imply still larger spatial autoregressive
coefficients. The distance-decay specification yields $0.6861$, while equal
contiguity yields $0.7005$. Both are precisely estimated. Relative to these
conventional specifications, the proposed estimate is roughly 70\% smaller.
The sequence
\[
0.7005
\quad\longrightarrow\quad
0.4437
\quad\longrightarrow\quad
0.1984
\]
provides a useful decomposition of the empirical sensitivity. Moving from
equal contiguity to a flexibly learned operator changes the estimated
interaction structure and lowers the plug-in estimate substantially.
Accounting additionally for the first-order effect of estimating that
operator lowers the estimate again. Hence both the functional form of $W$ and
the treatment of $W$ as a generated nuisance object are empirically
consequential.

\paragraph{What the learned interaction operator changes.}
An informative feature of the updated application is that the learned
operator does not achieve its result by creating a radically different
network. The support is identical to the Census contiguity network, and the
reweighting within that support is relatively moderate. The learned operator
has an average effective number of neighbors of approximately $5.79$, and the
largest neighbor receives about $20.7\%$ of a county's row weight on average.

On the common support, the correlation between the learned edge weights and
equal-contiguity weights is approximately $0.91$. The mean row-$L_1$
difference between the two matrices is approximately $0.088$, and about
$75\%$ of individual edge weights differ from their equal-contiguity values
by less than $0.02$. The induced spatial lags are also highly correlated:
the sample correlation between learned-$W$ and equal-contiguity values of
$WY$ is approximately $0.997$.

This is useful for interpreting the empirical result. The large difference in
$\widehat\rho$ does not require a completely different graph. Rather, modest
target-relevant reweighting within a common local support can materially
change a SAR estimate, and the subsequent operator-orthogonal correction can
change it further. This pattern is consistent with the motivation for treating
the functional form of $W$ as a nuisance object rather than as an innocuous
normalization.

\paragraph{Remaining structural coefficients.}
Table~\ref{tab:empirical_beta} reports the remaining coefficients. Because the
regressors are standardized, each coefficient gives the difference in
diabetes prevalence, in percentage points, associated with a
one-standard-deviation difference in the corresponding county characteristic,
conditional on the maintained SAR-IV and control-function specification.

\begin{table}[!htbp]
\centering
\caption{County diabetes coefficients across interaction operators}
\label{tab:empirical_beta}
\small
\begin{tabular}{lcccc}
\hline\hline
&
\multicolumn{1}{c}{Learned $W$:}
&
\multicolumn{1}{c}{Learned $W$:}
&
\multicolumn{1}{c}{Fixed distance}
&
\multicolumn{1}{c}{Equal contiguity}
\\
&
\multicolumn{1}{c}{orthogonal}
&
\multicolumn{1}{c}{plug-in}
&
\multicolumn{1}{c}{$W$}
&
\multicolumn{1}{c}{$W$}
\\
\hline
Log household income
& -0.7887*** & -0.8762*** & -0.8186*** & -0.8060*** \\
& (0.1304)   & (0.0750)   & (0.1036)   & (0.0922)   \\[0.3em]

Unemployment rate
& 0.3538*** & 0.5847*** & 0.2986*** & 0.3160*** \\
& (0.0868)  & (0.0481)  & (0.0633)  & (0.0637)  \\[0.3em]

Bachelor's degree or higher
& -0.3060*** & -0.2615*** & -0.3296*** & -0.3669*** \\
& (0.0993)   & (0.0376)   & (0.0448)   & (0.0430)   \\[0.3em]

RUCC rurality score
& -0.4147*** & -0.4558*** & -0.5020*** & -0.5174*** \\
& (0.1534)   & (0.0389)   & (0.0332)   & (0.0333)   \\
\hline\hline
\end{tabular}

\begin{minipage}{0.74\textwidth}
\footnotesize
Note: The table reports coefficients on standardized county characteristics
for age-adjusted diabetes prevalence. Spatial-HAC standard errors using a
500-km Bartlett kernel are in parentheses. The learned-$W$ orthogonal
specification is the proposed estimator. *** denotes significance at the 1\%
level.
\end{minipage}
\end{table}

The signs of the non-spatial coefficients are stable across all four
specifications. Higher log household income and a larger bachelor's-degree
share are negatively associated with diabetes prevalence, while unemployment
is positively associated with diabetes prevalence. The RUCC coefficient is
negative under all four specifications. Although the magnitudes and standard
errors vary across estimators, the qualitative pattern of these coefficients
is substantially more stable than that of the spatial autoregressive
coefficient. The main empirical sensitivity therefore concerns the magnitude
assigned to spatial propagation rather than a wholesale reversal of the
direct county-level associations.

\paragraph{Robustness and numerical diagnostics.}
The positive learned-$W$ operator-orthogonal estimate is robust to the
spatial-HAC bandwidth. Table~\ref{tab:empirical_hac} reports the corresponding
sensitivity calculations.

\begin{table}[!htbp]
\centering
\caption{Spatial-HAC sensitivity of the learned-$W$ orthogonal estimate}
\label{tab:empirical_hac}
\begin{tabular}{cccc}
\hline\hline
Bandwidth (km) & $\widehat\rho$ & Standard error & 95\% CI \\
\hline
300 & 0.1668 & 0.0598 & [0.0495,\ 0.2841] \\
500 & 0.1984 & 0.0512 & [0.0981,\ 0.2987] \\
750 & 0.2032 & 0.0453 & [0.1145,\ 0.2920] \\
\hline\hline
\end{tabular}

\begin{minipage}{0.55\textwidth}
\footnotesize
Note: The table reports the learned-$W$ operator-orthogonal estimate under
the spatial-HAC bandwidth sensitivity calculations. The 500-km specification
is the baseline.
\end{minipage}
\end{table}

Across the three bandwidth calculations, the estimate ranges from $0.167$ to
$0.203$, and every reported 95\% confidence interval remains above zero. The
result is also stable to the spatial split construction. Leaving out one of
the ten spatial splits at a time produces estimates between approximately
$0.136$ and $0.302$, with every leave-one-split-out 95\% confidence interval
remaining above zero.

The nuisance and target diagnostics are also informative. The flexible
first-stage poverty regression has a mean training-sample $R^2$ of
approximately $0.59$, with a minimum across fold fits of approximately $0.36$.
The final learned-$W$ target problem is interior to the imposed
$\rho\in[-0.85,0.85]$ parameter region, the final target optimizer succeeds,
and the smallest singular value of the final target Jacobian is approximately
$0.225$. Thus the reported $\widehat\rho=0.1984$ is not generated by an active
parameter boundary or an evidently singular low-dimensional target problem.

The Riesz diagnostics show more finite-sample variation across individual
training/evaluation splits, as expected in the flexible learned-operator
specification. The held-out residual target-derivative ratio averages
approximately $0.90$ and has a median of approximately $0.88$. I therefore
treat the Riesz and interaction-score diagnostics as part of the empirical
regularization assessment rather than as model-fit statistics. The stability
of the final target estimate across HAC bandwidths and leave-one-split-out
calculations provides the more direct robustness check for the substantive
conclusion.

Overall, the application illustrates the two margins of interaction-operator
uncertainty emphasized by the theoretical framework. Under predetermined
distance-decay or equal-contiguity weights, the estimated spatial
autoregressive coefficient is close to $0.70$. Allowing the relative weights
to be learned flexibly lowers the plug-in estimate to approximately $0.44$.
Accounting additionally for the first-order effect of estimating the operator
through the proposed orthogonal score lowers the estimate to approximately
$0.20$. The preferred specification therefore continues to find positive
spatial dependence in county diabetes prevalence, but at a substantially
smaller magnitude than conventional fixed-$W$ or learned-$W$ plug-in
estimators would suggest. The empirical conclusion is not that spatial dependence disappears once $W$ is learned, but that its estimated strength is highly sensitive to how the interaction operator is constructed and to whether its unknown functional form is handled robustly.

\section{Conclusion}\label{sec:conclusion}
This paper develops a framework for inference in spatial autoregressive models
when the spatial interaction operator is learned rather than treated as known.
Within a maintained admissible support, interaction strength is generated by a
flexible function of geographic and socioeconomic characteristics. Because
estimation error in this function affects the spatial lag and, when used,
spatially transformed instruments, conventional plug-in procedures can leave a
first-order generated-weight effect.

I address this problem by constructing an operator-orthogonal SAR-IV/GMM score
that removes the leading first-order sensitivity to estimation of the
interaction function and other nuisance components. The framework also allows
the characteristics generating spatial interaction to be endogenous through a
flexible control-function representation. A feasible joint sieve-GMM procedure
provides initial estimates without requiring prior knowledge of the true
spatial weights matrix, while the finite-sieve Riesz construction accommodates
mild ill-posedness in the nuisance-to-target correction.

Spatial dependence requires a further modification of conventional
double/debiased machine learning. I develop buffered spatial cross-fitting,
which separates the complete evaluation-score footprint from the
nuisance-training footprint by a spatial guard region. The primitive
dependence analysis is based on near-epoch dependence on a spatially mixing
innovation field rather than strong mixing of the observed SAR outcome. Under
local interaction support, a stable SAR process, smooth nuisance learners, and
appropriate guard and rate conditions, the remaining dependence between
nuisance training and score evaluation becomes asymptotically negligible.

Under these conditions, the feasible estimator has the same first-order
behavior as the corresponding oracle orthogonal-score estimator, permitting
root-$n$ inference for the low-dimensional SAR parameters even when the
interaction function and other nuisance components are estimated at slower
nonparametric rates. More broadly, the results show that uncertainty about the
spatial weights matrix is not only a specification problem but also an
inference problem when the interaction structure is learned from the data.

The empirical application illustrates the practical importance of the framework using county-level diabetes prevalence in the contiguous United States. Conventional distance-decay and contiguity specifications imply spatial autoregressive coefficients near $0.70$, while a learned-$W$ plug-in estimator yields about $0.44$ and the operator-orthogonal estimator about $0.20$. Spatial dependence therefore remains positive and statistically significant, but its estimated magnitude declines substantially when the interaction operator is learned flexibly and the first-order effects of that learning are incorporated. Because the plug-in and orthogonal estimators use the same learned $W$, their contrast highlights the inferential importance of treating the learned operator as an estimated component of the model. The results also show that even relatively modest reweighting within a common local support can lead to economically meaningful changes in the estimated strength of spatial dependence.

Several extensions are natural. Alternative smooth learners can be used when
the required convergence and differentiability conditions hold, and broader
interaction supports can be accommodated when spatial locality can be
verified. The same operator-orthogonal perspective may also be useful in other
spatial and network models in which dependence structures are estimated rather
than known. Under additional regularity conditions, established spatial-HAC
methods \citep{KelejianPrucha2007,KimSun2011} provide a complementary approach
to covariance estimation under residual spatial dependence.

\newpage
\clearpage
\bibliographystyle{apalike}
\bibliography{main}

@article{AhrensBhattacharjee2015,
  author  = {Ahrens, Achim and Bhattacharjee, Arnab},
  title   = {Two-Step Lasso Estimation of the Spatial Weights Matrix},
  journal = {Econometrics},
  year    = {2015},
  volume  = {3},
  number  = {1},
  pages   = {128--155},
  doi     = {10.3390/econometrics3010128}
}

@article{AiChen2003,
  author  = {Ai, Chunrong and Chen, Xiaohong},
  title   = {Efficient Estimation of Models with Conditional Moment Restrictions Containing Unknown Functions},
  journal = {Econometrica},
  year    = {2003},
  volume  = {71},
  number  = {6},
  pages   = {1795--1843},
  doi     = {10.1111/1468-0262.00470}
}

@book{Anselin1988,
  author    = {Anselin, Luc},
  title     = {Spatial Econometrics: Methods and Models},
  year      = {1988},
  publisher = {Kluwer Academic Publishers},
  address   = {Dordrecht},
  doi       = {10.1007/978-94-015-7799-1}
}

@article{BramoulleDjebbariFortin2009,
  author  = {Bramoull{\'e}, Yann and Djebbari, Habiba and Fortin, Bernard},
  title   = {Identification of Peer Effects through Social Networks},
  journal = {Journal of Econometrics},
  year    = {2009},
  volume  = {150},
  number  = {1},
  pages   = {41--55},
  doi     = {10.1016/j.jeconom.2008.12.021}
}

@article{ChenPouzo2012,
  author  = {Chen, Xiaohong and Pouzo, Demian},
  title   = {Estimation of Nonparametric Conditional Moment Models with Possibly Nonsmooth Generalized Residuals},
  journal = {Econometrica},
  year    = {2012},
  volume  = {80},
  number  = {1},
  pages   = {277--321},
  doi     = {10.3982/ECTA7888}
}

@article{ChenPouzo2015,
  author  = {Chen, Xiaohong and Pouzo, Demian},
  title   = {Sieve {Wald} and {QLR} Inferences on Semi/Nonparametric Conditional Moment Models},
  journal = {Econometrica},
  year    = {2015},
  volume  = {83},
  number  = {3},
  pages   = {1013--1079},
  doi     = {10.3982/ECTA10771}
}

@article{ChernozhukovEtAl2018,
  author  = {Chernozhukov, Victor and Chetverikov, Denis and Demirer, Mert
             and Duflo, Esther and Hansen, Christian and Newey, Whitney
             and Robins, James},
  title   = {Double/Debiased Machine Learning for Treatment and Structural Parameters},
  journal = {The Econometrics Journal},
  year    = {2018},
  volume  = {21},
  number  = {1},
  pages   = {C1--C68},
  doi     = {10.1111/ectj.12097}
}

@article{ChernozhukovEscancianoIchimuraNeweyRobins2022,
  author  = {Chernozhukov, Victor and Escanciano, Juan Carlos
             and Ichimura, Hidehiko and Newey, Whitney K.
             and Robins, James M.},
  title   = {Locally Robust Semiparametric Estimation},
  journal = {Econometrica},
  year    = {2022},
  volume  = {90},
  number  = {4},
  pages   = {1501--1535},
  doi     = {10.3982/ECTA16294}
}

@article{ChernozhukovNeweySingh2022,
  author  = {Chernozhukov, Victor and Newey, Whitney K. and Singh, Rahul},
  title   = {Debiased Machine Learning of Global and Local Parameters Using Regularized {Riesz} Representers},
  journal = {The Econometrics Journal},
  year    = {2022},
  volume  = {25},
  number  = {3},
  pages   = {576--601},
  doi     = {10.1093/ectj/utac002}
}

@article{ChernozhukovHuangWang2026,
  author  = {Chernozhukov, Victor and Huang, Chen and Wang, Weining},
  title   = {Uniform Inference on High-Dimensional Spatial Panel Networks},
  journal = {Journal of Business \& Economic Statistics},
  year    = {2026},
  volume  = {44},
  number  = {1},
  pages   = {348--359},
  doi     = {10.1080/07350015.2025.2530122}
}

@article{ChiangKatoMaSasaki2022,
  author  = {Chiang, Harold D. and Kato, Kengo and Ma, Yukun and Sasaki, Yuya},
  title   = {Multiway Cluster Robust Double/Debiased Machine Learning},
  journal = {Journal of Business \& Economic Statistics},
  year    = {2022},
  volume  = {40},
  number  = {3},
  pages   = {1046--1056},
  doi     = {10.1080/07350015.2021.1895815}
}

@article{ChiangMaRodrigueSasaki2026,
  author  = {Chiang, Harold D. and Ma, Yukun and Rodrigue, Joel B. and Sasaki, Yuya},
  title   = {Double/Debiased Machine Learning for Dyadic Data},
  journal = {Econometric Theory},
  year    = {2026},
  pages   = {1--22},
  doi     = {10.1017/S0266466625100273},
  note    = {}
}

@article{CiganovicEtAl2026,
  author  = {Ciganovic, Milos and D'Amario, Federico and Tancioni, Massimiliano},
  title   = {Double Machine Learning for Time Series},
  journal = {The Econometrics Journal},
  year    = {2026},
  doi     = {10.1093/ectj/utag019}
}

@article{ConleyLigon2002,
  author  = {Conley, Timothy G. and Ligon, Ethan},
  title   = {Economic Distance and Cross-Country Spillovers},
  journal = {Journal of Economic Growth},
  year    = {2002},
  volume  = {7},
  number  = {2},
  pages   = {157--187},
  doi     = {10.1023/A:1015676113101}
}

@article{dePaulaRasulSouza2025,
  author  = {de Paula, {\'A}ureo and Rasul, Imran and Souza, Pedro C. L.},
  title   = {Identifying Network Ties from Panel Data: Theory and an Application to Tax Competition},
  journal = {The Review of Economic Studies},
  year    = {2025},
  volume  = {92},
  number  = {4},
  pages   = {2691--2729},
  doi     = {10.1093/restud/rdae088}
}

@article{EmmeneggerEtAl2025,
  author  = {Emmenegger, Corinne and Spohn, Meta-Lina and Elmer, Timon
             and B\"{u}hlmann, Peter},
  title   = {Treatment Effect Estimation with Observational Network Data Using Machine Learning},
  journal = {Journal of Causal Inference},
  year    = {2025},
  volume  = {13},
  number  = {1},
  pages   = {20230082},
  doi     = {10.1515/jci-2023-0082}
}

@misc{GuptaQuZhang2026,
  author        = {Gupta, Abhimanyu and Qu, Xi and Zhang, Jiajun},
  title         = {Semi-Nonparametric Estimation of Spatial Dynamic Panel Data Models with Nonparametric Spatial Weights},
  year          = {2026},
  eprint        = {2606.24266},
  archivePrefix = {arXiv},
  primaryClass  = {econ.EM},
  doi           = {10.48550/arXiv.2606.24266}
}

@article{HarrisMoffatKravtsova2011,
  author  = {Harris, Richard and Moffat, John and Kravtsova, Victoria},
  title   = {In Search of {$W$}},
  journal = {Spatial Economic Analysis},
  year    = {2011},
  volume  = {6},
  number  = {3},
  pages   = {249--270},
  doi     = {10.1080/17421772.2011.586721}
}

@article{JenishPrucha2009,
  author  = {Jenish, Nazgul and Prucha, Ingmar R.},
  title   = {Central Limit Theorems and Uniform Laws of Large Numbers for Arrays of Random Fields},
  journal = {Journal of Econometrics},
  year    = {2009},
  volume  = {150},
  number  = {1},
  pages   = {86--98},
  doi     = {10.1016/j.jeconom.2009.02.009}
}

@article{JenishPrucha2012,
  author  = {Jenish, Nazgul and Prucha, Ingmar R.},
  title   = {On Spatial Processes and Asymptotic Inference under Near-Epoch Dependence},
  journal = {Journal of Econometrics},
  year    = {2012},
  volume  = {170},
  number  = {1},
  pages   = {178--190},
  doi     = {10.1016/j.jeconom.2012.05.022}
}

@article{Juhl2020,
  author  = {Juhl, Sebastian},
  title   = {The Sensitivity of Spatial Regression Models to Network Misspecification},
  journal = {Political Analysis},
  year    = {2020},
  volume  = {28},
  number  = {1},
  pages   = {1--19},
  doi     = {10.1017/pan.2019.12}
}

@article{KelejianPrucha1998,
  author  = {Kelejian, Harry H. and Prucha, Ingmar R.},
  title   = {A Generalized Spatial Two-Stage Least Squares Procedure for Estimating a Spatial Autoregressive Model with Autoregressive Disturbances},
  journal = {The Journal of Real Estate Finance and Economics},
  year    = {1998},
  volume  = {17},
  number  = {1},
  pages   = {99--121},
  doi     = {10.1023/A:1007707430416}
}

@article{KelejianPrucha1999,
  author  = {Kelejian, Harry H. and Prucha, Ingmar R.},
  title   = {A Generalized Moments Estimator for the Autoregressive Parameter in a Spatial Model},
  journal = {International Economic Review},
  year    = {1999},
  volume  = {40},
  number  = {2},
  pages   = {509--533},
  doi     = {10.1111/1468-2354.00027}
}

@article{KelejianPrucha2007,
  author  = {Kelejian, Harry H. and Prucha, Ingmar R.},
  title   = {{HAC} Estimation in a Spatial Framework},
  journal = {Journal of Econometrics},
  year    = {2007},
  volume  = {140},
  number  = {1},
  pages   = {131--154}
}

@article{KimSun2011,
  author  = {Kim, Min Seong and Sun, Yixiao},
  title   = {Spatial Heteroskedasticity and Autocorrelation Consistent
             Estimation of Covariance Matrix},
  journal = {Journal of Econometrics},
  year    = {2011},
  volume  = {160},
  number  = {2},
  pages   = {349--371}
}

@article{KojevnikovMarmerSong2021,
  author  = {Kojevnikov, Denis and Marmer, Vadim and Song, Kyungchul},
  title   = {Limit Theorems for Network Dependent Random Variables},
  journal = {Journal of Econometrics},
  year    = {2021},
  volume  = {222},
  number  = {2},
  pages   = {882--908},
  doi     = {10.1016/j.jeconom.2020.05.019}
}

@article{LamSouza2020,
  author  = {Lam, Clifford and Souza, Pedro C. L.},
  title   = {Estimation and Selection of Spatial Weight Matrix in a Spatial Lag Model},
  journal = {Journal of Business \& Economic Statistics},
  year    = {2020},
  volume  = {38},
  number  = {3},
  pages   = {693--710},
  doi     = {10.1080/07350015.2019.1569526}
}

@article{Lee2004,
  author  = {Lee, Lung-Fei},
  title   = {Asymptotic Distributions of Quasi-Maximum Likelihood Estimators for Spatial Autoregressive Models},
  journal = {Econometrica},
  year    = {2004},
  volume  = {72},
  number  = {6},
  pages   = {1899--1925},
  doi     = {10.1111/j.1468-0262.2004.00558.x}
}

@misc{Lee2022,
  author        = {Lee, Jieun},
  title         = {Evidence and Strategy on Economic Distance in Spatially Augmented Solow--Swan Growth Model},
  year          = {2022},
  eprint        = {2209.05562},
  archivePrefix = {arXiv},
  primaryClass  = {econ.EM},
  doi           = {10.48550/arXiv.2209.05562}
}

@book{LeSagePace2009,
  author    = {LeSage, James P. and Pace, R. Kelley},
  title     = {Introduction to Spatial Econometrics},
  year      = {2009},
  publisher = {Chapman \& Hall/CRC},
  address   = {Boca Raton, FL}
}

@article{LinLee2010,
  author  = {Lin, Xu and Lee, Lung-Fei},
  title   = {{GMM} Estimation of Spatial Autoregressive Models with Unknown Heteroskedasticity},
  journal = {Journal of Econometrics},
  year    = {2010},
  volume  = {157},
  number  = {1},
  pages   = {34--52},
  doi     = {10.1016/j.jeconom.2009.10.035}
}

@article{LinSong2025,
   author  = {Lin, Yanli and Song, Yichun},
   title   = {Addressing Endogeneity Issues in a Spatial Autoregressive Model Using Copulas},
   journal = {Journal of Econometrics},
   volume  = {252},
   pages   = {106106},
   year    = {2025},
   doi     = {10.1016/j.jeconom.2025.106106}
}

@article{Newey1994,
  author  = {Newey, Whitney K.},
  title   = {The Asymptotic Variance of Semiparametric Estimators},
  journal = {Econometrica},
  year    = {1994},
  volume  = {62},
  number  = {6},
  pages   = {1349--1382},
  doi     = {10.2307/2951752}
}

@article{QuLee2015,
  author  = {Qu, Xi and Lee, Lung-fei},
  title   = {Estimating a Spatial Autoregressive Model with an Endogenous Spatial Weight Matrix},
  journal = {Journal of Econometrics},
  year    = {2015},
  volume  = {184},
  number  = {2},
  pages   = {209--232},
  doi     = {10.1016/j.jeconom.2014.08.008}
}

@article{QuLeeYang2021,
  author  = {Qu, Xi and Lee, Lung-fei and Yang, Chao},
  title   = {Estimation of a {SAR} Model with Endogenous Spatial Weights Constructed by Bilateral Variables},
  journal = {Journal of Econometrics},
  year    = {2021},
  volume  = {221},
  number  = {1},
  pages   = {180--197},
  doi     = {10.1016/j.jeconom.2020.05.011}
}

@article{StakhovychBijmolt2009,
  author  = {Stakhovych, Stanislav and Bijmolt, Tammo H. A.},
  title   = {Specification of Spatial Models: A Simulation Study on Weights Matrices},
  journal = {Papers in Regional Science},
  year    = {2009},
  volume  = {88},
  number  = {2},
  pages   = {389--408},
  doi     = {10.1111/j.1435-5957.2008.00213.x}
}

@article{Sun2016,
  author  = {Sun, Yiguo},
  title   = {Functional-Coefficient Spatial Autoregressive Models with Nonparametric Spatial Weights},
  journal = {Journal of Econometrics},
  year    = {2016},
  volume  = {195},
  number  = {1},
  pages   = {134--153},
  doi     = {10.1016/j.jeconom.2016.07.005}
}

\newpage
\appendix

\section*{Appendix}
\setcounter{equation}{0}
\setcounter{table}{0}
\setcounter{figure}{0}
\setcounter{subsection}{0}

\renewcommand{\theequation}{A.\arabic{equation}}
\renewcommand{\thetable}{A.\arabic{table}}
\renewcommand{\thefigure}{A.\arabic{figure}}

\renewcommand{\thesection}{A}
\renewcommand{\thesubsection}{A.\arabic{subsection}}

% Proofs are ordered to match the sequence of results in the manuscript.

\subsection{Proof of Proposition \ref{prop:w_derivative}}
\label{app:proof_w_derivative}

\noindent\textbf{Proof idea.}
Differentiate the row-normalized exponential weights along the path
$g_t=g+t\delta g$. The quotient rule separates the direct effect of the
pair-specific score from the offsetting change in the row normalization.
Applying the resulting derivative to $Y$ gives the derivative of the spatial
lag.

For $t$ in a neighborhood of zero, let
\[
A_{ij}(t)
=
S_{ij,n}\exp\{g(R_{ij})+t\delta g(R_{ij})\},
\qquad
A_i(t)
=
\sum_{k\neq i}A_{ik}(t).
\]
For a supported pair $i\neq j$,
$
w_{ij}(g_t)=\frac{A_{ij}(t)}{A_i(t)},
\;
A_{ij}^{\prime}(0)=A_{ij}(0)\delta g(R_{ij}),
$
and
$
A_i^{\prime}(0)
=
\sum_{k\neq i}A_{ik}(0)\delta g(R_{ik}).
$
The quotient rule therefore gives
\begin{align*}
D_gw_{ij}(g)[\delta g]
&=
\frac{A_{ij}^{\prime}(0)A_i(0)-A_{ij}(0)A_i^{\prime}(0)}{A_i(0)^2}
\\
&=
w_{ij}(g)
\left[
\delta g(R_{ij})
-
\sum_{k\neq i}w_{ik}(g)\delta g(R_{ik})
\right].
\end{align*}
If $S_{ij,n}=0$, then $w_{ij}(g)=0$ for every admissible $g$, so the
corresponding derivative equals zero.

Holding the realized vector $Y$ fixed,
$
\{W_n(g)Y\}_i
=
\sum_{j\neq i}w_{ij}(g)Y_j.
$
Hence
\begin{align*}
D_g\{W_n(g)Y\}_i[\delta g]
&=
\sum_{j\neq i}Y_jD_gw_{ij}(g)[\delta g]
\\
&=
\sum_{j\neq i}w_{ij}(g)Y_j\delta g(R_{ij})
-
\{W_n(g)Y\}_i
\sum_{j\neq i}w_{ij}(g)\delta g(R_{ij})
\\
&=
\sum_{j\neq i}
w_{ij}(g)
\left[
Y_j-\{W_n(g)Y\}_i
\right]
\delta g(R_{ij}).
\end{align*}
\hfill$\square$

\subsection{Proof of Proposition \ref{prop:g_to_W_transfer}}
\label{app:proof_g_to_W_transfer}

\noindent\textbf{Proof idea.}
The first-order term is controlled by the target-relevant seminorm, while the
second-order remainder is controlled by one target norm and one stronger
local envelope norm. This is the reason for introducing
$\|\cdot\|_{\mathcal G,+,n}$.

Let
$
d_{\mathrm{tar}}
=
\operatorname{dist}_{\mathcal G,\mathrm{tar}}(g,[g_0]_n),
\;
d_+
=
\operatorname{dist}_{\mathcal G,+}(g,[g_0]_n).
$
Choose $g_0^\star\in[g_0]_n$ attaining $d_+$ up to an arbitrarily small
error and write
$
\Delta g=g-g_0^\star.
$
Under row normalization, two representatives in $[g_0]_n$ differ on every
supported row only through directions that leave the normalized weights
unchanged. Such directions are in the null space of both
$D_g\{W_n(g_0)Y\}$ and $D_gH_i(g_0)$. Consequently the target seminorm of
$g-\widetilde g_0$ is invariant to the choice of
$\widetilde g_0\in[g_0]_n$, and therefore, up to the arbitrary selection
error,
$
\|\Delta g\|_{\mathcal G,\mathrm{tar},n}=d_{\mathrm{tar}},
\;
\|\Delta g\|_{\mathcal G,+,n}=d_+.
$
Also $W_n(g_0^\star)=W_n(g_0)$ and $H_i(g_0^\star)=H_i(g_0)$ because the
latter depends on $g$ through the induced spatial operator.

A second-order Gateaux expansion along
$g_t=g_0^\star+t\Delta g$ yields
\[
\{W_n(g)-W_n(g_0)\}Y
=
D_g\{W_n(g_0)Y\}[\Delta g]
+
R_{W,n}(\Delta g),
\]
where, by the integral form of the remainder and Assumption
\ref{ass:operator_map_smoothness},
\begin{align*}
\|R_{W,n}(\Delta g)\|_{2,n}
&\leq
\int_0^1(1-t)
\left\|
D_g^2\{W_n(g_t)Y\}[\Delta g,\Delta g]
\right\|_{2,n}\,dt
\\
&\leq
C
\|\Delta g\|_{\mathcal G,\mathrm{tar},n}
\|\Delta g\|_{\mathcal G,+,n}
\\
&\leq
C d_{\mathrm{tar}}d_+.
\end{align*}
By definition of the target seminorm,
$
\left\|
D_g\{W_n(g_0)Y\}[\Delta g]
\right\|_{2,n}
\leq
d_{\mathrm{tar}}.
$
The triangle inequality therefore gives
\[
\|\{W_n(g)-W_n(g_0)\}Y\|_{2,n}
\leq
C\{d_{\mathrm{tar}}+d_{\mathrm{tar}}d_+\}.
\]

The same expansion applied coordinatewise to $H_i(g)$ gives
$
H_i(g)-H_i(g_0)
=
D_gH_i(g_0)[\Delta g]+R_{H,i}(\Delta g).
$
The first derivative is controlled by the second component of the target
seminorm, while the second derivative is controlled by Assumption
\ref{ass:operator_map_smoothness}. Hence
\[
\left[
\frac1n\sum_{i=1}^n
E\|H_i(g)-H_i(g_0)\|_2^2
\right]^{1/2}
\leq
C\{d_{\mathrm{tar}}+d_{\mathrm{tar}}d_+\}.
\]
If $d_{\mathrm{tar}}=O_p(r_{g,n})$ and $d_+=o_p(1)$, both generated objects
are $O_p(r_{g,n})$, proving the final assertion.
\hfill$\square$

\subsection{Proof of Proposition \ref{prop:local_identification}}
\label{app:proof_identification}

\noindent\textbf{Proof idea.}
The first-stage restriction identifies $m_0$ first. Differentiating the
conditional structural restriction and applying $\mathcal R_C$ removes the
control-function direction. The local separation condition then forces both
the target direction and the target-relevant operator direction to vanish.

The first-stage conditional-mean restriction uniquely identifies $m_0$ by
Assumption \ref{ass:local_separation}. Hence a locally observationally
equivalent differentiable path has zero first-stage direction. Consider a
path
$
t\mapsto(\theta_t,g_t,h_t)
$
through $(\theta_0,g_0,h_0)$ with derivative
$(\delta\theta,\delta g,\delta h)$. Differentiating
\eqref{eq:conditional_moment} at $t=0$ gives
\[
E\left[
-\Delta_i(\delta\theta,\delta g)
-\delta h(C_{i0})
\mid
C_{i0},\mathcal A_i
\right]
=0.
\]
Applying $\mathcal R_C$ eliminates $\delta h(C_{i0})$ because it is
measurable with respect to $C_{i0}$, yielding
$
\mathcal R_C\Delta_i(\delta\theta,\delta g)=0.
$
Assumption \ref{ass:local_separation} therefore implies
$
\|\delta\theta\|_2
+
\|\delta g\|_{\mathcal G,\mathrm{tar},n}
=0.
$
Thus $\delta\theta=0$ and, by the definition of the target-relevant norm,
$
A_{\delta g}Y=0
\quad\text{in }L^2.
$
Returning to the differentiated conditional moment gives
$
\delta h(C_{i0})=0
\quad\text{a.s.}
$
If Assumption \ref{ass:operator_richness} also holds, then
$
\frac1nE\|A_{\delta g}\|_F^2
\leq
\kappa_W^{-1}
\frac1nE\|A_{\delta g}Y\|_2^2
=0,
$
so $A_{\delta g}=0$ in $L^2$.
\hfill$\square$

\subsection{Proof of Proposition \ref{prop:sieve_moment_identification}}
\label{app:proof_sieve_moment_identification}

\noindent\textbf{Proof idea.}
The baseline interaction dictionary is predetermined and independent of $g$,
so its partial derivative does not enter the $g$ moment. Residualization by
$C_{i0}$ converts the moment derivative into the conditional variation
$R_{i,\delta g_J}$, which the richness condition approximates uniformly.

At the truth, $B_{g,i}^0$ does not depend on $g$, and $m$ and $\ell$ are held
fixed when taking the partial Gateaux derivative in the $g$ direction.
Therefore
$
D_g\widetilde B_{g,i}(\eta_0)[\delta g_J]=0.
$
From \eqref{eq:xi_candidate},
$
D_g\xi_i(\theta_0,\eta_0)[\delta g_J]
=
-\rho_0(A_{\delta g_J}Y)_i,
$
so
$
D_gE[s_{g,i}(\theta_0,\eta_0)][\delta g_J]
=
-\rho_0
E\left[
\widetilde B_{g,i}(\eta_0)
(A_{\delta g_J}Y)_i
\right].
$
Let $T_i=(A_{\delta g_J}Y)_i$. Since
$\widetilde B_{g,i}(\eta_0)$ is measurable with respect to
$\sigma(C_{i0},\mathcal A_i)$ and
$E[\widetilde B_{g,i}(\eta_0)\mid C_{i0}]=0$,
\begin{align*}
E[\widetilde B_{g,i}(\eta_0)T_i]
&=
E\left[
\widetilde B_{g,i}(\eta_0)
E[T_i\mid C_{i0},\mathcal A_i]
\right]
\\
&=
E\left[
\widetilde B_{g,i}(\eta_0)
\{E[T_i\mid C_{i0},\mathcal A_i]-E[T_i\mid C_{i0}]\}
\right]
\\
&=
E\left[
\widetilde B_{g,i}(\eta_0)
R_{i,\delta g_J}
\right].
\end{align*}
Let $a_J=a_J(\delta g_J)$ be supplied by Assumption
\ref{ass:sieve_moment_richness}. Then
\begin{align*}
&a_J^{\prime}
D_g\left\{
\frac1n\sum_{i=1}^nE[s_{g,i}(\theta_0,\eta_0)]
\right\}[\delta g_J]
\\
&\qquad=
-\rho_0
\frac1n\sum_{i=1}^n
E\left[
\{a_J^{\prime}\widetilde B_{g,i}(\eta_0)\}
R_{i,\delta g_J}
\right]
\\
&\qquad=
-\rho_0
\frac1n\sum_{i=1}^nE[R_{i,\delta g_J}^2]
+R_{J,n}.
\end{align*}
By Cauchy--Schwarz and Assumption \ref{ass:sieve_moment_richness},
$
|R_{J,n}|
\leq
C|\rho_0|\zeta_{J_n}
\left[
\frac1n\sum_{i=1}^nE[R_{i,\delta g_J}^2]
\right]^{1/2}.
$
For a unit target-relevant direction, Assumption
\ref{ass:local_separation} with $\delta\theta=0$ gives
$
|\rho_0|
\left[
\frac1n\sum_{i=1}^nE[R_{i,\delta g_J}^2]
\right]^{1/2}
\geq
\kappa_{\mathrm{id}}.
$
Because $|\rho_0|$ is bounded away from zero and bounded above by stability,
the leading quadratic term is bounded away from zero, while
$\zeta_{J_n}\to0$ makes the approximation error asymptotically smaller.
Hence, for all sufficiently large $J_n$,
$
\left|
a_J^{\prime}
D_g\left\{
\frac1n\sum_{i=1}^nE[s_{g,i}(\theta_0,\eta_0)]
\right\}[\delta g_J]
\right|
\geq c>0.
$
Since $\|a_J\|_2\leq C$, Cauchy--Schwarz implies the stated norm lower bound.
Positive-definite GMM weighting and regular local second derivatives then
imply local quadratic identification of the population sieve criterion.
\hfill$\square$

\subsection{Proof of Proposition \ref{prop:joint_preliminary}}
\label{app:proof_joint_preliminary}

\noindent\textbf{Proof idea.}
Uniform convergence and global separation first place the feasible joint
estimator in the identified local basin. Local quadratic curvature then
converts the empirical criterion fluctuation, penalty drift, and numerical
suboptimality into an explicit rate for the entire joint parameter vector.
This rate is obtained before profiling $g$ and therefore closes the rate
chain used in Proposition \ref{prop:sieve_g_rate}.

Write the unpenalized sample GMM part as
$
\widehat Q_k^{\mathrm{GMM}}(\vartheta)
=
\widehat{\mathfrak m}_k(\vartheta)^{\prime}
\widehat{\mathcal W}_k
\widehat{\mathfrak m}_k(\vartheta),
$
and define
$
\mathbb G_{k,n}(\vartheta)
=
\widehat Q_k^{\mathrm{GMM}}(\vartheta)
-
Q_{0,n}(\vartheta),
$
with the asymptotically negligible weighting-matrix replacement absorbed into
this empirical criterion fluctuation. Then
$
\widehat Q_k(\vartheta)
=
Q_{0,n}(\vartheta)
+
\mathbb G_{k,n}(\vartheta)
+
\operatorname{Pen}_n(\vartheta),
$
where
$
\operatorname{Pen}_n(\vartheta)
=
\sum_{a\in\{m,g,h,\ell\}}
\lambda_{a,n}\mathcal P_a(\gamma_a).
$

\paragraph{Consistency.}
The spatial ULLN, convergence of $\widehat{\mathcal W}_k$, and the vanishing
uniform penalty imply

\noindent $
\sup_{\vartheta\in\mathcal V_{J_n}}
\left|
\widehat Q_k(\vartheta)-Q_{0,n}(\vartheta)
\right|
=o_p(1)
$
uniformly over the fixed number of folds. Let
$\vartheta_{0,J_n}$ denote the normalized population sieve minimizer. Global
separation implies that for every $\epsilon>0$ there exists
$c_\epsilon>0$ such that, for all sufficiently large $n$,
\[
\inf_{\substack{
\vartheta\in\mathcal V_{J_n}:\\
d_{\vartheta,n}(\vartheta,\vartheta_{0,J_n})\geq\epsilon
}}
\left\{
Q_{0,n}(\vartheta)-Q_{0,n}(\vartheta_{0,J_n})
\right\}
\geq
c_\epsilon.
\]
The standard argmin argument therefore gives
$
d_{\vartheta,n}
\left(
\widetilde\vartheta^{(-k)},\vartheta_{0,J_n}
\right)
=o_p(1)
$
uniformly over folds. The minimum-norm convention affects only the
representative of the observationally equivalent $g$ coefficients and not
the induced operator or the metric.

\paragraph{Local rate.}
By consistency, with probability approaching one the estimator lies in the
local neighborhood $\mathcal B_{J_n}$ on which the quadratic lower bound and
the local empirical-process bound hold. Let
$
\widehat d_k
=
d_{\vartheta,n}
\left(
\widetilde\vartheta^{(-k)},\vartheta_{0,J_n}
\right)
$
and set
$
a_{k,n}
=
C
\sqrt{
\frac{\mathfrak C_{\mathrm{joint},n}}
{|\mathcal T_k|}
}
+
b_{\mathrm{joint},n}^{\mathrm{pen}}.
$
The numerical solution satisfies
$
\widehat Q_k
\left(
\widetilde\vartheta^{(-k)}
\right)
\leq
\widehat Q_k(\vartheta_{0,J_n})
+
\epsilon_{\mathrm{opt},n}.
$
Subtract $Q_{0,n}(\vartheta_{0,J_n})$ and use the local criterion fluctuation
and penalty-drift bounds. Up to an event whose probability tends to one,
$
Q_{0,n}
\left(
\widetilde\vartheta^{(-k)}
\right)
-
Q_{0,n}(\vartheta_{0,J_n})
\leq
C a_{k,n}\widehat d_k
+
\epsilon_{\mathrm{opt},n}.
$
The local quadratic lower bound therefore implies
$
\kappa_{\mathrm{joint},n}\widehat d_k^2
\leq
C a_{k,n}\widehat d_k
+
\epsilon_{\mathrm{opt},n}.
$
Using $2xy\leq x^2+y^2$, or equivalently solving the quadratic inequality,
we obtain
$
\widehat d_k
=
O_p\left(
\frac{a_{k,n}}{\kappa_{\mathrm{joint},n}}
+
\sqrt{
\frac{\epsilon_{\mathrm{opt},n}}
{\kappa_{\mathrm{joint},n}}
}
\right).
$
Substituting the definition of $a_{k,n}$ gives exactly
$
\widehat d_k
=
O_p(r_{\mathrm{joint},n}).
$
The assumptions that the stochastic and penalty terms are
$o(\kappa_{\mathrm{joint},n})$ and that the optimization error is locally
negligible ensure that this rate remains inside $\mathcal B_{J_n}$, so the
local argument is self-consistent.

If
$
d_{\vartheta,n}(\vartheta_{0,J_n},\vartheta_0)
\leq
a_{\mathrm{joint},n},
$
the triangle inequality gives
$
d_{\vartheta,n}
\left(
\widetilde\vartheta^{(-k)},\vartheta_0
\right)
=
O_p\left(
r_{\mathrm{joint},n}+a_{\mathrm{joint},n}
\right).
$
Finally, under the stated fourth-moment versions of the normalized empirical
fluctuation and optimization bounds, the preceding quadratic inequality
holds in $L^4$ with the same normalization. Applying Minkowski's inequality
to the approximation term yields
$
\max_{1\leq k\leq K}
E\left[
d_{\vartheta,n}^4
\left(
\widetilde\vartheta^{(-k)},\vartheta_0
\right)
\right]
\leq
C
\left(
r_{\mathrm{joint},n}+a_{\mathrm{joint},n}
\right)^4.
$\hfill$\square$

\subsection{Proof of Proposition \ref{prop:sieve_g_rate}}
\label{app:proof_sieve_g_rate}

\noindent\textbf{Proof idea.}
Local curvature of the profiled criterion converts its empirical gradient
into a coefficient rate. The nuisance components held fixed in the profile
contribute only through the already-established joint rate from Proposition
\ref{prop:joint_preliminary}. A separate spline maximal inequality supplies
the stronger local envelope rate required by the second-order operator
expansion.

Let
$
\Delta\gamma_k
=
\widehat\gamma^{(-k)}-\gamma_{0,J_n}.
$
On the normalized local basin, local quadratic identification gives, for a
constant $c_g>0$,
$
Q_{g,0}(\gamma_{0,J_n}+\Delta\gamma_k)
-
Q_{g,0}(\gamma_{0,J_n})
\geq
c_g\|\Delta\gamma_k\|_2^2
+
o(\|\Delta\gamma_k\|_2^2).
$
The empirical first-order perturbation around $\gamma_{0,J_n}$ is bounded by
$
O_p\left(
\sqrt{
\frac{\mathfrak C_{g,n}}
{|\mathcal T_k|}
}
\right)
\|\Delta\gamma_k\|_2.
$
Sieve approximation contributes
$O(a_{J_n}^{\mathrm{tar}}\|\Delta\gamma_k\|_2)$ in the target-relevant
criterion. Smoothness of the profiled moment map in the components held
fixed at $\widetilde\vartheta_{-g}^{(-k)}$ contributes
$
O_p(r_{-g,n}\|\Delta\gamma_k\|_2),
$
and Proposition \ref{prop:joint_preliminary} gives
$
r_{-g,n}
\lesssim
r_{\mathrm{joint},n}+a_{\mathrm{joint},n}.
$
Finally, the local penalty perturbation is bounded by
$
O\left(
b_{g,n}^{\mathrm{pen}}\|\Delta\gamma_k\|_2
\right).
$
The basic inequality for the local minimizer therefore yields
\[
\|\Delta\gamma_k\|_2
=
O_p\left[
\sqrt{
\frac{\mathfrak C_{g,n}}
{|\mathcal T_k|}
}
+
a_{J_n}^{\mathrm{tar}}
+
r_{\mathrm{joint},n}
+
a_{\mathrm{joint},n}
+
b_{g,n}^{\mathrm{pen}}
\right].
\]
Local domination of the target-relevant sieve norm by the normalized
coefficient norm gives
\[
\operatorname{dist}_{\mathcal G,\mathrm{tar}}
\left(
\widehat g^{(-k)},[g_0]_n
\right)
=
O_p\left[
\sqrt{
\frac{\mathfrak C_{g,n}}
{|\mathcal T_k|}
}
+
a_{J_n}^{\mathrm{tar}}
+
r_{\mathrm{joint},n}
+
a_{\mathrm{joint},n}
+
b_{g,n}^{\mathrm{pen}}
\right].
\]
This proves the first rate and, in particular, removes the old circular term:
the preliminary nuisance contribution is controlled by an independently
proved joint rate.

For the stronger norm, decompose
$
\widehat g^{(-k)}-g_0^\star
=
\{
\widehat g^{(-k)}-g_{0,J_n}
\}
+
\{
g_{0,J_n}-g_0^\star
\},
$
where $g_0^\star\in[g_0]_n$ is the normalized representative used in the
sieve approximation. The assumed local spline maximal inequality gives
\[
\|\widehat g^{(-k)}-g_{0,J_n}\|_{\infty,\mathcal R_n}
=
O_p\left[
\sqrt{
\frac{\mathfrak C_{g,\infty,n}\log n}
{|\mathcal T_k|}
}
+
r_{\mathrm{joint},n}
+
a_{\mathrm{joint},n}
+
b_{g,n}^{\mathrm{pen}}
\right],
\]
while the deterministic approximation obeys
$
\|g_{0,J_n}-g_0^\star\|_{\infty,\mathcal R_n}
\lesssim
a_{J_n}^{+}.
$
Combining this sup-norm bound with the target-relevant rate proves
$
\operatorname{dist}_{\mathcal G,+}
\left(
\widehat g^{(-k)},[g_0]_n
\right)
=
O_p(r_{g,+,n}).
$

Under the benchmark
$|\mathcal T_k|\asymp n$,
$\mathfrak C_{g,n}\lesssim J_n$, and
$a_{J_n}^{\mathrm{tar}}\asymp J_n^{-\alpha/d_R}$, balancing
$J_n^{1/2}n^{-1/2}$ and $J_n^{-\alpha/d_R}$ gives
$
J_n\asymp n^{d_R/(2\alpha+d_R)}
$
and hence
$
\operatorname{dist}_{\mathcal G,\mathrm{tar}}
\left(
\widehat g^{(-k)},[g_0]_n
\right)
=
O_p\left(
n^{-\alpha/(2\alpha+d_R)}
\right).
$
With $\mathfrak C_{g,\infty,n}\lesssim J_n$, the local envelope rate differs
only by the usual $\sqrt{\log n}$ factor. Since
$
\frac{\alpha}{2\alpha+d_R}>\frac14
\quad\Longleftrightarrow\quad
\alpha>\frac{d_R}{2},
$
both rates are $o(n^{-1/4})$ under the stated smoothness condition. The
fourth-moment conclusion follows from the corresponding fourth-moment
maximal inequalities and the same deterministic approximation decomposition.
\hfill$\square$

\subsection{Proof of Proposition \ref{prop:naive_nonorthogonal}}
\label{app:proof_naive_nonorthogonal}

\noindent\textbf{Proof idea.}
Even when a valid instrument component does not depend on $g$, the structural
residual does. Hence a perturbation of the learned interaction function
changes the target moment at first order.

Because $H_i^0$ is independent of $g$,
$
D_gE[H_i^0\xi_i(\theta_0,\eta_0)][\delta g]
=
E\left[
H_i^0D_g\xi_i(\theta_0,\eta_0)[\delta g]
\right].
$
By \eqref{eq:xi_candidate},
$
D_g\xi_i(\theta_0,\eta_0)[\delta g]
=
-\rho_0D_g\{W_n(g_0)Y\}_i[\delta g].
$
Therefore
$
D_gE[H_i^0\xi_i(\theta_0,\eta_0)][\delta g]
=
-\rho_0
E\left[
H_i^0D_g\{W_n(g_0)Y\}_i[\delta g]
\right].
$
Assumption \ref{ass:w_sensitivity} makes this quantity nonzero for at least
one admissible direction, proving the proposition.
\hfill$\square$

\subsection{Proof of Proposition \ref{prop:riesz_existence}}
\label{app:proof_riesz_existence}

\noindent\textbf{Proof idea.}
On a finite target-relevant sieve space, full column rank of the stacked
nuisance Jacobian gives an exact Moore--Penrose representation of the target
sensitivity. Approximation of the sieve representer and asymptotic
homogeneity of blockwise derivative maps then yield the common blockwise
Riesz representation.

Because
$\operatorname{rank}(A_{\eta,0}^{(J)})=J$, the Moore--Penrose inverse
satisfies
$
\left(A_{\eta,0}^{(J)}\right)^\dagger
A_{\eta,0}^{(J)}
=I_J.
$
Define
$
\Gamma_{0,J}
=
G_{\eta,0}^{(J)}
\left(A_{\eta,0}^{(J)}\right)^\dagger.
$
Then
$
\Gamma_{0,J}A_{\eta,0}^{(J)}
=
G_{\eta,0}^{(J)}
\left(A_{\eta,0}^{(J)}\right)^\dagger
A_{\eta,0}^{(J)}
=
G_{\eta,0}^{(J)},
$
so the representation is exact on the sieve tangent space. The source
condition gives
$
\sup_J\|\Gamma_{0,J}\|_{\mathrm{op}}<\infty,
$
even if $\kappa_J\downarrow0$.

To obtain the blockwise approximation, let $A_{\eta,0}$ and
$G_{\eta,0}$ denote the common population derivative maps and use the same
symbols for their restrictions to the target-relevant sieve tangent. On this
space,
$
G_{\eta,0}=\Gamma_{0,J}A_{\eta,0}.
$
For each block,
$
G_{\eta,0,k}-\Gamma_0A_{\eta,0,k}
={}
\{G_{\eta,0,k}-G_{\eta,0}\}
+
(\Gamma_{0,J}-\Gamma_0)A_{\eta,0}
+
\Gamma_0\{A_{\eta,0}-A_{\eta,0,k}\}.
$
Under normalized tangent coordinates the operator norm of
$A_{\eta,0}$ is bounded. Hence
$
\|G_{\eta,0,k}-\Gamma_0A_{\eta,0,k}\|_{\eta,n}^*
\leq
C\{a_{\Gamma,J}+\delta_{B,n}\}
$
uniformly over the fixed number of blocks. Thus Assumption \ref{ass:riesz}
holds with
$
\delta_{R,n}
\lesssim
a_{\Gamma,J}+\delta_{B,n}.
$
The argument applies to the complete stacked Jacobian
$A_{\eta,0}^{(J)}$. In particular, derivatives of the $h$, $\ell$, and $g$
blocks with respect to $m$ are part of this matrix and therefore carry the
spatially aggregated directions induced by $C_i(m)$.
\hfill$\square$

\subsection{Proof of Proposition \ref{prop:gamma_rate}}
\label{app:proof_gamma_rate}

\noindent\textbf{Proof idea.}
The debiasing-operator error has four components: sieve approximation,
estimation of the target sensitivity, estimation of the nuisance Jacobian,
and Tikhonov regularization bias. Small singular values amplify the latter
three components. The source condition improves the population
regularization bias from the raw pseudoinverse order
$\lambda_{\Gamma,n}/\kappa_{\Gamma,n}^3$ to the target-relevant order
$\lambda_{\Gamma,n}/\kappa_{\Gamma,n}^2$.

Write
$
A_0=A_{\eta,0}^{(J_{\Gamma,n})},
\;
G_0=G_{\eta,0}^{(J_{\Gamma,n})},
\;
\kappa=\kappa_{\Gamma,n}.
$
For a fixed fold suppress the superscript $(-k)$ and write
$\widehat A=\widehat A_\eta$, $\widehat G=\widehat G_\eta$. Since
$
\|\widehat A-A_0\|_{\mathrm{op}}
=O_p(\Delta_{A,n})=o_p(\kappa),
$
Weyl's inequality implies
$
\sigma_{\min}(\widehat A)\geq\kappa/2
$
with probability approaching one. Because $A_0$ has full column rank,
$
\|A_0^\dagger\|_{\mathrm{op}}=\kappa^{-1}.
$
Standard perturbation bounds for full-column-rank pseudoinverses give
$
\|\widehat A^\dagger-A_0^\dagger\|_{\mathrm{op}}
=
O_p\left(
\frac{\Delta_{A,n}}{\kappa^2}
\right).
$

Define the regularized inverse
$
A_\lambda^\dagger
=
(A^{\prime}A+\lambda_{\Gamma,n}I)^{-1}A^{\prime}.
$
For a singular value $\sigma$, the scalar difference between
$\sigma/(\sigma^2+\lambda)$ and $1/\sigma$ is
$
\frac{\lambda}{\sigma(\sigma^2+\lambda)}.
$
Thus, on the event $\sigma_{\min}(\widehat A)\geq\kappa/2$,
$
\|\widehat A_\lambda^\dagger-\widehat A^\dagger\|_{\mathrm{op}}
=
O_p\left(
\frac{\lambda_{\Gamma,n}}{\kappa^3}
\right).
$
Local perturbation of the regularized inverse and
$\lambda_{\Gamma,n}=o(\kappa^2)$ yield
$
\|\widehat A_\lambda^\dagger-A_{0,\lambda}^\dagger\|_{\mathrm{op}}
=
O_p\left(
\frac{\Delta_{A,n}}{\kappa^2}
\right).
$

The raw regularization gap contains $\kappa^{-3}$, but after premultiplication
by $G_0$ the source condition gives the sharper target-relevant bound. To see
this, use
$
\Gamma_{0,J_{\Gamma,n}}=G_0A_0^\dagger,
\;
\sup_n\|\Gamma_{0,J_{\Gamma,n}}\|_{\mathrm{op}}<\infty.
$
An SVD calculation gives
$
\|G_0(A_{0,\lambda}^\dagger-A_0^\dagger)\|_{\mathrm{op}}
\leq
C
\frac{\lambda_{\Gamma,n}}{\kappa^2}.
$
Combining this with the perturbation bound gives
$
\|G_0(\widehat A_\lambda^\dagger-A_0^\dagger)\|_{\mathrm{op}}
=
O_p\left(
\frac{\Delta_{A,n}}{\kappa^2}
+
\frac{\lambda_{\Gamma,n}}{\kappa^2}
\right).
$

Now decompose
$
\widehat\Gamma-\Gamma_{0,J_{\Gamma,n}}
={}
(\widehat G-G_0)A_0^\dagger
+
G_0(\widehat A_\lambda^\dagger-A_0^\dagger)
+
(\widehat G-G_0)
(\widehat A_\lambda^\dagger-A_0^\dagger).
$
The first term is
$
O_p\left(\frac{\Delta_{G,n}}{\kappa}\right).
$
The second is bounded as above. For the interaction term,
$
\|\widehat A_\lambda^\dagger-A_0^\dagger\|_{\mathrm{op}}
=
O_p\left(
\frac{\Delta_{A,n}}{\kappa^2}
+
\frac{\lambda_{\Gamma,n}}{\kappa^3}
\right),
$
so
\[
\|(\widehat G-G_0)
(\widehat A_\lambda^\dagger-A_0^\dagger)\|_{\mathrm{op}}
=
O_p\left(
\frac{\Delta_{G,n}\Delta_{A,n}}{\kappa^2}
+
\frac{\Delta_{G,n}\lambda_{\Gamma,n}}{\kappa^3}
\right).
\]
Because $\lambda_{\Gamma,n}=o(\kappa^2)$, the last term is
$o_p(\Delta_{G,n}/\kappa)$ and is absorbed by the leading
$\Delta_{G,n}/\kappa$ term. Therefore
\[
\|\widehat\Gamma-\Gamma_{0,J_{\Gamma,n}}\|_{\mathrm{op}}
=
O_p\left(
\frac{\Delta_{G,n}}{\kappa}
+
\frac{\Delta_{A,n}}{\kappa^2}
+
\frac{\Delta_{G,n}\Delta_{A,n}}{\kappa^2}
+
\frac{\lambda_{\Gamma,n}}{\kappa^2}
\right).
\]
Finally,
$
\|\Gamma_{0,J_{\Gamma,n}}-\Gamma_0\|_{\mathrm{op}}
\leq
a_{\Gamma,J_{\Gamma,n}},
$
so the triangle inequality yields the claimed rate. The stated sufficient
condition for $r_{\Gamma,n}=o(n^{-1/4})$ follows immediately; the product
term is then of smaller order under the displayed first-order restrictions.
\hfill$\square$

\subsection{Proof of Proposition \ref{prop:orthogonality}}
\label{app:proof_orthogonality}

\noindent\textbf{Proof idea.}
The Riesz correction reproduces the first derivative of the target moment by
a linear combination of derivatives of the nuisance moments. Subtracting it
therefore removes the first-order nuisance sensitivity. Perturbations of the
Riesz operator are exactly orthogonal because the nuisance moments have zero
population mean.

For block $k$ and a nuisance direction $\delta\eta$,
$
\frac1{|\mathcal I_k|}
\sum_{i\in\mathcal I_k}
D_\eta E[\psi_i(\theta_0,\eta_0,\Gamma_0)][\delta\eta]
=
G_{\eta,0,k}[\delta\eta]
-
\Gamma_0A_{\eta,0,k}[\delta\eta].
$
Taking the target-relevant operator norm and applying Assumption
\ref{ass:riesz} gives
$
\sup_{1\leq k\leq K}
\left\|
\frac1{|\mathcal I_k|}
\sum_{i\in\mathcal I_k}
D_\eta E[\psi_i(\theta_0,\eta_0,\Gamma_0)]
\right\|_{\eta,n}^*
\leq
\delta_{R,n}.
$
For a perturbation $\delta\Gamma$,
$
D_\Gamma E[\psi_i(\theta_0,\eta_0,\Gamma_0)][\delta\Gamma]
=
-\delta\Gamma E[s_i(\theta_0,\eta_0)].
$
Since $E[s_i(\theta_0,\eta_0)]=0$, this derivative is exactly zero.
\hfill$\square$

\subsection{Proof of Lemma \ref{lem:dependent_crossfit}}
\label{app:proof_dependent_crossfit}

\noindent\textbf{Proof idea.}
Expand the feasible localized score around the oracle nuisance values. The
centered derivative fluctuation is controlled by spatial short memory, its
nonzero mean by the learner-specific leakage coefficient, and the average
population derivative by approximate orthogonality plus localization. The
quadratic remainder is controlled by the fourth-moment nuisance rate.

For fold $k$, write
$
\Delta\zeta_k
=
\widehat\zeta^{(-k)}-\zeta_0.
$
Assumption \ref{ass:score_smoothness} gives
$
\psi_{i,b_n}(\theta_0,\widehat\zeta^{(-k)})
-
\psi_{i,b_n}(\theta_0,\zeta_0)
=
G_{i,b_n}[\Delta\zeta_k]+r_{i,k}.
$
Let
$
L_k
=
\frac1{|\mathcal I_k|}
\sum_{i\in\mathcal I_k}
G_{i,b_n}[\Delta\zeta_k].
$
By the exact finite-dimensional tangent representation,
$
G_{i,b_n}[\Delta\zeta_k]
=
D_{i,b_n}\Delta v_k,
\;
\Delta v_k=r_{\zeta,n}u_k.
$
Therefore
$
L_k
={}
\overline D_{k,b_n}\Delta v_k
+
\left\{
\frac1{|\mathcal I_k|}
\sum_{i\in\mathcal I_k}E[D_{i,b_n}]
\right\}
\Delta v_k.
$
Decompose the first term into a centered fluctuation and its mean. Assumption
\ref{ass:crossfit_short_memory} gives
$
\overline D_{k,b_n}\Delta v_k
-
E[\overline D_{k,b_n}\Delta v_k]
=
O_p\left(
r_{\zeta,n}
\sqrt{\frac{J_{\zeta,n}}{|\mathcal I_k|}}
\right),
$
while the definition of $\chi_n(s_n)$ gives
$
\left\|
E[\overline D_{k,b_n}\Delta v_k]
\right\|_2
\leq
r_{\zeta,n}\chi_n(s_n).
$

For the mean derivative, Proposition \ref{prop:orthogonality} controls the
full-score population derivative by $\delta_{R,n}$. Assumption
\ref{ass:score_localization} changes the derivative by at most
$\delta_n^{\mathrm{loc}}(a_n^W,b_n)$. Hence
\[
\left\|
\frac1{|\mathcal I_k|}
\sum_{i\in\mathcal I_k}E[G_{i,b_n}]
\right\|_{\mathcal H,n}^*
\leq
C\left\{
\delta_{R,n}
+
\delta_n^{\mathrm{loc}}(a_n^W,b_n)
\right\}.
\]
Assumption \ref{ass:rates} therefore implies
\[
\left\|
\left\{
\frac1{|\mathcal I_k|}
\sum_{i\in\mathcal I_k}E[D_{i,b_n}]
\right\}
\Delta v_k
\right\|_2
=
O_p\left[
r_{\zeta,n}
\left\{
\delta_{R,n}
+
\delta_n^{\mathrm{loc}}(a_n^W,b_n)
\right\}
\right].
\]
Combining the three first-order components gives
\begin{align*}
L_k
=O_p\Bigg(&
r_{\zeta,n}
\sqrt{\frac{J_{\zeta,n}}{|\mathcal I_k|}}
+
r_{\zeta,n}\chi_n(s_n)+
r_{\zeta,n}\delta_n^{\mathrm{loc}}(a_n^W,b_n)
+
r_{\zeta,n}\delta_{R,n}
\Bigg).
\end{align*}

For the quadratic remainder, Assumptions \ref{ass:score_smoothness} and
\ref{ass:rates} imply
$
E\|r_{i,k}\|_2^2
\leq
C r_{\zeta,n}^4.
$
Minkowski's inequality therefore gives
$
\left\|
\frac1{|\mathcal I_k|}
\sum_{i\in\mathcal I_k}r_{i,k}
\right\|_{L^2}
\leq
C r_{\zeta,n}^2,
$
so the average remainder is $O_p(r_{\zeta,n}^2)$.

Since $|\mathcal I_k|\asymp n$, Assumption \ref{ass:spatial_blocks} gives
$
r_{\zeta,n}\sqrt{J_{\zeta,n}}=o(1)
$
and
$
\sqrt n\,r_{\zeta,n}
\left\{
\chi_n(s_n)
+
\delta_n^{\mathrm{loc}}(a_n^W,b_n)
+
\delta_{R,n}
\right\}
=o(1).
$
Assumption \ref{ass:rates} gives
$\sqrt n\,r_{\zeta,n}^2=o(1)$. Hence every term above is
$o_p(n^{-1/2})$, proving the foldwise oracle reduction.

Pooling over the fixed number of folds preserves the order. Finally,
Assumption \ref{ass:score_localization} implies that replacing the pooled
localized oracle score by the full oracle score contributes
$o_p(n^{-1/2})$. Thus
\[
\frac1n\sum_{k=1}^K\sum_{i\in\mathcal I_k}
\psi_{i,b_n}(\theta_0,\widehat\zeta^{(-k)})
=
\frac1n\sum_{i=1}^n\psi_i(\theta_0,\zeta_0)
+
o_p(n^{-1/2}).
\]
\hfill$\square$

\subsection{Proof of Lemma \ref{lem:random_weight_ned}}
\label{app:proof_sar_ned_weights}

\noindent\textbf{Proof idea.}
Local support and smooth row normalization transfer the NED approximation of
the weight-generating characteristics to the random weights. An interpolation
argument then controls the product of the weight-approximation error with a
possibly dependent random field.

On the candidate support define
$
q_{ij}=g_0(r(Z_i,Z_j,D_{ij})),
\;
q_{ij}^{[c]}
=
g_0(r(Z_i^{[c]},Z_j^{[c]},D_{ij})).
$
The Lipschitz assumptions on $r$ and $g_0$ imply
$
|q_{ij}-q_{ij}^{[c]}|
\leq
C\left(
\|Z_i-Z_i^{[c]}\|_2
+
\|Z_j-Z_j^{[c]}\|_2
\right).
$
Within a row, the normalized exponential map is a softmax map and
$
\frac{\partial w_{ij}}{\partial q_{i\ell}}
=
w_{ij}
\left(
\mathbbm 1\{j=\ell\}-w_{i\ell}
\right).
$
Its Jacobian is therefore uniformly Lipschitz in the rowwise $\ell_1$ norm.
Because each row contains at most $\bar d$ supported neighbors,
$
\sum_{j=1}^n|w_{ij}-w_{ij}^{[c]}|
\leq
C\sum_{j:S_{ij,n}=1}|q_{ij}-q_{ij}^{[c]}|.
$
The $L^{p+\delta_p}$-NED approximation of $Z$ and bounded degree yield
$
\sup_i
\left\|
\sum_{j=1}^n|w_{ij}-w_{ij}^{[c]}|
\right\|_{L^{p+\delta_p}}
\leq
C\delta_0(c).
$
Since supported neighbors lie within a fixed distance of unit $i$,
$w_{ij}^{[c]}$ is measurable with respect to innovations in a
$c+O(1)$ neighborhood of $i$. Hence each supported weight is $L^p$-NED.

For the operator action define
$
B_{i,c}
=
\sum_{j=1}^n|w_{ij}-w_{ij}^{[c]}|,
\;
M_i
=
\max_{j:S_{ij,n}=1}|A_j|.
$
Because both rows are probability vectors,
$0\leq B_{i,c}\leq2$, and bounded degree plus the maintained moment bound
gives
$
\sup_i\|M_i\|_{L^{p+\delta_p}}<\infty.
$
Moreover,
$
\left|\{(W_0-W_0^{[c]})A\}_i\right|
\leq
B_{i,c}M_i.
$
Set
$
r_p
=
\frac{p(p+\delta_p)}{\delta_p},
\;
\frac1p
=
\frac1{r_p}
+
\frac1{p+\delta_p}.
$
If $0<\delta_p\leq p$, interpolation between
$\|B_{i,c}\|_{L^{p+\delta_p}}\leq C\delta_0(c)$ and
$\|B_{i,c}\|_{L^\infty}\leq2$ gives
$
\|B_{i,c}\|_{L^{r_p}}
\leq
C\delta_0(c)^{\delta_p/p}.
$
If $\delta_p>p$, then $r_p<p+\delta_p$ and norm monotonicity gives
$
\|B_{i,c}\|_{L^{r_p}}
\leq
C\delta_0(c).
$
Thus in both cases
$
\|B_{i,c}\|_{L^{r_p}}
\leq
C\delta_0(c)^{\chi_p}
=
C\widetilde\delta_0(c).
$
H\"older's inequality now gives
$
\sup_i
\left\|\{(W_0-W_0^{[c]})A\}_i\right\|_{L^p}
\leq
C\widetilde\delta_0(c),
$
which proves the lemma.
\hfill$\square$

\subsection{Proof of Proposition \ref{prop:sar_ned}}
\label{app:proof_sar_ned}

\noindent\textbf{Proof idea.}
Truncate the SAR resolvent after a finite number of spatial propagation steps
and then replace the primitive variables and random weights by local
innovation approximations. Bounded candidate degree controls random
$L^p$ propagation, while $q_p<1$ makes the resolvent tail geometrically
summable.

For a random field $A=(A_1,\ldots,A_n)^{\prime}$ define
$
\|A\|_{p,\infty}
=
\sup_i\|A_i\|_{L^p}.
$
Since every row has at most $\bar d$ supported entries and
$0\leq w_{ij,0}\leq1$, Jensen's inequality gives
\begin{align*}
E|(W_0A)_i|^p
&\leq
E\sum_j w_{ij,0}|A_j|^p
\\
&\leq
\sum_{j:S_{ij,n}=1}E|A_j|^p
\leq
\bar d\|A\|_{p,\infty}^p.
\end{align*}
Hence
$
\|W_0A\|_{p,\infty}
\leq
\bar d^{1/p}\|A\|_{p,\infty},
$
and iterating gives
$
\|W_0^\ell A\|_{p,\infty}
\leq
\bar d^{\ell/p}\|A\|_{p,\infty}.
$
The same bound holds for the localized operator $W_0^{[c]}$ because it has
the same candidate support and is row stochastic.

The reduced form is
$
Y
=
\sum_{\ell=0}^{\infty}\rho_0^\ell W_0^\ell v.
$
Let
$
L=L(r)
=
\left\lfloor\frac{r}{4\bar a_W}\right\rfloor,
\;
Y^{(L)}
=
\sum_{\ell=0}^{L}\rho_0^\ell W_0^\ell v.
$
Using $q_p=|\rho_0|\bar d^{1/p}<1$,
\[
\sup_i\|Y_i-Y_i^{(L)}\|_{L^p}
\leq
C\sum_{\ell=L+1}^{\infty}q_p^\ell
\leq
Cq_p^{L+1}.
\]

Set $c=r/4$. Replace the primitive variables entering $v$ and the
weight-generating characteristics entering $W_0$ by their $c$-innovation
approximations, obtaining $v^{[c]}$ and $W_0^{[c]}$. Define
$
Y_i^{[r]}
=
\left\{
\sum_{\ell=0}^{L}
\rho_0^\ell(W_0^{[c]})^\ell v^{[c]}
\right\}_i.
$
Each application of the spatial operator moves by at most $\bar a_W$.
Therefore every path of length at most $L$ remains within distance
$L\bar a_W\leq r/4$ of $i$. The $c=r/4$ innovation approximation enlarges
the underlying innovation neighborhood by at most another $r/4$, up to the
fixed primitive locality radius. Hence $Y_i^{[r]}$ is measurable with respect
to innovations within distance $r+O(1)$ of unit $i$.

For $\ell\geq1$, the telescoping identity gives
\[
W_0^\ell v-(W_0^{[c]})^\ell v^{[c]}
=
W_0^\ell(v-v^{[c]})
+
\sum_{s=0}^{\ell-1}
W_0^s(W_0-W_0^{[c]})
(W_0^{[c]})^{\ell-1-s}v^{[c]}.
\]
The first term is bounded by
$
C\bar d^{\ell/p}\delta_0(c).
$
For each operator-difference term, Lemma
\ref{lem:random_weight_ned}, the maintained higher-moment propagation bound,
and the outer $W_0^s$ propagation bound give
\[
\left\|
W_0^s(W_0-W_0^{[c]})
(W_0^{[c]})^{\ell-1-s}v^{[c]}
\right\|_{p,\infty}
\leq
C\bar d^{\ell/p}\widetilde\delta_0(c).
\]
Consequently,
$
\|W_0^\ell v-(W_0^{[c]})^\ell v^{[c]}\|_{p,\infty}
\leq
C(1+\ell)\bar d^{\ell/p}\widetilde\delta_0(c).
$
Summing over $\ell\leq L$ and using $q_p<1$ yields
\[
\sup_i\|Y_i^{(L)}-Y_i^{[r]}\|_{L^p}
\leq
C\widetilde\delta_0(c).
\]
Combining this with the resolvent tail and setting $c=r/4$ gives
\[
\sup_i\|Y_i-Y_i^{[r]}\|_{L^p}
\leq
C\left\{
\widetilde\delta_0(r/4)
+
q_p^{L(r)+1}
\right\}.
\]
Thus $Y$ is $L^p$-NED.

The same finite-propagation and local-approximation argument applies to every
fixed-order transform $W_0^\ell X$ and $W_0^\ell Q$. Smooth finite-sieve score
derivatives are finite compositions and products of these NED objects. Under
the maintained envelope, moment, and sieve-complexity conditions, they
inherit the corresponding NED property up to the stated finite-sieve
complexity factor.
\hfill$\square$

\subsection{Proof of Proposition \ref{prop:primitive_short_memory}}
\label{app:proof_derivative_short_memory}

\noindent\textbf{Proof idea.}
The centered score-derivative matrix is a spatial sample average. The assumed
fourth-order spatial moment inequality gives the usual
$|\mathcal I_k|^{-1/2}$ scale per target-sensitive coordinate. H\"older's
inequality then controls multiplication by the normalized nuisance-training
direction.

Recall
$
\overline D_{k,b_n}
=
\frac1{|\mathcal I_k|}
\sum_{i\in\mathcal I_k}
\{D_{i,b_n}-E[D_{i,b_n}]\}.
$
For every target-sensitive coordinate $(j,\ell)$, the assumed fourth-moment
bound gives
$
E|\overline D_{k,b_n,j\ell}|^4
\leq
\frac{C}{|\mathcal I_k|^2}.
$
The target dimension $q$ is fixed. Summing over the
$J_{\zeta,n}$ target-sensitive coordinates and using
$
\left(\sum_m x_m^2\right)^2
\leq
J_{\zeta,n}\sum_mx_m^4
$
gives
$
E\|\overline D_{k,b_n}\|_F^4
\leq
C\frac{J_{\zeta,n}^2}{|\mathcal I_k|^2}.
$
Hence
$
\left(E\|\overline D_{k,b_n}\|_F^4\right)^{1/4}
\leq
C\sqrt{\frac{J_{\zeta,n}}{|\mathcal I_k|}}.
$
Since $E\|u_k\|_2^4\leq C$, H\"older's inequality implies
\[
\left(
E\|\overline D_{k,b_n}u_k\|_2^2
\right)^{1/2}
\leq
\left(E\|\overline D_{k,b_n}\|_F^4\right)^{1/4}
\left(E\|u_k\|_2^4\right)^{1/4}
\leq
C\sqrt{\frac{J_{\zeta,n}}{|\mathcal I_k|}}.
\]
The same bound applies to the norm of the expectation. Therefore
$
\left\|
\overline D_{k,b_n}u_k
-
E[\overline D_{k,b_n}u_k]
\right\|_2
=
O_p\left(
\sqrt{\frac{J_{\zeta,n}}{|\mathcal I_k|}}
\right),
$
which is Assumption \ref{ass:crossfit_short_memory}.
\hfill$\square$

\subsection{Proof of Proposition \ref{prop:ned_leakage}}
\label{app:proof_ned_leakage}

\noindent\textbf{Proof idea.}
Couple both the evaluation derivative and the fold-specific nuisance learner
to innovation-local approximations. Within the separated local basin the
coupled learner is stable by Assumption \ref{ass:learner_coupling}. Once both
objects are localized, their innovation sigma-fields are separated by at
least $s_n-2c_n$, so a spatial mixing covariance inequality controls the
remaining mean dependence.

For fold $k$, let
$
u_k^{[c_n]}
=
\frac{\Delta v_k^{[c_n]}}{r_{\zeta,n}}
$
be the normalized local-basin coupled learner direction. Assumption
\ref{ass:learner_coupling} gives
\[
\left(E\|u_k-u_k^{[c_n]}\|_2^2\right)^{1/2}
\leq
C\delta_{\mathrm{tr},n}^{\mathrm{NED}}(c_n)
\]
and a uniform fourth-moment bound for $u_k^{[c_n]}$. Replace
$D_{i,b_n}$ by $D_{i,b_n}^{[c_n]}$. By assumption,
\[
\max_i
\|D_{i,b_n}-D_{i,b_n}^{[c_n]}\|_{L^{2+\delta_\alpha},F}
\leq
\delta_{D,n}^{\mathrm{NED}}(c_n).
\]

The localized evaluation derivative depends on innovations in the
$c_n$-enlargement of the evaluation-score footprint, while the coupled local
learner depends on innovations in the $c_n$-enlargement of the training
footprint. The raw footprints are separated by $s_n$, so if $s_n>2c_n$ the
two innovation sets are separated by at least $s_n-2c_n$.

Consider one output coordinate and one target-sensitive sieve coordinate.
Davydov's covariance inequality with moment exponents $2+\delta_\alpha$ and
$2$ gives
\[
\left|
\operatorname{Cov}
\left(
D_{i,b_n,j\ell}^{[c_n]},
u_{k,\ell}^{[c_n]}
\right)
\right|
\leq
C
\left[
\alpha_{v_{\psi,n},v_{T,n},n}^{\mathcal E}(s_n-2c_n)
\right]^{\nu_\alpha},
\]
where
$
\nu_\alpha
=
1-
\frac1{2+\delta_\alpha}
-
\frac12
=
\frac{\delta_\alpha}{2(2+\delta_\alpha)}.
$
Assumption \ref{ass:innovation_mixing} therefore yields
\[
\left|
\operatorname{Cov}
\left(
D_{i,b_n,j\ell}^{[c_n]},
u_{k,\ell}^{[c_n]}
\right)
\right|
\leq
C
(v_{\psi,n}v_{T,n})^{\zeta_\alpha\nu_\alpha}
\bar\alpha_{\mathcal E}(s_n-2c_n)^{\nu_\alpha}.
\]
Averaging over $i$ does not increase this bound. Summing over the
$J_{\zeta,n}$ target-sensitive coordinates and applying Cauchy--Schwarz costs
at most $\sqrt{J_{\zeta,n}}$. Hence the fully localized contribution to
$\chi_n(s_n)$ is bounded by
\[
C\sqrt{J_{\zeta,n}}
(v_{\psi,n}v_{T,n})^{\zeta_\alpha\nu_\alpha}
\bar\alpha_{\mathcal E}(s_n-2c_n)^{\nu_\alpha}.
\]

Replacing $D_{i,b_n}^{[c_n]}$ by $D_{i,b_n}$ adds at most
$
C\sqrt{J_{\zeta,n}}
\delta_{D,n}^{\mathrm{NED}}(c_n),
$
while replacing $u_k^{[c_n]}$ by $u_k$ adds at most
$
C\delta_{\mathrm{tr},n}^{\mathrm{NED}}(c_n).
$
Therefore
\begin{align*}
\chi_n(s_n)
\lesssim{}
\sqrt{J_{\zeta,n}}
(v_{\psi,n}v_{T,n})^{\zeta_\alpha\nu_\alpha}
\bar\alpha_{\mathcal E}(s_n-2c_n)^{\nu_\alpha}
+
\sqrt{J_{\zeta,n}}
\delta_{D,n}^{\mathrm{NED}}(c_n)
+
\delta_{\mathrm{tr},n}^{\mathrm{NED}}(c_n),
\end{align*}
which proves the proposition. The argument uses only stability of the local
identified solution; it does not require continuity of a global argmin map
across separated nonconvex basins.
\hfill$\square$

\subsection{Proof of Corollary \ref{cor:primitive_guard}}
\label{app:proof_primitive_guard}

\noindent\textbf{Proof idea.}
Under exponential decay, logarithmic localization and guard radii transform
every dependence error into a power of $n$. The stated inequalities make the
exponents of all root-$n$ leakage terms strictly negative.

Suppose
$
r_{\zeta,n}=O(n^{-a}),
\;
J_{\zeta,n}=O(n^{\omega_J}),
\;
v_{\psi,n}=O(n^{\omega_\psi}),
\;
v_{T,n}=O(n).
$
Set
$
c_n=c_c\log n,
\;
s_n=c_s\log n.
$
For the derivative-approximation component,
$
\sqrt n\,r_{\zeta,n}\sqrt{J_{\zeta,n}}
\delta_{D,n}^{\mathrm{NED}}(c_n)
=
O\left(
n^{1/2-a+\omega_J/2-c_Dc_c}
\right),
$
which converges to zero if
$
c_c
>
\frac{1/2-a+\omega_J/2}{c_D}.
$
For learner coupling,
$
\sqrt n\,r_{\zeta,n}
\delta_{\mathrm{tr},n}^{\mathrm{NED}}(c_n)
=
O\left(
n^{1/2-a+\omega_{\mathrm{tr}}-c_{\mathrm{tr}}c_c}
\right),
$
which vanishes under the second displayed restriction in the corollary.

For the mixing component,
\begin{align*}
\sqrt n\,r_{\zeta,n}\sqrt{J_{\zeta,n}}
(v_{\psi,n}v_{T,n})^{\zeta_\alpha\nu_\alpha}
\bar\alpha_{\mathcal E}(s_n-2c_n)^{\nu_\alpha}
=
O\left(
n^{
1/2-a+\omega_J/2
+\zeta_\alpha\nu_\alpha(1+\omega_\psi)
-c_\alpha\nu_\alpha(c_s-2c_c)
}
\right).
\end{align*}
The third restriction makes this exponent negative. Finally,
$
r_{\zeta,n}\sqrt{J_{\zeta,n}}
=
O(n^{-a+\omega_J/2})
\to0
$
whenever $a>\omega_J/2$. Together with the separately imposed localization
and Riesz-approximation rate, these inequalities imply Assumption
\ref{ass:spatial_blocks}.
\hfill$\square$

\subsection{Proof of Corollary \ref{cor:primitive_benchmark}}
\label{app:proof_primitive_benchmark}

\noindent\textbf{Proof idea.}
The corollary collects the primitive pieces already established. The revised
joint-rate result closes the nuisance-learning chain, the stronger
$\mathcal G,+$ rate controls nonlinear operator remainders, and the NED
results supply the oracle limit theory and dependent cross-fit bounds.

By Lemma \ref{lem:random_weight_ned}, local innovation approximations of the
weight-generating characteristics induce local approximations of the random
interaction operator and its action on sufficiently integrable random
fields. Proposition \ref{prop:sar_ned} then shows that the SAR outcome,
fixed-order spatial transforms, and smooth finite-sieve score derivatives
are NED on the primitive innovation field. Under increasing-domain
regularity, exponential innovation mixing, the maintained moment bounds, and
Assumption \ref{ass:ned_summability}, the standard spatial NED LLN and CLT
therefore yield the oracle laws in Assumption \ref{ass:oracle_clt}.

Proposition \ref{prop:joint_preliminary} gives a rate for the full feasible
joint sieve start, including the target and all nuisance components held fixed
in the subsequent interaction profile. Proposition \ref{prop:sieve_g_rate}
then yields both the target-relevant $g$ rate and the stronger local envelope
rate without circularity. The maintained smoothness and dimensionality
conditions for $m$, $h$, and $\ell$, including the aggregate
$\sqrt{d_{B_g,n}}$ factor for the vector-valued projection, combine with
Proposition \ref{prop:gamma_rate} to give Assumption \ref{ass:rates} and the
required fourth-moment controls.

Proposition \ref{prop:riesz_existence} gives the finite-sieve Riesz
representation and the blockwise approximation error, while Proposition
\ref{prop:gamma_rate} controls estimation of the regularized representer.
The target-rank and interaction-richness assumptions maintain identification
of the relevant finite-dimensional and operator directions.

Proposition \ref{prop:primitive_short_memory} establishes Assumption
\ref{ass:crossfit_short_memory}. The local-basin coupling condition in
Assumption \ref{ass:learner_coupling}, Proposition \ref{prop:ned_leakage},
and Corollary \ref{cor:primitive_guard} make learner-specific
training-to-evaluation leakage asymptotically negligible with logarithmic
innovation-localization and guard radii. The condition
$\sqrt n\,r_{\zeta,n}^2\to0$ controls the quadratic score remainder, while
the displayed localization and Riesz rates control the remaining first-order
terms.

All conditions of Lemma \ref{lem:dependent_crossfit} therefore hold, so the
feasible cross-fitted orthogonal score admits the oracle reduction. This
establishes the claimed compatibility of random $W_0$, globally simultaneous
$Y$, smooth sieve learning, and buffered spatial cross-fitting.
\hfill$\square$

\subsection{Proof of Lemma \ref{lem:uniform_gmm_reduction}}
\label{app:proof_uniform_gmm}

\noindent\textbf{Proof idea.}
The dependent cross-fit lemma gives oracle equivalence at the true target
parameter. Conditional on the nuisance functions, the target and nuisance
moments are affine in $\theta$, so convergence of the target derivative
extends the equivalence uniformly over the compact parameter space. The
explicit oracle ULLN then transfers the feasible GMM criterion to its
population counterpart.

Because $\xi_i(\theta,\eta)$ is affine in $\theta$ and the nuisance moments
$s_{h,i}$ and $s_{g,i}$ depend on $\theta$ only through this residual,
$\psi_i(\theta,\eta,\Gamma)$ is affine in $\theta$ once
$(\eta,\Gamma)$ are fixed. Hence
\[
\overline\psi_n(\theta)-\overline\psi_n^0(\theta)
=
\{\overline\psi_n(\theta_0)-\overline\psi_n^0(\theta_0)\}
+
(\widehat J_n-J_n^0)(\theta-\theta_0),
\]
where $\widehat J_n$ and $J_n^0$ are the feasible and oracle sample target
Jacobians. Lemma \ref{lem:dependent_crossfit} gives
$
\|\overline\psi_n(\theta_0)-\overline\psi_n^0(\theta_0)\|_2
=o_p(n^{-1/2}),
$
and Assumption \ref{ass:score_smoothness}, together with the oracle
derivative ULLN in Assumption \ref{ass:oracle_clt}, gives
$
\|\widehat J_n-J_n^0\|_{\mathrm{op}}=o_p(1).
$
Compactness of $\Theta$ therefore implies
$
\sup_{\theta\in\Theta}
\|\overline\psi_n(\theta)-\overline\psi_n^0(\theta)\|_2
=o_p(1).
$
The first ULLN in Assumption \ref{ass:oracle_clt} gives directly
$
\sup_{\theta\in\Theta}
\|\overline\psi_n^0(\theta)-\mu_n(\theta)\|_2
=o_p(1).
$
Since $\widehat{\mathcal M}\xrightarrow{p}\mathcal M_0$ and the moments are
uniformly bounded in probability on compact $\Theta$, the inequality
$
|a^{\prime}Ma-b^{\prime}Nb|
\leq
C\|a-b\|(\|a\|+\|b\|)
+
\|M-N\|_{\mathrm{op}}\|b\|^2
$
implies
\[
\sup_{\theta\in\Theta}
\left|
\overline\psi_n(\theta)^{\prime}\widehat{\mathcal M}
\overline\psi_n(\theta)
-
Q_n(\theta)
\right|
=o_p(1).
\]
\hfill$\square$

\subsection{Proof of Theorem \ref{thm:asymptotic_normality}}
\label{app:proof_asymptotic_normality}

\noindent\textbf{Proof idea.}
Uniform convergence of the feasible GMM criterion gives consistency. The
oracle-reduction lemma then replaces the feasible score at the truth by the
oracle score up to $o_p(n^{-1/2})$. A standard GMM linearization and the
oracle spatial CLT yield asymptotic normality.

Lemma \ref{lem:uniform_gmm_reduction} and Assumption
\ref{ass:global_identification} imply
$
\widehat\theta\xrightarrow{p}\theta_0.
$
Since $\theta_0$ lies in the interior of $\Theta$, the GMM first-order
condition holds with probability approaching one, up to the assumed
asymptotically negligible numerical optimization error:
$
\widehat J_\psi(\widehat\theta)^{\prime}
\widehat{\mathcal M}
\overline\psi_n(\widehat\theta)
=o_p(n^{-1/2}).
$
A mean-value expansion gives
$
\overline\psi_n(\widehat\theta)
=
\overline\psi_n(\theta_0)
+
\widetilde J_\psi
(\widehat\theta-\theta_0),
$
where $\widetilde J_\psi$ is evaluated between $\theta_0$ and
$\widehat\theta$. The score-smoothness condition and the oracle derivative
ULLN imply
$
\widetilde J_\psi
\xrightarrow{p}
J_{\psi,0}.
$
Assumption \ref{ass:rank} implies that
$J_{\psi,0}^{\prime}\mathcal M_0J_{\psi,0}$ is nonsingular. Hence
\begin{align*}
\sqrt n(\widehat\theta-\theta_0)
={}&
-
\left(
\widetilde J_\psi^{\prime}
\widehat{\mathcal M}
\widetilde J_\psi
\right)^{-1}
\widetilde J_\psi^{\prime}
\widehat{\mathcal M}
\sqrt n\,\overline\psi_n(\theta_0)
+
o_p(1).
\end{align*}
Lemma \ref{lem:dependent_crossfit} gives
$
\sqrt n\,\overline\psi_n(\theta_0)
=
\frac1{\sqrt n}\sum_{i=1}^n
\psi_i(\theta_0,\eta_0,\Gamma_0)
+
o_p(1).
$
Furthermore,
$
\left(
\widetilde J_\psi^{\prime}
\widehat{\mathcal M}
\widetilde J_\psi
\right)^{-1}
\widetilde J_\psi^{\prime}
\widehat{\mathcal M}
\xrightarrow{p}
B_0.
$
Therefore
\[
\sqrt n(\widehat\theta-\theta_0)
=
-B_0
\frac1{\sqrt n}\sum_{i=1}^n
\psi_i(\theta_0,\eta_0,\Gamma_0)
+
o_p(1).
\]
The oracle spatial CLT in Assumption \ref{ass:oracle_clt} and Slutsky's
theorem yield
$
\sqrt n(\widehat\theta-\theta_0)
\xrightarrow{d}
N(0,B_0\Omega_0B_0^{\prime}).
$
If $\mathcal M_0=\Omega_0^{-1}$, straightforward matrix algebra gives
$
B_0\Omega_0B_0^{\prime}
=
\left(
J_{\psi,0}^{\prime}\Omega_0^{-1}J_{\psi,0}
\right)^{-1}.
$
\hfill$\square$

\subsection{Proof of Corollary \ref{cor:slow_rates}}
\label{app:proof_slow_rates}

\noindent\textbf{Proof idea.}
Approximate Neyman orthogonality removes the ordinary first-order nuisance
term. The remaining feasible-to-oracle error consists of second-order
nuisance products and first-order terms specific to spatial dependence,
localization, finite-sieve derivative complexity, and Riesz approximation.

The expansion in the proof of Lemma \ref{lem:dependent_crossfit} shows that
the first-order feasible-to-oracle difference is bounded by
\[
O_p\left[
r_{\zeta,n}\sqrt{\frac{J_{\zeta,n}}{n}}
+
r_{\zeta,n}\chi_n(s_n)
+
r_{\zeta,n}\delta_n^{\mathrm{loc}}(a_n^W,b_n)
+
r_{\zeta,n}\delta_{R,n}
\right].
\]
At root-$n$ scale, the centered derivative term is negligible if
$
r_{\zeta,n}\sqrt{J_{\zeta,n}}=o(1),
$
and the remaining first-order spatial terms are negligible if
$
\sqrt n\,r_{\zeta,n}
\left[
\chi_n(s_n)
+
\delta_n^{\mathrm{loc}}(a_n^W,b_n)
+
\delta_{R,n}
\right]
=o(1).
$
Approximate orthogonality removes the other first-order nuisance effects.
The nonlinear remainder is a sum of products $r_{a,n}r_{b,n}$, so
$
\sqrt n\,r_{a,n}r_{b,n}=o(1)
$
for every product appearing in the expansion is sufficient. Under a common
fixed-complexity nuisance rate $r_n$, this reduces to
$
\sqrt n\,r_n^2=o(1),
$
or equivalently $r_n=o(n^{-1/4})$, together with the spatial terms above.
\hfill$\square$

\subsection{Proof of Proposition \ref{prop:variance}}
\label{app:proof_variance}

\noindent\textbf{Proof idea.}
The feasible score is close to the oracle score in mean square. The spatial
kernel row-sum bound converts this score-replacement error into a bound on
the difference between feasible and oracle HAC matrices. Oracle HAC
consistency then transfers to the feasible estimator, and projection onto the
positive-semidefinite cone preserves the probability limit.

Let
$
e_i
=
\widehat\psi_i-
\psi_i^0.
$
A mean-value expansion in $\theta$, Theorem
\ref{thm:asymptotic_normality}, and the score-replacement condition imply
\[
\left[
\frac1n\sum_{i=1}^nE\|e_i\|_2^2
\right]^{1/2}
=
O_p\left(
\delta_{\psi,n}+n^{-1/2}
\right).
\]
Define
$
\delta_{\psi,n}^{\mathrm{tot}}
=
O_p\left(
\delta_{\psi,n}+n^{-1/2}
\right).
$
After centering,
$
\widetilde\psi_i
=
\widetilde\psi_i^0+
\widetilde e_i,
$
where
$
\left[
\frac1n\sum_{i=1}^n
\|\widetilde e_i\|_2^2
\right]^{1/2}
=
O_p(\delta_{\psi,n}^{\mathrm{tot}}).
$
Hence
\[
\widetilde\psi_i\widetilde\psi_j^{\prime}
-
\widetilde\psi_i^0\widetilde\psi_j^{0\prime}
=
\widetilde e_i\widetilde\psi_j^{0\prime}
+
\widetilde\psi_i^0\widetilde e_j^{\prime}
+
\widetilde e_i\widetilde e_j^{\prime}.
\]
Let
$
K_{ij,n}
=
\mathcal K\left(
\frac{d_{ij}^*}{\nu_n}
\right).
$
Using Cauchy--Schwarz twice and
$
\kappa_n(\nu_n)
=
\max_i\sum_j|K_{ij,n}|,
$
we obtain
\[
\|\widehat\Omega-\widehat\Omega^0\|_F
=
O_p\left[
\kappa_n(\nu_n)
\delta_{\psi,n}^{\mathrm{tot}}
\{1+\delta_{\psi,n}^{\mathrm{tot}}\}
\right].
\]
Assumption \ref{ass:hac} gives
$
\kappa_n(\nu_n)\delta_{\psi,n}=o_p(1)
$
and
$
\frac{\kappa_n(\nu_n)}{\sqrt n}\to0.
$
Therefore
$
\widehat\Omega-
\widehat\Omega^0=o_p(1).
$
Since
$
\widehat\Omega^0\xrightarrow{p}\Omega_0,
$
it follows that
$
\widehat\Omega\xrightarrow{p}\Omega_0.
$
Further,
$
\widehat J_\psi\xrightarrow{p}J_{\psi,0},
\;
\widehat{\mathcal M}\xrightarrow{p}\mathcal M_0,
$
so
$
\widehat B\xrightarrow{p}B_0.
$
Continuous mapping therefore yields
$
\widehat V
=
\widehat B\widehat\Omega\widehat B^{\prime}
\xrightarrow{p}
B_0\Omega_0B_0^{\prime}
=
V_0.
$

Finally, let
$
\widehat\Omega^{\mathrm{sym}}
=
\frac{\widehat\Omega+\widehat\Omega^{\prime}}{2}.
$
Because $\Omega_0$ is symmetric and positive definite,
$
\widehat\Omega^{\mathrm{sym}}
\xrightarrow{p}
\Omega_0.
$
Projection onto the closed convex cone $\mathbb S_+^q$ is nonexpansive in
the Frobenius norm, and $\Pi_{\mathbb S_+^q}(\Omega_0)=\Omega_0$. Hence
$
\|\widehat\Omega^+-\Omega_0\|_F
\leq
\|\widehat\Omega^{\mathrm{sym}}-\Omega_0\|_F
=
o_p(1).
$
Thus
$
\widehat\Omega^+\xrightarrow{p}\Omega_0,
\;
\widehat V^+
=
\widehat B\widehat\Omega^+\widehat B^{\prime}
\xrightarrow{p}
V_0.
$
\hfill$\square$

\end{document}